\documentclass[11pt]{article}
\usepackage{subcaption}
\usepackage{definitions}

\begin{document}

\title{The Best of Many Worlds: Dual Mirror Descent for\\ Online Allocation Problems}


\author{Santiago Balseiro\thanks{Columbia University Business School (\texttt{srb2155@columbia.edu})}
\and
Haihao Lu\thanks{University of Chicago Booth School of Business (\texttt{haihao.lu@chicagobooth.edu})}
\and
Vahab Mirrokni \thanks{Google Research (\texttt{mirrokni@google.com})}}

\date{This version: \today \\ Previous versions: November 18, 2020 and July 9, 2021 \\[2em]
 Forthcoming in \emph{Operations Research}
}
\maketitle



\begin{abstract}

Online allocation problems with resource constraints are central problems in revenue management and online advertising. In these problems, requests arrive sequentially during a finite horizon and, for each request, a decision maker needs to choose an action that consumes a certain amount of resources and generates reward. The objective is to maximize cumulative rewards subject to a constraint on the total consumption of resources. In this paper, we consider a data-driven setting in which the reward and resource consumption of each request are generated using an input model that is unknown to the decision maker. 
    
{We design a general class of algorithms that attain good performance in various input models without knowing which type of input they are facing. In particular, our algorithms are asymptotically optimal under independent and identically distributed inputs as well as various non-stationary stochastic input models, and they attain an asymptotically optimal fixed competitive ratio when the input is adversarial. Our algorithms operate in the Lagrangian dual space: they maintain a dual multiplier for each resource that is updated using online mirror descent. By choosing the reference function accordingly, we recover the dual sub-gradient descent and dual multiplicative weights update algorithm. The resulting algorithms are simple, fast, and do not require convexity in the revenue function, consumption function and action space, in contrast to existing methods for online allocation problems. We discuss applications to network revenue management, online bidding in repeated auctions with budget constraints, online proportional matching with high entropy, and personalized assortment optimization with limited inventory.}

\bigskip


\end{abstract}

    
    


\nvspace
\setstretch{1.5}
\section{Introduction}






A central problem in revenue management and online advertising is the online allocation of requests subject to resource constraints. In revenue management, for example, firms such as hotels and airlines need to decide, when a request for a room or a flight arrives, whether to accept or decline the request~\citep{TalluriVanRyzin2004}. In search advertising, each time a user makes a search, the search engine has an opportunity to show an advertisement next to the organic search results~\citep{Mehta2007JACM}. For each arriving user, the website collects bids from various advertisers who are interested in showing an ad and then needs to decide, in real time, which ad to show to the user. Such decisions are not made in isolation because resources are limited: hotels have a limited number of rooms, planes have a limited number of seats, and advertisers have limited budgets.

In this paper, we study allocation problems with non-linear reward functions, non-linear consumption functions, and potentially integral action space. Requests arrive sequentially during a finite horizon and, for each request, the decision maker needs to choose an action that consumes certain amount of resources and generates a reward. The objective of the decision maker is to maximize cumulative rewards subject to a constraint on the total consumption of resources. The reward  and resource consumption functions of each request are learnt by the decision maker before making a decision. For example, airlines know the fare requested by the consumer before deciding whether to sell the ticket and search engines know advertisers' bids before deciding which ad to show. The decision maker, however, does not get to observe the reward and consumption functions of future requests until their arrival. 
Thus motivated, we consider a data-driven setting in which the reward and resource consumption function of each request are generated from an input model that is unknown to the decision maker. 

The objective of this paper is to design algorithms that attain good performance relative to the best allocation with the benefit of hindsight (also referred as the offline optimum). In particular, our goal is to design \emph{fast} and \emph{robust} algorithms that attain \emph{asymptotic good performance} for \emph{generic} online allocation problems on various inputs model while being \emph{oblivious to the input model}, i.e., without knowing in advance which input model they are facing. {In doing so, our hope is to design algorithms that  attain the best of ``many'' worlds.}

\nnvspace
\subsection{Main Contributions}

We design a general class of algorithms that operate in the Lagrangian dual space and present a flexible framework to analyze their performance under various settings.

If the optimal dual variables were known in advance, the decision maker could, in principle, use these dual variables to price resources and decompose the problem across time periods. In practice, however, the optimal dual variables depend on the entire sequence of requests and are not known to the decision maker in advance. Our algorithms circumvent this issue by maintaining a dual multiplier for each resource, which is updated after each request using online mirror descent. Actions are then taken using the estimated dual variables as a proxy for the opportunity cost of consuming resources. By choosing the reference function in mirror descent accordingly, we recover dual sub-gradient descent and dual multiplicative weights update, which are two popular algorithms used in practice. {Furthermore, our analysis can be easily extended to study other online dual algorithms.}

From the computational perspective, our algorithms are efficient; in many cases the dual variables can be updated after each request in linear time. This is in sharp contrast to most existing algorithms, which require periodically solving large convex optimization problems or knowing bounds on the value benchmark~(see Section~\ref{sec:related} for a literature review). In many applications, such as online advertising, a massive number of decisions need to made in milliseconds and solving large optimization problems is not operationally feasible.

Our algorithms have minimal requirements on the reward functions, the consumption functions and the action space. While most of the previous works on this topic require linear or concave reward functions and linear consumption in decision variables, our analysis is flexible and works for non-concave reward functions, non-linear consumption functions, and integral action space (see Section~\ref{sec:assumption} for details).   Such flexibility allows us to handle more applications, such as assortment optimization, without introducing an exponential number of auxiliary variables. {\color{black} Computationally, our algorithms require that the single-period Lagrangian problem obtained by dualizing the constraints can be efficiently solved, which, as we discuss, is possible in many applications.} Furthermore, our algorithm is robust to noise and corruptions in the observations.

We study the performance of our algorithms on different input models as a function of the number of time periods $T$, when resources are scaled proportionally to the length of the horizon. Our algorithms obtain good performance under various input models without knowing which input model they are facing:
\begin{itemize}
    \item When the input is stochastic and requests are independently and identically drawn (i.i.d.) from a distribution that is unknown to the decision maker, our algorithms attain regret $O(T^{1/2})$, where regret is measured as the difference between the rewards attained by the optimal allocation with the benefit of hindsight and the cumulative rewards collected by the decision maker (Theorem~\ref{thm:master}). Because no algorithm can attain regret lower than $\Omega(T^{1/2})$ under our minimal assumption (Lemma~\ref{lem:lower_bound}), these two results imply that our algorithms attain the optimal order of regret.
    
    \item On the other extreme, when requests are adversarially chosen, no algorithm can attain vanishing regret, but our algorithms are shown to attain a fixed competitive ratio, i.e., they guarantee a fixed fraction of the optimal allocation in hindsight (Theorem~\ref{thm:adversial-no-regularizer}). Our competitive ratios are tight and no algorithm can obtain better competitive ratios without further assumptions on the input.
    
    \item We also consider non-stationary stochastic input models that fill the gap between i.i.d.~input, which is too optimistic, and adversarial input, which is too pessimistic. These inputs are more realistic in many applications. In particular, we show that our algorithms attain regret $\tilde O(T^{1/2})$ when the input is ergodic (Theorem~\ref{thm:ergodic}), $O(T^{1/2}$) regret when the input exhibits seasonality (Theorem~\ref{thm:periodic}), or vanishing regret when an adversary corrupts $o(T)$ number of requests (Theorem~\ref{thm:independent}). 
    
\end{itemize}  



A requisite for obtaining good primal performance is not depleting resources too early; otherwise, the decision maker could miss good future opportunities. Our algorithms have a natural self-correcting feature that prevents them from depleting resources too early. By design, they target to consume a constant number of resources per period to deplete resources exactly at the end of the horizon. When a request consumes more (less) resources than the target, the corresponding dual variable is increased (decreased). Because resources are then priced higher (lower), future actions are chosen to consume resources more conservatively (aggressively). As a result, using the update rule of the dual variables, we can show that our algorithms never deplete resources too early (Proposition~\ref{thm:stopping_time}). Interestingly, this result holds for every sample path regardless of the input model. To the best of our knowledge, this result is new to the online allocation literature and can be of interest for practitioners as, for example, advertisers have a preference for their ads to be delivered smoothly over time so as to maximize reach~\citep{Bhalgat2012smooth,Lee2013smooth,Xu2015smart}.

We then discuss applications to four central problems in revenue management and online advertising: online linear programming, bidding in repeated auctions with budgets, online matching with high entropy, and personalized assortment optimization problems with limited inventories. In all these applications our algorithms yield new results or match existing results with simpler, more efficient implementations. {\color{black} We conclude the paper by discussing several extensions of our work and presenting numerical experiments that validate our theoretical results.}

\nnvspace

\subsection{Related Work}\label{sec:related}

Online allocation problems have been extensively studied in computer science and operations research literature. We overview the related literature below.

\nnvspace
\paragraph{Stochastic inputs.} Early work on online allocation with stochastic input models focus on the case when the reward and resource consumption functions are linear in the decision variable, in particular in the so-called \emph{random permutation model}. In the random permutation model, an adversary first selects a sequence of requests which are then presented to the decision maker in random order. This model is more general than our stochastic i.i.d. setting in which requests are drawn independently and at random from an unknown distribution. \citet{DevanurHayes2009} study the AdWords problem and present a dual training algorithm with two phases: a training phase in which data is used to estimate the dual variables by solving a linear program and an exploitation phase in which actions are taken using the estimated dual variables. Their algorithm can be shown to obtain regret of order $O(T^{2/3})$. \citet{Feldman2010} present similar training-based algorithms for more general linear online allocation problems with similar regret guarantees. 

Pushing these ideas one step further, \citet{Agrawal2014OR}, \citet{Devanur2019near} and \citet{kesselheim2014primal} consider primal- and/or dual-based algorithms that dynamically update decisions by periodically solving a linear program using the data collected so far. These more sophisticated algorithms improve upon previous work by obtaining regret of order $O(T^{1/2})$ and better dependencies on the number of resources. 
{\citet{gupta2016experts} extend the results to the random permutation model.} \citet{AgrawalDevanue2015fast} study general online allocation problems and  present a similar algorithm that maintains and updates dual variables for the constraints. Their algorithm requires an estimate of the value of the benchmark or, when an estimate is not available, their algorithm requires solving an optimization program to estimate the value of the benchmark. Our paper extends this line of work by developing a class of algorithms for general online allocation problems with potentially \emph{non-convex} reward and resource consumption functions, which yield similar regret guarantees with simpler update rules that \emph{do not require} solving large linear programs nor estimates of the value of the benchmark. Moreover, our algorithms are shown to  attain good performance over various input models. 

While the algorithms described above usually require solving large linear problems periodically, there is a recent line of work seeking simple algorithms that have no need of solving a large linear program. {In a preliminary version of this paper~\citep{balseiro2020dual},  we studied a simple dual mirror descent algorithm for online allocation problems with concave reward functions and stochastic inputs, which attains $O(T^{1/2})$ regret. The algorithm updates dual variables in each period in linear time and avoids solving large auxiliary programs. The analysis in this paper is simpler as we do not need to explicitly bound the stopping time corresponding to the first time a resource is depleted.} In simultaneous and independent work, ~\cite{li2020simple} present a similar fast algorithm that attains $O(T^{1/2})$ regret for linear rewards, but our regret bound has better dependence on the number of resources. {\color{black} \cite{sun2020nearoptimal} give a fast algorithm that attains $o(T^{1/2})$ regret when the distribution of requests is known to the decision maker.} Our proposed algorithms fall into this category: the update per iteration can be efficiently computed in linear time and there is no need to solve large convex optimization problems. {Finally, we remark that the algorithmic principle of using dual mirror descent in an online fashion to tackle a complex control problem has been recently exploited by \citet{kanoria2020blind} for controlling resources that circulate in closed network with applications to ride-hailing.}

\nnvspace
\paragraph{Adversarial inputs.} There is a stream of literature that studies online allocation problems under adversarial input models, i.e., when the incoming requests are adversarially chosen.  In this case, it is generally impossible to attain sublinear regret and, instead, the focus is on designing algorithms that obtain constant factor approximations to the offline optimum solution. \citet{Mehta2007JACM} and \citet{buchbinder2007primaldual} study the \emph{AdWords} problem, an online matching problem in which rewards are proportional to resource consumption, and provide an algorithm that obtains a $(1 - 1/e)-$fraction of the optimal allocation in hindsight, which is optimal. In general, when rewards are not proportional to resource consumption, it is not possible to attain fixed competitive ratio. To circumvent this issue, \citet{Feldman2009} consider a variant with free disposal in which resource constraints can be violated and only the assigned requests with the highest rewards count toward the objective. They provide a primal-dual algorithm that attains a $(1-1/e)$-fraction of the optimal allocation. Both papers assume that the amount of resources is large relative to each individual request. In this paper, we consider general allocation problems without the free disposal assumption and provide parametric competitive ratios for adversarial inputs that depend on the scarcity of resources. Our bounds are also shown to be asymptotical optimal for the more general setting considered in this paper. 
{In Section~\ref{sec:applications}, we provide a more detailed comparison of our results with existing competitive algorithms for the different applications explored in this paper.}

\nnvspace
\paragraph{Other related work.} There has been a recent interest in devising algorithms that achieve the best of all words, i.e., they attain good performance in various inputs models. \cite{mirrokni2012simultaneous} study the AdWords problem and provide an algorithm that attains the optimal competitive ratio for adversarial input and improved competitive ratios (though not asymptotic optimality) for stochastic inputs. {Moreover, \cite{mirrokni2012simultaneous} showed that, in the AdWords problem, no algorithm with vanishing regret under stochastic input can attain a constant competitive ratio for adversarial input, where the constant is independent from the model data.  Our results do not contradict their findings because our competitive ratio is data dependent and relies on the consumption-budget ratio. Although the AdWords problem is a special case of our online allocation problem, it imposes special structure, namely a linear relationship between the reward and consumption, which is key to obtain a constant competitive ratio. The competitive ratio of our algorithms for adversarial input might not be tight for the AdWords problem, but it is optimal for the more general online allocation setting we study herein. Meanwhile, our algorithms are asymptotically optimal for stochastic input.}




Moving beyond i.i.d.~input, \cite{ciocan2012dynamic} study online allocation problems in which the demand process is volatile. They provide a primal algorithm that estimates the demand rates and computes an optimal allocation assuming that the demand for the remaining of the horizon will remain unchanged. The algorithm is shown to attain a constant factor approximation for arbitrarily correlated Gaussian processes. Our results are complementary: while their algorithm attains better competitive ratios for processes with drift, our algorithms are asymptotically optimal for ergodic processes. Furthermore, \cite{esfandiari2018allocation} present an online allocation algorithm with forecast that applies to mixed adversarial and stochastic settings. If the forecast is perfect their algorithm achieves sublinear regret, otherwise they obtain constant competitive ratios that degrade gracefully with the quality of the forecast. In contrast, our algorithms and analysis do not rely on any forecasting.

Our work is also related to the literature on multi-arm bandits with knapsacks. Our feedback structure is stronger because we get to observe the reward function and consumption matrix before making a decision, while, in the bandit literature, these are revealed after making a decision. While algorithms for bandits with knapsacks are not directly applicable, our problem can be thought of as a \emph{contextual} multi-arm bandit problem with knapsacks, where the context would correspond to the information of the request. The algorithms of \citet{badanidiyuru2014resourceful} and \citet{agrawal2016efficient} can be applied to our setting after discretizing the context and action space. Discretization, however, leads to sub-optimal performance guarantees. In particular, the cardinality of the support of the action space does not appear in our regret bounds, while it must appear in the bandit setting since bandit algorithms need to explore the rewards for different actions. 

\nvspace
\section{Problem Formulation and Algorithm}\label{sec:formulation}
We consider a generic online allocation problem with a finite horizon of $T$ time periods and resource constraints. At time $t$, the decision maker receives a request $\gamma_t = (f_t,b_t,\cX_t) \in \mathcal S$ where $f_t: \cX_t\rightarrow\RR_+$ is a non-negative (and potentially non-concave) reward function, $b_t:\cX_t\rightarrow\RR_+^m$ is an non-negative (and potentially non-linear) resource consumption function, and $\cX_t \subset \RR^d_+$ is a (potentially non-convex or integral) compact set. {\color{black} Here, $\RR_+$ denotes the non-negative real numbers.} We denote by $\mathcal S$ the set of all possible requests that can be received. After observing the request, the decision maker takes an action $x_t\in \cX_t\subseteq \RR^d$ that leads to reward $f_t(x_t)$ and consumes $b_t(x_t)$ resources. The total amount of resources is $T \rho$, where $\rho\in \RR^m_{+}$ is a resource constraint vector with $\rho_j > 0$ for all $j$. The assumption $b_t(\cdot)\ge 0$ implies that we cannot replenish resources once they are consumed. We assume that $0\in \cX_t$ and $b_t(0) = 0$ so that it is always possible to take a void action by choosing $x_t=0$ and avoid violating the resource constraints. This guarantees the existence of a feasible solution. While we assume throughout the paper that rewards and resource consumptions are deterministic given the action, we show in Section \ref{sec:stochastic-input} that our results apply when reward and resource consumption are stochastic given an action. {\color{black} In Section~\ref{sec:random-time} we discuss how to incorporate stochastic horizons.}

We utilize $\vgamma = (\gamma_1,\ldots,\gamma_T)$ to denote the inputs  over time $1,\ldots,T$. 
The baseline we compare with is the reward of the optimal solution when the request sequence $\vgamma$ is known in advance, which amounts to solving the optimal allocation under full information of all requests:
\begin{align}\label{eq:OPT}
\begin{split}
    \OPT(\vec \gamma)=
	\max_{x: x_t\in \cX_t} \  \sum_{t=1}^T f_t(x_t)\ \ \      	\text{s.t.}     \  \sum_{t=1}^T b_t(x_t) \le T\rho\,.
\end{split}
\end{align}
The latter problem is referred to as the offline optimum in the computer science or hindsight optimum in the operations research literature.



An online algorithm $A$ makes, at time $t$, a real-time decision $x_t$ (potentially randomized) based on the current request $(f_t, b_t, \cX_t)$ and the previous history $\cH_{t-1}:=\{f_s, b_s, \cX_s, x_s\}_{s=1}^{t-1}$, i.e., $
x_t =A(f_t, b_t, \cX_t|\cH_{t-1})$. We define the reward of an algorithm for input $\vgamma$ as 
\[
    R(A | \vgamma) = \sum_{t=1}^T f_t(x_t) \,,
\]
where $x_t$ is the decision of the algorithm at time $t$. Moreover, algorithm $A$ must satisfy constraints $\sum_{t=1}^{T} b_{t}(x_{t}) \le \rho T$ and $x_{t}\in \cX$ for every $t\le T$.  

Our goal is to design an algorithm $A$ that attains good performance, while satisfying the above constraints, for different input models. Additionally, the algorithm $A$ should be \emph{oblivious to the input model}, i.e., it should attain good performance in various input models without knowing in advance which input model it is facing.  The notion of performance depends on the input model:
\begin{itemize}
    \item \textbf{Stochastic I.I.D.~Input} (Section~\ref{sec:iid-model}). The requests are independently and identically distributed (i.i.d.) from a probability distribution $\cP\in \Delta(\cS)$ that is unknown to the decision maker, where $\Delta(\cS)$ is the space of all probability distributions over support set $\cS$. We measure the regret of an algorithm as the worst-case difference over distributions in $\Delta(\cS)$, between the expected performance of the benchmark and the algorithm:
    \nnvspace
    \begin{align*}
    \Regret{A} = \sup_{\cP \in \Delta(\cS)}  \left\{ \EE_{\vgamma \sim \cP^T} \left[ \OPT(\vgamma) - R(A|\vgamma) \right] \right\}\,.
    \end{align*}
    \nnvspace
    We say an algorithm is low regret if the regret grows sublinearly with the number of periods.
    
    \item \textbf{Adversarial Input} (Section~\ref{sec:adversarial}). The requests are arbitrary and chosen adversarially. Unlike the stochastic i.i.d.~input model, regret can be shown to grow linearly with $T$ and it becomes less meaningful to study the order of regret over $T$. Instead, we say that algorithm $A$ is asymptotically $\alpha$-competitive, for $\alpha \ge 1$, if
    \nnvspace
    \[
        \lim\sup_{T\rightarrow\infty} \sup_{\vec \gamma \in \cS^T}  \left\{ \frac 1 T \Big( \OPT(\vec \gamma) - \alpha R(A | \vec \gamma) \Big) \right\}\le 0\,.
    \]
    \nnvspace
    An asymptotic $\alpha$-competitive algorithm asymptotically guarantees fraction of at least $1/\alpha$ of the best performance in hindsight.\footnote{While we assume that the requests $\vgamma_T = (\gamma_1,\ldots,\gamma_T)$ are fixed in advance, our results allow requests to be chosen by a nonoblivious or adaptive adversary who does not know the internal randomization of the algorithm.}
    
    \item \textbf{Non-Stationary Input} (Section~\ref{sec:non-stationary}). The request at time $t$ is drawn from an arbitrary distribution that may be correlated across time. We denote by $\cP \in \Delta(\cS^T)$ the joint distribution over inputs. The regret of an algorithm $A$ over a class of joint distributions over inputs $\mathcal C$ is given by 
    \nnvspace
     \begin{align*}
    \Regret{A | \mathcal C} = \sup_{\cP \in \mathcal C}  \left\{ \EE_{\vgamma \sim \cP} \left[ \OPT(\vgamma) - R(A|\vgamma) \right] \right\}\,.
    \end{align*}
    \nnvspace
    We consider three classes of stochastic processes: ergodic input, periodic input, and adversarially corrupted i.i.d.~input.

\end{itemize}



\nnvspace
\subsection{The Dual Problem}

Our algorithms are of dual-descent nature and, thus, the Lagrangian dual problem of \eqref{eq:OPT} plays a key role. We construct a Lagragian dual of \eqref{eq:OPT} in which we move the constraints to the objective using a vector of Lagrange multipliers $\mu \ge 0$. For $\mu \in \RR_+^m$ we define
\begin{equation}\label{eq:conjugate}
f_t^*(\mu):=\sup_{x\in \cX_t} \{f_t(x)-\mu^{\top} b_t(x)\}\,,
\end{equation}
as the optimal opportunity-cost-adjusted reward of request $\gamma_t$. $f_t^*(\mu)$ is a generalization of the convex conjugate of the function $f_t(x)$ that takes account of the consumption function $b_t(x)$ and the constraint space $\cX_t$. In particular, when $b_t(x)=x$ and $\cX_t$ is the whole space, we recover the standard definition of convex conjugate. For a fixed input $\vgamma$, define the Lagrangian dual function $D(\mu | \vgamma):\RR_+^m\rightarrow\RR$ as
$$D(\mu | \vgamma ):=\sum_{t=1}^T f_t^*(\mu)+T \rho^\top \mu \ .$$ Then, we have by weak duality that $D(\mu|\vgamma)$ provides an upper bound on $\OPT(\vgamma)$. 
\begin{prop}\label{prop:upper_bound}
It holds for every $\mu \in \RR_+^m$ that $\OPT(\vgamma) \le D(\mu | \vgamma )$.
\end{prop}
All proofs are available in the appendix. Note, however, that strong duality does not necessarily hold due to the potential non-convexity of the reward and resource consumption functions, and the action space. {\color{black} Why our algorithms attain good primal performance despite strong duality not holding? The intuition can be explained in terms of  existing results on the primal-dual gap of time-separable non-convex optimization problems. It can be shown using Shapley-Folkman Theorem that the primal-dual gap in the offline problem is upper bounded by a constant that is independent from the number of requests $T$ even without convexity (see Proposition 5.26 of~\citealt{bertsekas2014constrained} for a detailed explanation). Therefore, dual approaches can be effective for online allocation problems because the primal-dual gap is usually small.  Interestingly, the analysis of our algorithms does not use Shapley-Folkman Theorem and, as a result, it provides an alternative proof that the primal-dual gap of separable non-convex problems is relatively small when the number of requests is large.}



\nnvspace
\subsection{Algorithm}
\begin{algorithm}[t]
	\SetAlgoLined
	{\bf Input:} Initial dual solution $\mu_1$, total time periods $T$, initial resources $B_1=T\rho$, reference function $h(\cdot): \RR^m\rightarrow \RR$, and step-size $\eta$. \\
	\For{$t=1,\ldots,T$}{
		Receive request $(f_t, b_t, \cX_t)$.\\
		Make the primal decision $x_t$ and update the remaining resources $B_t$:
		\begin{align}
		\tilde{x}_{t} &\in \arg\max_{x\in\cX_t}\left\{f_{t}(x)-\mu_{t}^{\top} b_{t}(x)\right\} \ ,\label{eq:primal_decision} \\
		x_{t}&=\left\{\begin{array}{cl}
		\tilde{x}_t    & \text{  if } b_t(\tilde{x}_t)\le B_t  \\
		0
		& \text{ otherwise} 
		\end{array} \right. \ , \nonumber \\
		B_{t+1} &= B_t - b_t(x_t) . \nonumber
		\end{align}\\
		Obtain a sub-gradient of the dual function: $$ \  \tg_t = -b_t(\tilde x_t) + \rho\ .$$\\

		 Update the dual variable by mirror descent: 
			\begin{equation}\label{eq:dual_update}
           \mu_{t+1} = \arg\min_{\mu\in\cD} \tg_t^\top \mu + \frac{1}{\eta} V_h(\mu, \mu_t) \ ,
       \end{equation}
       {\color{black} where $V_h(x,y)=h(x)-h(y)-\nabla h(y)^\top (x-y)$ is the Bregman divergence.}
		
	}
	\caption{Dual Mirror Descent Algorithm for \eqref{eq:OPT}}
	\label{al:sg}
\end{algorithm}


Algorithm~\ref{al:sg} presents the main algorithm we study in this paper. Our algorithm keeps a dual variable $\mu_t \in \RR^m$ for each resource that is updated using mirror descent, which is the workhorse algorithm of convex optimization~\citep{NemirovskyYudin83,beck2003mirror,hazan2016introduction,lu2018relatively}. 

{\color{black} Algorithm~\ref{al:sg} takes an initial dual variable, a step-size, and a reference function as inputs.} At time $t$, the algorithm receives a request $(f_t,b_t, \cX_t)$, and computes the optimal response $\tx_t$ that maximizes an opportunity-cost-adjusted reward of this request based on the current dual solution $\mu_t$. It then takes this action (i.e., $x_t=\tx_t$) if the action does not exceed the resource constraint, otherwise it takes a void action (i.e., $x_t=0$). Writing the dual function as $D(\mu | \vgamma ):=\sum_{t=1}^T D_t(\mu | \gamma_t)$ where the $t$-th term of the dual function is given by $D_t(\mu | \gamma_t) = f_t^*(\mu) + \mu^\top \rho$, it follows that $\tg_t:=-b_t(\tx_t)+\rho$ is a sub-gradient of $D_t(\mu| \gamma_t)$ at $\mu_t$ under our assumptions by Danskin's theorem (see, e.g., \citealt[Proposition B.25]{Bertsekas}). 
Finally, the algorithm utilizes $\tg_t$ to update the dual variable by performing a mirror descent descent step \eqref{eq:dual_update} with step-size $\eta$ and reference function $h$. The mirror descent step~\eqref{eq:dual_update} can be interpreted as minimizing over the non-negative orthant a first-order Taylor expansion of the dual objective plus a term that penalizes movement from the incumbent solution $\mu_{t}$ using the Bregman divergence as a measure of distance. Intuitively, by minimizing the dual function, the algorithm seeks to obtain dual variables that can be used to price resources and yield good performance when optimizing primal decisions. While we focus on online mirror descent for the dual update in the paper, our analysis is readily extendable to other algorithms for online linear optimization (see Section~\ref{sec:beyond-md} for more details).

{\color{black} The algorithm can be efficiently implemented when the optimization problem \eqref{eq:primal_decision} is simple to solve. As we discuss in Section~\ref{sec:applications}, in many applications the primal decision problem can be solved in closed form or by a linear-time algorithm. When the primal decision problem cannot be efficiently solved to optimality, it is sufficient for our purpose to produce approximately optimal solutions, and the approximation error would appear additively in the regret bound (see the discussion in Section \ref{sec:approximation}).} In practice, and whenever tractable, it is convenient to determine the actions $\tilde x_t$ directly under the constraint $b_t(\tilde x_t) \le B_t$ in \eqref{eq:primal_decision}{, namely $\tilde{x}_{t} \in \arg\max_{x\in\cX_t, b_t(x)\le B_t}\left\{f_{t}(x)-\mu_{t}^{\top} b_{t}(x)\right\}$}. Such choice does not affect our analysis, but can lead to better performance in applications.

{Finally, in most cases, the mirror descent step can be computed in linear time as \eqref{eq:dual_update} admits a closed-form solution. For example, if the reference function is $h(\mu)=\sum_j \mu_j\log(\mu_j)$, the dual update becomes
\begin{equation}\label{eq:mwu}
    {\mu_{t+1}} = {\mu_t} * \exp(-\eta \tg_t) \ ,
\end{equation}
{\color{black} where $x * y = (x_j y_j)_{j=1}^m$ is the element-wise product of vectors $x,y \in \RR^m$.} In this case, we recover the multiplicative weights update algorithm~\citep{arora2012multiplicative}. If the reference function is $h(\mu)=\frac{1}{2}\|\mu\|^2_2$, the dual update becomes 
\begin{equation}\label{eq:ogd}
    {\mu_{t+1}}= \text{Proj}_{\mu\ge 0} \{\mu_{t} - \eta  \tg_t\}\ ,
\end{equation}
{\color{black} where $\text{Proj}_{\mathcal A}\{\mu\}$ denotes the projection of the vector $\mu \in \RR^m$ to the set $\mathcal A \subset \RR^m$.} The latter recovers the online sub-gradient descent method.}




\nnvspace

\subsection{General Assumptions}\label{sec:assumption}
In this section, we present some common assumptions required in our analysis. We equip the primal space of the resource constraints $\mathbb R^m$ with the $\ell_\infty$ norm $\|\cdot\primalnorm$, and the Lagrangian dual space with the $\ell_1$ norm $\|\cdot\|_{1}$. Such choices of norms come naturally from our analysis, and also make sure that the parameters $\ubb,\ubrho$ defined below are dimension independent (i.e., independent from $m$). {\color{black} Similar regret guarantees with the same dependence on the number of resources can be obtained using the $\ell_p$ norm for the primal space and $\ell_q$ norm for the dual space with $1/p+1/q = 1$ and $p \in [2,\infty]$.}


	\begin{ass}[Assumptions on the requests]\label{ass:p}
	There exists $\ubf\in\RR_{+}$ and $\ubb \in \RR_+$ such that for all requests $(f,b,\cX) \in \cS$ in the support, it holds that \color{black}
	\begin{enumerate}
 	    \item The feasible set satisfies $0\in \cX$.
 	    \item The reward functions satisfy $0 \le f(x)\le \ubf$ for every $x\in\cX$.
 	    \item The resource consumption functions satisfy $b(x)\ge 0$ and $\|b(x)\primalnorm \le \ubb$ for every $x\in\cX$.
 	    \item The optimization problems in \eqref{eq:primal_decision} admit an optimal solution.
 	\end{enumerate}
	\end{ass}
	
The upper bounds $\ubf$ and $\ubb$ impose regularity on the space of requests, and we do not need these upper bounds to run the algorithm but they appear in our performance bounds. {\color{black} With the above assumptions, $f(x)$ can be non-concave, $b(x)$ can be non-convex, and $\cX$ can be non-convex and even integral, in contrast to previous literature.  We assume that the primal optimization problems in \eqref{eq:primal_decision} admit an optimal solution to simplify the exposition---our results hold even when approximately optimal solutions are available (see Section \ref{sec:approximation}). By Weierstrass theorem, a sufficient condition for the existence of optimal solutions is that the reward function $f$ is upper-semicontinuous, the resource consumption function $b$ is component-wise lower-semicontinuous, and the feasible set $\cX$ is compact.}

\begin{ass}[Assumptions on reference function]\label{ass:h} We assume
 \begin{enumerate}
     \item $h(\mu)$ is either differentiable or essentially smooth \citep{bauschke2001essential} in $\RR_+^m$.

     \item $h(\mu)$ is $\sigma$-strongly convex in $\|\cdot\dualnorm$-norm in $\RR_+^m$, i.e., $h(\mu_1)\ge h(\mu_2) + \nabla h(\mu_2)^\top (\mu_1-\mu_2) + \frac{\sigma}{2}\|\mu_1-\mu_2\dualnorm^2$ for any $\mu_1,\mu_2\in\RR_+^m$.
 \end{enumerate}
 \end{ass}
 

Strong convexity of the reference function is a standard assumption for the analysis of mirror descent algorithms~\citep{bubeck2015convex}. The previous assumptions imply, among other things, that the projection step \eqref{eq:dual_update} of the algorithm always admits a solution  by Weierstrass theorem.

We denote by $\lbrho = \min_{j \in [m]} \rho_j >0$ the lowest resource parameter and $\ubrho= \max_{j \in [m]} \rho_j=\|\rho\primalnorm$ the largest resource parameter, which is also the primal norm of the resource vector.

\nvspace

\section{Stochastic I.I.D.~Input Model}\label{sec:iid-model}

In this section, we assume the request $(f_t,b_t,\cX_t)$ at time $t$ is generated i.i.d.~from an unknown distribution $\cP\in \Delta(\cS)$ and where $\Delta(\cS)$ is the space of all probability distributions over support set $\cS$. The next theorem presents the worst-case regret bound of Algorithm \ref{al:sg}. 
	\begin{thm}\label{thm:master}		Consider Algorithm \ref{al:sg} with step-size $\eta \ge 0$ and initial solution $\mu_1\in \RR_+^m$. Suppose Assumptions~\ref{ass:p}-\ref{ass:h} are satisfied and the requests come from an i.i.d.~model with unknown distribution. Then, it holds for any $T\ge 1$ that
		\begin{align}\label{eq:master}
		\begin{split}
		\Regret{A}\le C_1 + C_2 \eta T + \frac {C_3}{\eta} \,.
		\end{split}
		\end{align}
		where $C_1 = \ubf \ubbinfty  / \lbrho$, $C_2 = (\ubb + \ubrho)^2/2\sigma$, {\color{black} $C_3= \max \left\{  V_h(\mu, \mu_1) : \mu \in \{0, (\bar f/\rho_1) e_1,\ldots, (\bar f/\rho_m) e_m\}  \right\}$ and where $e_j\in\RR^m$ is the $j$-th unit vector.}
	\end{thm}

\begin{proof}

We prove the result in three steps. First, we lower bound the cumulative reward of the algorithm up to the first time that a resource is close to being depleted in terms of the dual objective and complementary slackness. Second, we bound the complementary slackness term by picking a suitable ``pivot'' for online mirror descent. We conclude by putting it all together in step three.

\paragraph{Step 1 (Primal performance.)} First, we define the stopping time $\tA$ of Algorithm~\ref{al:sg} as the first time less than $T$ that there exists resource $j$ such that $\sum_{t=1}^{\tA} (b_t(x_t))_j + \ubbinfty \ge \rho_j T$. Notice that $\tau_A$ is a random variable, and moreover, we will not violate the resource constraints before the stopping time $\tau_A$. We here study the primal-dual gap until the stopping-time $\tA$. Notice that before the stopping time $\tA$, Algorithm \ref{al:sg} performs the standard mirror descent steps on the dual function because $\tx_t = x_t$.


Consider a time $t \le \tA$ so that actions are not constrained by resources. Because $x_{t} \in \arg\max_{x\in\cX}\{f_{t}(x)-\mu_{t}^{\top} b_t(x)\}$, we have that
\begin{equation}\label{eq:primal-approximation}
    f_t(x_t) = f_t^*(\mu_t) + \mu_{t}^{\top} b_t(x_t)\,.
\end{equation}

{Let $\bar D(\mu | \cP) = \frac 1 T \mathbb E_{\vgamma \sim \cP^T}\left[D(\mu | \vgamma)\right] = \EE_{(f,b) \sim \cP}\left[f^*(\mu_t)\right] + \mu_t^{\top} \rho$ be the expected dual objective at $\mu$ when requests are drawn i.i.d.~from $\cP \in \Delta(\cS)$.} Let $\xi_t=\{\gamma_0,\ldots, \gamma_t\}$ and $\sigma(\xi_t)$ be the sigma-algebra generated by $\xi_t$. Adding the last two equations and taking expectations conditional on $\sigma(\xi_{t-1})$ we obtain, because $\mu_t \in \sigma(\xi_{t-1})$ and $(f_t,b_t) \sim \cP$, that
\begin{align}\label{eq:bound_one_period}
\mathbb E\left[ f_t(x_t)  | \sigma(\xi_{t-1}) \right] \nonumber 
&=\EE_{(f,b) \sim \cP}\left[f^*(\mu_t)\right] + \mu_t^{\top} \rho + \mu_t^\top \left( \mathbb E\left[ b_t(x_t)| \sigma(\xi_{t-1}) \right] - \rho \right)    \nonumber\\
&=   \bar D(\mu_t| \cP) - \mathbb E\left[ \mu_t^\top \left( \rho - b_t(x_t) \right) | \sigma(\xi_{t-1}) \right]
\end{align}
where the second equality follows the definition of the dual function. 

Consider the process $Z_t = \sum_{s=1}^t \mu_s^\top \left(a_s - b_{s} (x_s)\right) - \mathbb E\left[ \mu_s^\top \left(a_s - b_{s} (x_s) \right) | \sigma(\xi_{s-1}) \right]$, which is martingale with respect to $\xi_t$ (i.e., $Z_t \in \sigma(\xi_t)$ and $\EE[Z_{t+1} | \sigma(\xi_t)] = Z_t$). Since $\tA$ is a stopping time with respect to $\xi_t$ and $\tA$ is bounded, the Optional Stopping Theorem implies that $\EE\left[Z_{\tA}\right] = 0$. Therefore, 
\begin{align*}
\EE\left[\sum_{t=1}^{{\tA}} \mu_t^\top \left(\rho - b_t(x_t)\right) \right] 
= \EE\left[\sum_{t=1}^{{\tA}} \mathbb E\left[ \mu_t^\top \left(\rho - b_t(x_t)\right) | \sigma(\xi_{t-1}) \right] \right]\,.
\end{align*}
Using a similar martingale argument for $f_t(x_t)$ and summing \eqref{eq:bound_one_period} from $t=1,\ldots,\tA$ we obtain that
\begin{align}\label{eq:f_and_r}
\mathbb E\left[ \sum_{t=1}^{\tA} f_t(x_t)  \right] 
&= 
\mathbb E\left[ \sum_{t=1}^{\tA} \bar D(\mu_t| \cP) \right ]  - \mathbb E\left[ \sum_{t=1}^{\tA} \mu_t^\top \left(  \rho - b_t(x_t) \right) \right] \nonumber\\
&\ge 
\mathbb E\left[ \tA \bar D(\bmu_{\tA}| \cP) \right ] - \mathbb E\left[ \sum_{t=1}^{\tA} \mu_t^\top \left( \rho - b_t(x_t) \right) \right]\,.  
\end{align}
where the inequality follows from denoting $\bmu_{\tA} = \frac 1 {\tA} \sum_{t=1}^{\tA} \mu_t$ to be the average dual variable and using that the dual function is convex.

\paragraph{Step 2 (Complementary slackness).} Consider the sequence of functions $w_t(\mu) = \mu^\top(\rho - b_t(x_t))$, which capture the complementary slackness at time $t$. The gradients are given $\nabla_\mu w_t(\mu) = \rho - b_t(x_t)$, {\color{black} which are bounded as follows $\|\nabla_\mu w_t(\mu) \primalnorm \le \|b_t(x_t)\primalnorm + \|\rho\primalnorm \le \ubb + \ubrho$}. Therefore, Algorithm~\ref{al:sg} applies online mirror descent to these sequence of functions $w_t(\mu)$, and we obtain from Proposition~\ref{prop:omd} that for every $\mu \in \cD$
\begin{align}\label{eq:regret_omd}
\sum_{t=1}^{\tA} w_t(\mu_t) - w_t(\mu) \le E(\tA, \mu) \le E(T,\mu)\,,
\end{align}
where $E(t,\mu) = \frac 1 {2\sigma} (\ubb + \ubrho)^2 \eta \cdot t + \frac{1}{\eta} V_h(\mu,\mu_1)$ is the regret of the online mirror descent algorithm after $t$ iterations, and the second inequality follows because $\tA \le T$ and the error term $E(t,\mu)$ is increasing in $t$.

\paragraph{Step 3 (Putting it all together).}
For any $\cP\in \Delta(\cS)$ and $\tA \in [0,T]$ we have that
\begin{align}\label{eq:bound-opt}
\EE_{\vgamma \sim \cP^T} \left[ \OPT(\vgamma) \right] &=  \frac{\tA}{T} \EE_{\vgamma \sim \cP^T} \left[ \OPT(\vgamma) \right] + \frac{T-\tA}{T}\EE_{\vgamma \sim \cP^T} \left[ \OPT(\vgamma) \right] \le  \tA \bar D(\bmu_{\tA}| \cP) + \pran{T-\tA}{\ubf}\ ,
\end{align}
where the inequality uses Proposition \ref{prop:upper_bound} and the fact that $\OPT(\vgamma)\le T \ubf$. Let $\Regret{A|\cP} = \EE_{\vgamma \sim \cP^T} \left[ \OPT(\vgamma) - R(A|\vgamma)\right]$ be the regret under distribution $\cP$. Therefore,
{\small
\begin{align*}
\Regret{A|\cP} &= \EEcP{ \OPT(\vgamma)-R(A|\vgamma)} 
\le \EEcP{ \OPT(\vgamma) - \sum_{t=1}^{\tA} f_t(x_t)}  \le \EEcP{ \OPT(\vgamma) - \tA D(\bmu_{\tA}| \cP)  +\sum_{t=1}^{\tA} \left(w_t(\mu_t) \right)  }\\  
&\le \EEcP{ \OPT(\vgamma) - \tA D(\bmu_{\tA}| \cP)  +\sum_{t=1}^{\tA} w_t(\mu) + E(T,\mu)
	 }       \le \EE_{\cP} \Bigg[ \underbrace{(T - \tA) \cdot \bar f    + \sum_{t=1}^{\tA} w_t(\mu) + E(T,\mu)}_{\clubsuit} \Bigg] \,,
\end{align*}}%
where the first inequality follows from using that $\tA \le T$ together with $f_t(\cdot) \ge 0$ to drop all requests after $\tA$; the second is from \eqref{eq:f_and_r}; the third follows from \eqref{eq:regret_omd}; and the last from \eqref{eq:bound-opt}.

We now discuss the choice of $\mu \in \RR_+^m$. If $\tA = T$, then set $\mu = 0$ to obtain that $\clubsuit \le E(T, 0)$. If $\tA < T$, then there exists a resource $j\in[m]$ such that $\sum_{t=1}^{\tA} (b_t(x_t))_j + \ubbinfty \ge T \rho_j $. Set $\mu = (\ubf/\rho_j) e_j$ with $e_j$ being the $j$-th unit vector. This yields
\begin{align*}
\sum_{t=1}^{\tA} w_t(\mu)
&=\sum_{t=1}^{\tA} \mu^\top (\rho - b_t(x_t))
= \frac \ubf {\rho_j} \sum_{t=1}^{\tA} \left( \rho_j - (b_t(x_t))_j  \right)\le \frac \ubf {\rho_j} \left( \tA \rho_j - T \rho_j + \ubbinfty \right)
= \frac \ubf {\rho_j} \ubbinfty - \ubf (T - \tA)\,,
\end{align*}
where the inequality follows because of the definition of the stopping time $\tA$. Therefore, using that  $\rho_j \ge \lbrho$ for every resource $j \in [m]$
\[
\clubsuit \le \frac{\ubf \ubbinfty}{\lbrho} + E(T,\mu) \le \frac{\bar f \ubbinfty }{\lbrho} + \frac 1 {2 \sigma} (\ubb + \ubrho)^2 \eta\cdot T + \frac{1}{\eta} V_h(\mu, \mu_1) \,,
\]
where the second inequality follows from our formulas for $\ubf$ and $E(T,\mu)$. {\color{black} We conclude by combining the cases for $\tA = T$ and $\tA < T$, and using that $\mu \in \{0, (\bar f/\rho_1) e_1,\ldots, (\bar f/\rho_m) e_m\}$ to upper bound $V_h(\mu, \mu_1)$ in terms of the worst-case realization of the pivot.}
\end{proof}

{\color{black} The constant $C_1$ in \eqref{eq:master} arises from the analysis of the stopping time $\tA$. The constants $C_2$ and $C_3$ follow from the standard regret analysis of online mirror descent (see Proposition~\ref{prop:omd} in the appendix). In particular, $C_2$ depends on the norm of the sub-gradients of the dual function and the strong-convexity constant of the reference function, while $C_3$ depends on the distance of the initial dual solution to the pivot as measured according to the Bregman divergence.} A salient feature of our regret bound is its independence from the cardinality of the action space $\mathcal X$.   This is attractive for many applications, such as assortment optimization (see Section~\ref{sec:applications}), in which the cardinality of actions can be exponentially large.

The previous result implies that, by choosing a step-size of order $\eta \sim T^{-1/2}$, Algorithm~\ref{al:sg} attains regret of order $O(T^{1/2})$ when the length of the horizon and the initial amount of resources are simultaneously scaled. {\color{black} In particular, the optimal choice of the step size in Theorem~\ref{thm:master} is given by $\eta = \sqrt{C_3/ (C_2 T)}$, which yields $\Regret{A}\le C_1 + 2 \sqrt{C_2 C_3 T}$. Therefore, our algorithms are also asymptotically optimal, in the sense that $\lim_{T\rightarrow \infty} \Regret{A}/T \rightarrow 0$.}

{\color{black} We now briefly discuss the instantiation of our regret bound for the reference functions discussed in the previous section (full details are available in Appendix~\ref{sec:reference-functions}).

\begin{itemize}
    \item Suppose $h(\mu)=\frac{1}{2}\|\mu\|_2^2$ and $\mu_1=0$. Then, Algorithm~\ref{al:sg} recovers dual online sub-gradient descent, and with proper step-size $\eta$ we can obtain a regret of order $O( m^{1/2} T^{1/2} )$. This follows because the reference function is now {$1/m$}-strongly convex with respect to the dual norm.

    \item Suppose $h(\mu)=\sum_j \mu_j \log(\mu_j)$ and $\mu_1=e/m$. For the multiplicative weights update algorithm, we cannot invoke Theorem~\ref{thm:master} directly because the reference function $h(\mu)$ is not strongly convex over the non-negative orthant as its ``curvature'' converges to zero for large values of $\mu$. Using Proposition~\ref{thm:stopping_time}, we can show that the dual variables are uniformly bounded and, by restricting the reference function to a box, we can obtain a regret bound of order $O( (m \log m)^{1/2} T^{1/2} )$.
\end{itemize}
}

We remark that the dependency of our regret bounds on the number of resources $m$ is sub-optimal. \citet{Agrawal2014OR} shows that the best possible dependence on the number of resources is of order $\log^{1/2} (m)$ while our algorithms' dependency is of polynomial order on $m$. The algorithms in \citet{Agrawal2014OR}, \citet{AgrawalDevanue2015fast}, and \citet{Devanur2019near} attain the optimal dependency on the number of resources, but, differently to ours, require either knowing an estimate on the value of benchmark or periodically solving large optimization problems. {\color{black} We would like to mention that if the upper bound $\bar f$ is available (or alternatively, a bound on the expected benchmark), then by constraining dual variables to lie in the scaled unit simplex, our algorithm can be shown to attain regret bounds of order $O( \log^{1/2} (m) T^{1/2} )$. We provide more details in Appendix~\ref{sec:reference-functions}}.

We next discuss the tightness of our regret bound. The following result, which we reproduce without proof, shows that one cannot hope to attain a regret lower than $\Omega( T^{1/2} )$ under our minimal modeling assumptions.


\begin{lem}[Lemma 1 from \citealt{ArlottoGurvich2019}]\label{lem:lower_bound} For every $T \ge 1$, there exists a probability distribution $\cP$ such that
\[
    \inf_A  \EE_{\vgamma \sim \cP^T} \left[ \OPT(\vgamma) - R(A|\vgamma) \right]  \ge C T^{1/2}\,,
\]
where $C$ is a constant independent of $T$.
\end{lem}

The previous result shows that, for every $T$, there exists a probability distribution under which all algorithms---even those that know the probability distribution---incur $\Omega(T^{1/2})$ regret. The worst-case distribution used in the proof of the result assigns mass to three points with one point having mass of order $T^{-1/2}$. Because the regret bound of Algorithm~\ref{al:sg} provided in Theorem~\ref{thm:master} does not depend on the probability mass function of the distribution $\cP$, it readily follows that our algorithm also attains $O(T^{1/2})$ in such worst-case instances. This implies that our algorithm attains the optimal order of regret when the length of the horizon and initial number of resources are scaled proportionally. {Under further assumptions on the input, however, it is sometimes possible to attain better regret guarantees (see, e.g., \citealt{Jasin2015unknown} and \citealt{LiYe2019online}).}






\nvspace
\section{Adversarial Input}\label{sec:adversarial}
In this section, we assume the request $(f_t,b_t,\cX_t)$ at time $t$ is chosen by an adversary, and we look at the worst-case performance over all possible inputs. The next theorem shows that Algorithm \ref{al:sg} is $\alpha^*$-competitive with $\alpha^*=\max\{\sup_{(f, b, \cX) \in \cS} \sup_{j \in [m], x\in\cX} b_j(x) / \rho_j,1\}$.

\begin{thm}\label{thm:adversial-no-regularizer}
Consider Algorithm \ref{al:sg} with step-size $\eta \ge 0$ and initial solution $\mu_1\in \cD$. Suppose Assumptions~\ref{ass:p}-\ref{ass:h} are satisfied, and the requests are chosen by an adversary. Then, it holds for any $T\ge 1$ that
		\begin{align}\label{eq:master-adver-one}
		\begin{split}
		\OPT(\vec \gamma) - \alpha^* R(A | \vec \gamma) \le
		C_1 + C_2 \eta T + \frac{C_3}{\eta}  \ ,
		\end{split}
		\end{align}
		where $C_1 = \ubf \ubbinfty  / \lbrho$, $C_2 = \alpha^* (\ubb + \ubrho)^2/2$, {\color{black} $C_3= \max \left\{ \alpha^* V_h(\mu, \mu_1) : \alpha^* \mu \in \{0, (\bar f/\rho_1) e_1,\ldots, (\bar f/\rho_m) e_m\}  \right\}$.}
\end{thm}

{
Differently to the case of stochastic inputs, we prove the result in the primal space by comparing the choices of our algorithm to an optimal solution of the benchmark~\eqref{eq:OPT}. In particular, we can show that, for each request, the reward obtained by our algorithm is at most a fraction $1/\alpha^*$ of the reward of an optimal, offline solution minus a complementary slackness term. Then, as in the stochastic case, we relate the complementary slackness term to the reward lost when the stopping time is not close to the end of the horizon.}

When the step-size is $\eta\sim T^{-1/2}$, Theorem \ref{thm:adversial-no-regularizer} shows that Algorithm~\ref{al:sg} is $\alpha^*$-competitive, i.e., it guarantees at least a fraction $1/\alpha^*$ of the best performance in hindsight as $T$ grows large. Indeed, Theorem~1 from \citet{BalseiroGur2019MS} presents a one-dimensional example which implies that one cannot hope to attain competitive ratio lower than $\alpha^*$ under our assumptions. Theorem~\ref{thm:adversial-no-regularizer} matches their lower bound, which implies that the competitive ratio is optimal without any further assumptions on the input. {\color{black} While the optimal step-size is different to that of the i.i.d.~the case, any step-size of order $\eta\sim T^{-1/2}$ guarantees the same asymptotic performance in both models.}

The competitive ratio $\alpha^*$ measures how resource constrained is the decision maker. For each resource $j \in [m]$, the expression $\max_{x\in\cX} b_j(x) / \rho_j$ captures the ratio of the highest possible resource consumption to the ``average'' amount of resource available per time period. Theorem~\ref{thm:adversial-no-regularizer} thus implies that the competitive ratio deteriorates as the problem becomes more resource constrained. 

We now show that Algorithm~\ref{al:sg} never depletes resources too early. Define the stopping time $\tA$ of Algorithm \ref{al:sg} as the first time less than $T$ that there exists resource $j$ such that 
\begin{equation*}
    \sum_{t=1}^{\tA} (b_t(x_t))_j + \ubbinfty \ge \rho_j T \ .
\end{equation*}
By construction, our algorithm does not violate the resource constraints before the stopping time $\tau_A$. We prove our result under the following assumption.

\begin{ass}[Separability of the reference function $h$]\label{ass:h-sep} The reference function $h(\mu)$ is coordinate-wisely separable, i.e., $h(\mu)=\sum_{j=1}^m h_j(\mu_j)$ where $h_j:\RR_+ \rightarrow \RR $ is an univariate function. Moreover, for every resource $j$ the function $h_j$ is $\sigma_2$-strongly convex over $[0, \mumax_j]$ with $\mumax_j:=\ubf / \rho_j+1$.
\end{ass}

If the reference function is not a coordinate-wise separable function, the projection step \eqref{eq:dual_update} can be hard to solve. Furthermore, most examples in the mirror descent literature utilize coordinate-wise separable reference functions~\citep{beck2003mirror,lu2018relatively,lu2017relative}. The next proposition says the stopping time $\tau_A$ is always close to the end of the horizon $T$.

\begin{prop}\label{thm:stopping_time}
Let $\mumax\in\RR^m$ be such that $\mumax_j:=\ubf / \rho_j+1$. Suppose Assumptions~\ref{ass:p} and \ref{ass:h-sep} holds, the initial dual solution satisfies $\mu \le \mumax$, and the  step-size satisfies $\eta\le\sigma_2/\ubbinfty$.  Then, it holds that $\mu_t\le \mumax$ for any $t\le T$. Furthermore,  it holds for every sample path $\vgamma$ that
\begin{equation*}
T-\tA \le \frac{1}{\eta\lbrho} \themax + \frac{{\ubb}}{\lbrho} \ .
\end{equation*}
\end{prop}

{When the step-size is $\eta\sim T^{-1/2}$, Proposition~\ref{thm:stopping_time} implies that $T-\tA = O(T^{1/2})$ and resources are never depleted too early.} In the proof of Proposition~\ref{thm:stopping_time} we first use the separability of the reference function to argue that, as long as $0\le\mu_1\le \mumax$, the dual variable obtained by Algorithm~\ref{al:sg} always stays in the domain $\mathcal{D}:=\{\mu \in \RR^m \mid 0\le \mu\le\mumax \}$. Recall that dual variables are increased when a request consumes more resources than the target. Because the dual variables are always bounded from above, resource consumption can never exceed the target by a large amount, which, in turn, implies that the resources are never depleted too early. As a result, resources are depleted smoothly over time, which is a desired feature in many settings.

\nvspace

\section{Non-stationary Input}\label{sec:non-stationary}
The previous two sections study two classic input models, i.e., i.i.d.~and adversarial. However, the i.i.d.~input may be too optimistic and the adversarial model may be too pessimistic in practice.
In this section, we consider three non-stationary stochastic input models that fill up the gap between i.i.d.~and adversarial input models, and they are more realistic in many applications.
\nnvspace
\subsection{Independent Inputs and Robustness to Adversarial Corruptions}\label{sec:independent}

We consider the case where requests are drawn from independent but not necessarily identical distributions. We introduce some new notations that will be used in the regret bound. Given two probability distributions $\cP_1, \cP_2$, we denote by $\|\cP_1-\cP_2\|_\TV$ the total variation distance between $\cP_1$ and $\cP_2$. We denote by $\text{MD}(\vcP) = \sum_{t=1}^T \left\| \cP_t - \frac 1 T \sum_{s=1}^T \cP_s \right\|_\TV$ the mean deviation of a vector of independent distributions $\vec \cP \in \Delta(\cS)^T$ from the average distribution in total variation norm. Moreover, let $\mathcal C^{\rm ID}(\delta) = \left\{ \vcP \in \Delta(\cS)^T : \MD(\vcP) \le \delta\right\}$ be the set of all independent inputs with mean deviation at most $\delta > 0$. {Such measure on the non-stationarity of the distributions is closely related to Kolmogorov metric, which was used to study posted pricing \citep{dutting2019posted}.} The next theorem presents our regret bound for Algorithm~\ref{al:sg} under independent inputs:


\begin{thm}\label{thm:independent}
Consider Algorithm \ref{al:sg} with step-size $\eta \ge 0$ and initial solution $\mu_1\in \cD$. Suppose Assumptions~\ref{ass:p}-\ref{ass:h} are satisfied, and the requests are drawn from independent (non-identical) distributions. Then, it holds for any $T\ge 1$ and mean deviation $\delta>0$ that
\[
    \Regret{A \mid \mathcal C^{\rm ID}(\delta)} \le C_1 + C_2 \eta T + \frac {C_3} {\eta} + C_4 \delta \ ,
\]
where {\color{black} the constants $C_1, C_2, C_3$ are defined in Theorem~\ref{thm:master}} and $C_4=\ubf$.
\end{thm}

When the step-size is of order $\eta\sim T^{-1/2}$, the regret of Algorithm~\ref{al:sg} becomes $O(T^{1/2}+\delta)$, which implies the performance of algorithm degrades gracefully with the mean deviation $\delta$ from the average distribution. Theorem~\ref{thm:independent} shows a natural transition from i.i.d.~input to adversarial input: when the requests are i.i.d., the mean deviation is $\delta=0$ and $\Regret{A}\sim O(T^{1/2})$; when all requests are adversarial, it is likely that the mean deviation is of order $\delta\sim T $ and thus $\Regret{A}\sim T$.

Theorem~\ref{thm:independent} implies that  Algorithm~\ref{al:sg} is robust to adversarial corruptions to i.i.d.~input. Adversarial corruptions to stochastic inputs have been recently studied by \cite{lykouris2018corruptions} and \cite{chen2019robust}. The main motivation of this line of work is to design algorithms that are robust to perturbations of the input. These perturbations can be either malicious, for example, in the case of click fraud; or non-malicious, for example, due to traffic spikes caused by unpredictable events. 

Consider the situation when most requests are drawn i.i.d.~from an unknown distribution model and an adversary can corrupt at most $r$ requests. A direct consequence of Theorem~\ref{thm:independent} is that the regret is of order $\max\{T^{1/2}, r\}$ by noticing $\text{MD}(\vcP)$ is proportional to $r$. In particular, if the adversary corrupts at most $O(T^{1/2})$ request, then the regret bound is still $O(T^{1/2})$, which showcases the robustness of Algorithm~\ref{al:sg} to adversarial corruptions. In a recent paper \cite{kesselheim2020knapsack}, study the secretary knapsack problem when an adversary corrupts a limited number of requests in a bursty pattern. While in their analysis they only consider the performance of the uncorrupted requests, they provide an algorithm with similar performance guarantees than ours. Their algorithm, however, is more complex and requires knowing the length of the bursts that the adversary corrupts. Subsequently, \cite{bradac2020robust} design robust algorithms for multiple secretary problems, but benchmark their algorithms against the uncorrupted requests and exclude the uncorrupted request with the highest reward. {Contemporaneously, \cite{jiang2020online} study an online resource allocation problem in which request are drawn from a non-stationary distribution. They consider a first setting in which the true distribution is unknown but a prior (potentially inexact) estimate is available and a second setting in which the true distribution is completely unknown. For the latter setting, they propose a gradient-descent algorithm that is similar to ours and prove regret bounds in which the mean-deviation is measured using Wasserstein distance instead of total variation distance.} 


Many algorithms studied in the literature are susceptible to adversarial corruptions. For example, training-based algorithms such as the one of \cite{DevanurHayes2009} or algorithms that require solving optimization problems with historical data such as the one of \cite{AgrawalDevanue2015fast} might perform poorly if an adversary corrupts a few, selected requests. This follows because, in these algorithms, most decisions are determined based on a few requests received in the first periods, which an adversary can corrupt.

The previous result implies that the performance of our algorithm degrades linearly with the amount of corruption from i.i.d.~input. The next theorem shows that linear degradation is necessary in the sense that every algorithm incurs a similar, linear degradation  in performance. We prove the result by invoking Yao's lemma and constructing a distribution over distributions under which no algorithm can perform well.

\begin{thm}\label{thm:lower-bound-inde}
For any length of horizon $T \ge 1$ and mean deviation $8\le\delta\le 4T$, there exist constants $C_1, C_2 > 0$, such that
\begin{equation}\label{eq:lower_bound_corruption}
    \inf_A \Regret{A\mid\mathcal C^{\rm ID}(\delta)}\ge C_1 \delta + C_2 T^{1/2} \ . 
\end{equation}
\end{thm}

{We remark that multi-arm bandits and online convex optimization problems have been recently studied in similar non-stationary stochastic settings~\citep{besbes2014stochastic, besbes2015nonstationary}. In this line of work, the goal is to design algorithms that perform well when the model primitives change throughout time. The metric considered by these papers is the variation budget, which captures how much the input changes from one time step
to the next, and they show that vanishing regret can be achieved as long as the variation budget is of order $o(T)$. In our setting, this metric would be given by  $\sum_{t=1}^{T-1}\|\cP_t-\cP_{t+1}\|_{\TV}$. The instance given in the proof Theorem~\ref{thm:lower-bound-inde} readily implies that the regret of every algorithm is $\Omega(T)$ even if the distributions $\cP_t$ change once throughout the horizon, i.e., the variation budget is constant. Therefore, this metric is not appropriate for our setting.}

\nnvspace
\subsection{Ergodic Input and Markov Processes}

We now consider stochastic input models that are not necessarily independent across time. In particular, we restrict attention to ergodic input processes, which, intuitively, satisfy the property that requests tend to be independent as they grow apart in time. These input processes allow for requests that are close in time to be correlated. Examples of ergodic processes are irreducible and aperiodic Markov chains and stationary autoregressive processes. Such processes are extensively used in time series analysis to estimate the arrivals of users to a website or jobs to a server and might lead to more realistic inputs models.

Let $\cP \in \Delta(\cS^T)$ be a stochastic process. Denoting by $\gamma_{1:t} = (\gamma_s)_{s=1}^t$ the sequence of inputs up to time $t$, we let $\mathcal P_t(\gamma_{1:s})$ be the conditional distribution of $\gamma_t$ given $\gamma_{1:s}$ for $s < t$.  For every $k \in [T]$ and a one-period distribution $\bar \cP \in \Delta(\cS)$ we denote by
\[
    \mathrm{TV}_k(\cP, \bar \cP) = \sup_{\vec \gamma \in \mathcal S^{T}} \sup_{t = 1,\ldots,T-k} \left \| \cP_{t+k} \left(\gamma_{1:{t-1}}\right) -  \bar \cP \right\|_\TV\,,
\]
the worst-case total variation distance between the distributions in period $t+k$ conditional on the data at the beginning of period $t$ and $\bar \cP$. When $\bar \cP$ is the stationary distribution of the process $\cP$, the metric $\mathrm{TV}_k(\cP, \bar \cP)$ gives the maximum distance between the $k$-step transition probability and the stationary distribution. Intuitively, if the process is ergodic, it should mix relatively quickly and the latter metric should decrease as $k$ increases. We refer to this metric as the \emph{$k$-step distance from stationarity}.

Let $\mathcal C^{\rm E}(\delta, k) = \left\{ \cP \in \Delta(\cS^T) :  \mathrm{TV}_k(\cP, \bar \cP) \le \delta \text{ for some } \bar \cP \in \Delta(\cS) \right\}$ be the set of all stochastic processes with $k$-step distance from stationary no larger than $\delta > 0$.  The next theorem presents our regret bound for Algorithm~\ref{al:sg} under ergodic inputs:

\begin{thm}\label{thm:ergodic}
Consider Algorithm \ref{al:sg} with step-size $\eta \ge 0$ and initial solution $\mu_1\in \cD$. Suppose Assumptions~\ref{ass:p}-\ref{ass:h} are satisfied and the requests come from an ergodic process. Then, it holds for any $T\ge 1$, $\delta \ge 0$, and $k \ge 0$ that
\[
\Regret{A \mid \mathcal C^{\rm E}(\delta, k) }  \le C_1 + C_2 \eta T + \frac {1} {\eta} C_3  + C_4 \eta T k + 2 \ubf  T \delta + 2 \ubf k\,,
\]
where {\color{black} the constants $C_1, C_2, C_3$ are defined in Theorem~\ref{thm:master}}  and $C_4= \frac{\sqrt{2} \ubb}{\sigma} (\ubb+\ubrho)$.

\end{thm}

The proof of the theorem is inspired by an analysis given by \cite{duchi2012ergodic} for the convergence of mirror descent for unconstrained stochastic optimization problems with ergodic input. A key step of the proof involves comparing the expected dual performance at time $t$ when the dual variable of the algorithm is $\mu_t$ under the ergodic process to the expected performance under its stationary distribution. Instead of looking at the performance at time $t$, we shift time by $k$ periods and compare to the expected performance at time $t+k$. On the one hand, ergodicity guarantees that the dual performance at time $t+k$ is close to the stationary expected performance assuming that the dual variables do not change. On the other hand, the dual variable at time $t+k$ is different to $\mu_t$ because our algorithm updates the dual variable after every iteration. To control the change in dual performance, we show that the dual variables do not change too much in $k$ steps and then use that the dual objective is Lipschitz continuous.


When the input is i.i.d., we have that $k = \delta = 0$ and we recover the bound from Theorem~\ref{thm:master}. Now, suppose that requests follow an irreducible and aperiodic Markov process. In this case, we can write $\cP_t(\gamma_s)$ for the distribution at time $t$ when the state is $\gamma_s$ at $s<t$ by the Markov property. Let $\bar \cP$ be its stationary distribution. If the Markov chain has a finite state-space or has a general state-space and is uniformly ergodic, then there exist constants $R>0$ and $\alpha \in (0,1)$ such that $\sup_{\gamma_s \in \cS} \| \cP_t(\gamma_s) - \bar \cP \|_\TV \le R \alpha^{t-s}$ (see, e.g., \citealt[Theorem 4.9]{levin2017markov} or \citealt[Chapter 16]{meyn2012markov}). This implies that the $k$-step distance from stationarity decreases exponentially fast in $k$. Therefore, we obtain a regret of $O(1/\eta + \eta T k + T \alpha^k)$. Setting $\eta \sim T^{-1/2}$ and $k=-\log T/(2\log \alpha)$ yields a regret of $O(T^{1/2} \log T)$. This regret bound matches the lower bound for i.i.d.~input up to a logarithmic term.

\nnvspace
\subsection{Periodic Input}

In many practical applications, requests exhibit periodicity or seasonality. For example, in internet advertising, traffics in the mornings are different from those in the evenings, but daily patterns tend to be consistent from one day to the next. Similarly,  requests during weekdays are different to those during weekends, but weekly patterns tend to repeat over time~(see, e.g., \citealt{zhou2019robust}). 

In this section, we consider a periodic, and thus dependent, input model. Suppose that requests have cycles of length $q \in \mathbb N$ so that requests within a cycle can be arbitrarily correlated but cycles, as a whole, are independently and identically distributed. Assume for simplicity that $T$ is divisible by $q$.  More formally, we define the class of all $q$-periodic requests distributions by $\mathcal C^{\rm P}(q) = \left\{ \cP \in \Delta(\cS^q)^{T/q} : \cP_{1:q} = \cP_{q+1:2q} = \ldots = \cP_{T-q+1:T}\right\}$ where $\cP_{s:t}$ denotes the joint distribution of requests $s\le t$. The next theorem presents the worst-case regret bound of Algorithm~\ref{al:sg} for this input model.


\begin{thm}\label{thm:periodic}
Consider Algorithm \ref{al:sg} with step-size $\eta \ge 0$ and initial solution $\mu_1\in \cD$. Suppose Assumptions~\ref{ass:p}-\ref{ass:h} are satisfied and the requests come from a periodic model. Then, it holds for any $T\ge 1$ and $q\ge 0$ that
\[
 \Regret{A \mid \mathcal C^{\rm P}(q) } \le C_1 + C_2 \eta T + \frac {1} {\eta} C_3 +  C_4 q \eta T \,,
\]
where {\color{black} the constants $C_1, C_2, C_3$ are defined in Theorem~\ref{thm:master}} and 
$C_4= \frac{\sqrt{2}}{\sigma} (\ubb+\ubrho)^2$.
\end{thm}

The above theorem leads to a regret of order $O(1/\eta + \eta q T)$. The optimal step-size is of order $\eta \sim (q T)^{-1/2}$ yielding a regret of $O((q T)^{1/2})$. Therefore, if the length of the cycles is $o(T)$, then our algorithm attains vanishing regret. If the step-size is chosen obliviously to the length of the period, i.e., $\eta \sim T^{-1/2}$, then our algorithm attains regret $O(q T^{1/2})$. 

Finally, Theorem~\ref{thm:periodic} presents another regret transition from i.i.d.~to adversarial input: when requests are i.i.d.~we have $q=1$ and the regret is $O(T^{1/2})$; when requests are adversarially chosen, we have $q=T$ and thus $\Regret{A}\sim O(T)$.

\nvspace
\section{Applications}\label{sec:applications}


{\color{black}
\nnvspace
\subsection{Online Linear Programming}\label{sec:nrm}

In an online linear program, a decision maker tries to dynamically allocate $m$ types of resources with inventory $B=T\rho\in \RR^m_+$ over a finite horizon $T$. At each time period $t$, a customer arrives and makes a request with an associated revenue vector $r_t\in \RR_+^d$ and consumption matrix $c_t\in \RR^{m\times d}_+$. The decision maker needs to choose, in real time, an action $x_t\in \cX_t \subset \RR^d_{+}$, where $\cX_t$ is the  action space at time $t$. Then, the reward function is $f_t(x_t)=r_t^\top x_t$ and the consumption function is $b_t(x_t)= c_t x_t$, and the offline problem is given by the following linear program:
\begin{align*}
	\max_{x: x_t\in \cX_t}\  \sum_{t=1}^T r_t^\top x_t     \ \ \ 
	\text{s.t.}    \   \sum_{t=1}^T c_t x_t \le B\,.
\end{align*}

\nvspace

Online linear programming has many applications in operations management. A special case when the decision maker makes only accept/reject decision, i.e., $d=1$ and $\cX_t=\{0,1\}$, is network revenue management, which dates back to \citet{glover1982passenger,wang1983optimum} and the algorithmic insights developed in the literature have been extensively applied in practice, with applications in airlines, hospitality, railways, and cloud computing. See \citet{bitran2003overview,talluri2006theory,gallego2019revenue} for the applications and more recent developments. Another application considerably studied in the computer science literature is online matching in which each request is assigned to at most one resource, i.e., $\cX_t=\{ x\in \{0,1\}^m : \sum_{j=1}^n x_j \le 1 \}$ and $c_t$ is a diagonal matrix (see, e.g., \citealt{karp1990optimal,feldman2009online}).}

Algorithm \ref{al:sg} and its analysis can be directly applied to online linear programming. The primal update \eqref{eq:primal_decision} becomes
$x_t \in \arg\max_{x\in\cX_t} \left\{ (r_t^\top-\mu_t^\top c_t) x \right\}$, which results in an online gradient $g_t = - c_t x_t + \rho$ that can be used in the dual mirror descent update \eqref{eq:dual_update}.

Compared with previous works on this problem, such as~\citet{agrawal2016efficient,DevanurHayes2009,Feldman2010,Devanur2019near}, our algorithm obtain the optimal $O(T^{1/2})$ regret under the stochastic i.i.d.~inputs, and it is fast as we do not need to solve a auxiliary linear programs. To break the $O(T^{1/2})$ regret rate, under additional strongly convexity assumptions on the dual problem, \cite{LiYe2019online} proposed an online algorithm for network revenue management (i.e., $d=1$ and $\cX_t=\{0,1\}$), which obtained $\log(T)$ regret. When input is adversarial, our algorithm yields the optimal asymptotic competitive ratio when the revenue vectors $r_t$ and resource consumption matrices $c_t$ are arbitrary. In particular, the worst-case instance of \cite{BalseiroGur2019MS} can be modified to show that our algorithm yields the optimal competitive ratio even when there is a single resource and each request consumes one unit, i.e., $c_t = 1$. This special case is called the \emph{single-leg revenue management problem}~\citep{TalluriVanRyzin2004} and, in this case, the competitive ratio of our algorithm is $1/\rho$, which is tight. Furthermore, we obtain new results in our three non-stationary stochastic input as specified in Section~\ref{sec:non-stationary}.

We remark, however, that our competitive ratio is not optimal when the problem has more structure. In the single-leg revenue management problem, \cite{ball2009toward} provide an algorithm whose competitive ratio depends on the support of the revenues $r_t$. In particular, when revenues can take $n$ different values, the competitive ratio is at most $n$ independently of $\rho$, which is tight. Their worst-case instances have $\rho = 1/n$, which matches the hardness result above described. {\color{black} Our adversarial results are also not optimal in the AdWords problem~\citep{Mehta2007JACM}---a special version of the online matching problem in which rewards are proportional to resource consumption, i.e., $c_t = \text{diag}(r_t)$.} In light of \cite{mirrokni2012simultaneous}, this should not be surprising as no algorithm that attains vanishing regret under stochastic input (as ours) for the AdWords problem can obtain a fixed competitive ratio that is independent of the resource vector $\rho$ under adversarial input.


\nnvspace
\subsection{Bidding in Repeated Auctions with Budgets}

As of 2019, around 85\% of all display advertisements are bought programmatically---using automated algorithms~\citep{eMarketer2019}. A common mechanism used by advertisers to buy ad slots is real-time auction: each time a user visits a website, an auction is run to determine the ad to be shown in the user's browser. Because there is a large number of these advertising opportunities in a given day, advertisers set budgets to control their cumulative expenditure. We discuss how to apply our methods to the problem of bidding in repeated auctions with budgets.


{\color{black}
We consider an advertiser with a budget $\rho T$ that limits the cumulative expenditure over $T$ auctions. Each request corresponds to an auction in which an \emph{impression} becomes available for sale. When the $t$-th impression arrives, the advertiser first learns a value $v_t$ for winning the impression based viewer-specific information and then determines a bid $x_t$ to submit to the auction.  We assume that impressions are sold using a second-price auction. Denoting by $d_t$ the highest bid submitted by competitors, the advertiser wins whenever his bid is the highest (i.e., $x_t \ge d_t$) and pays the second-highest bid in case of winning (i.e., $d_t \mathbf 1\{x_t \ge d_t\})$. To simplify the exposition, we assume that ties are broken in favor of the advertiser. At the point of bidding, the advertiser does not know the highest competing bid. Consistent with practice, we assume that the advertiser only observes his payment in case of winning.



This problem can be mapped to our framework by setting $f_t(x) = (v_t - d_t)\mathbf 1\{x_t \ge d_t\}$ and $b_t(x) = d_t \mathbf 1\{x_t \ge d_t\}$. With the benefit of hindsight, a decision maker can win an auction by bidding an amount equal to the highest competing bid (i.e., $x_t = d_t$). Therefore, the optimal solution in hindsight reduces to solving a knapsack problem in which the impressions to be won are chosen to maximize the net utility subject to the budget constraint. The problem is given by:
\begin{align*}
\begin{split}
        \max_{y_t \in \{0,1\} }\    \sum_{t=1}^T (v_t - d_t) y_t \ \ \ 
    \text{s.t.} \  \sum_{t=1}^T d_t y_t  \le T \rho,
    \end{split}
\end{align*}
where $y_t \in \{0,1\}$ is a decision variable indicating whether the advertiser wins the $t$-th impression.

Note that the informational assumptions are different from the ones of our baseline model because the competing bid $d_t$ is not assumed to be known at the point of bidding. Interestingly, because ads are sold using an ex-post incentive compatible auction, such information is not necessary for our algorithm: the algorithm only needs to know the payment incurred. In fact, our analysis applies to any other ex-post incentive compatible auction. To see this, denote by $\mu_t \ge 0$ the dual multiplier of the budget constraint and observe that the primal decision in Algorithm~\ref{al:sg} is
\begin{align*}
x_{t} 
&= \arg\max_{x}\left\{f_{t}(x)-\mu_{t} b_t(x)\right\} 
= \arg\max_{x}\left\{\big(v_t -(1+\mu_t) d_t \big)  \mathbf 1\{x \ge d_t\}\right\}\\
&= \arg\max_{x}\left\{ \left (\frac{v_t}{1+\mu_t} - d_t\right) \mathbf  1\{x \ge d_t\} \right\} = \frac{v_t}{1+\mu_t}\,,
\end{align*}
where we used that the sub-problem is equivalent to that of bidding in a second-price auction with value $v_t/(1+\mu_t)$ together with the truthfulness of the auction. The optimal decision can be implemented without knowing the maximum competing bid. After the bid, we observe the payment $b_t(x_t)$, which leads to an online dual sub-gradient $\tg_t = - b_t(x_t) + \rho$ that can be used in the dual mirror descent update~\eqref{eq:dual_update}.}

The problem of bidding in repeated auctions with budgets has been studied recently in \citet{BalseiroGur2019MS}. In their paper, they present an adaptive pacing strategy that attempts to learn an optimal Lagrange multiplier using sub-gradient descent. Their adaptive pacing strategy is shown to attain $O(T^{1/2})$ regret under stochastic i.i.d.~input with restrictive assumptions on the distribution of inputs. Specifically, they assume that values and competing bids are independent, and that the expected dual function $\mathbb E_{\vgamma} \left[ D(\mu | \vgamma) \right]$ is thrice differentiable and strongly convex. In practice, however, values and competing bids are positively correlated. Our algorithms attain similar regret bounds without such restrictive assumptions on the inputs in the stochastic i.i.d.~model, as well as other input models. In the case of adversarial input, \cite{BalseiroGur2019MS} showed that no algorithm can attain a competitive ratio better than $\bar v / \rho$, where $\bar v$ is a uniform upper bound on the advertiser's values. The competitive ratio of our algorithm is thus optimal for this problem. \cite{zhou2008budget} study an online knapsack problem with one resource and provide an algorithm whose competitive ratio depends on the range of the value-to-weight ratio of each item. Our parametric, adversarial bounds are not directly comparable with theirs.

\nnvspace
\subsection{Proportional Matching with High Entropy}\label{sec:matching}

Online matching is a central problem in computer science, with applications in online advertisement allocation, job/server allocation in cloud computing, product recommendation under resource constraints, etc. High-entropy proportional matching is a variant that has attracted attention lately because it has been shown to posses additional desirable properties, such as fairness and diversity~\citep{lan2010axiomatic, venkatasubramanian2010fairness, qin2013promoting, ahmed2017diverse}.


   

We here consider an online matching problem using the terminology of online advertising. To wit, we study an online advertisement allocation problem, where at each time period, the decision maker matches an incoming impression with one advertiser (who may have a capacity constraint), aiming to maximize the total reward over all incoming impressions while keeping a high entropy of such matchings. In this example, reward functions are non-linear but concave.

Suppose there are $m$ advertisers, a total of $T$ time periods, and the capacity of the $j$-th advertiser is $\rho_j T$. At time period $t$, an impression with revenue vector $r_t\in \RR^m$ arrives, i.e., if we allocate it to advertiser $j\in [m]$, then it generates revenue $(r_t)_j$. When an impression arrives, we decide an assignment probability variable $x_t\in \cX:=\{x\in\RR^m_+|\sum_{i=1}^m x_i\le 1 \}$, and assign the arriving impression to advertiser $j$ with probability $(x_t)_j$. Notice that with probability $1-\sum_{j=1}^m (x_t)_j$ the impression is not assigned to any advertiser, and in practice, such impressions will go to other traffic. The reward of an algorithm $A$ is given by $R(A|\vgamma) = \sum_{t=1}^T r_t^{\top} x_t + \lambda H(x_t)$, where $\lambda$ is the parameter of the entropy regularizer and  {\small$$H(x):=-\sum_{j=1}^m {x_j} \log(x_j) - \pran{1-\sum_{j=1}^m {x_j}}\log\pran{1-\sum_{j=1}^m{x_j}} $$}is the  entropy function of assignment probability $x$. A notable difference of this application is that the decisions are randomized. As a result, in the constraints, we need to take into account the actual realization of the probabilistic matching. Define the random variable
\begin{equation*}
    v_t = \left\{
        \begin{array}{cl}
        e_j     & \text{w.p. } x_j \\
        0     & \text{w.p. } 1-\sum_{j=1}^m x_j
        \end{array} \ ,
    \right.
\end{equation*}
where $e_j\in\RR^m$ is the $j$-th standard unit vector in $\RR^m$. The random variable $v_t$ characterizes the realized assignment of the impression at time $t$. Then, algorithm $A$ must satisfy $\sum_{t=1}^T v_t \le T \rho$. 

Finally, the hindsight problem is:
\begin{align}\label{eq:matching_prob}
\begin{split}
    \max_{x_t\in\cX} \  \sum_{t=1}^T r_t^{\top} x_t + \lambda H(x_t) \ \ \ 
    \text{s.t.} \  \sum_{t=1}^T x_t \le T \rho\ .\ \ \  
\end{split}
\end{align}
The resource constraint of \eqref{eq:matching_prob} is stated in terms of  the expected allocation $x_t$. As we argue in Appendix~\ref{sec:stochastic-functions}, this problem is a valid upper bound on the performance of every online algorithm. 


Invoking Algorithm~\ref{al:sg} with $f_t(x)=r_t^\top x + \lambda H(x)$ and $b_t(x)=x$, we obtain that the primal decision in \eqref{eq:primal_decision} can be computed in closed form as follows:
\begin{align*}
           (x_t)_j=\frac{\exp((r_t(j)-\mu_t(j))/\lambda)}{\sum_{l=1}^m\exp((r_t(l)-\mu_t(l))/\lambda)+1}\ 
\end{align*}
When we implement the algorithm, we account for the stochasticity of resource consumption by updating resources using the actual realization of the probabilistic matching $v_t$, i.e., we update the remaining capacity as $B_{t+1} = B_t - v_t$. Dual sub-gradients are computed using the probabilistic matching $x_t$, i.e, $\tg_t = - x_t + \rho$. {Our algorithm and analysis work even under stochastic resource consumption~(see Section~\ref{sec:stochastic-input} for a discussion).}


{We conclude by discussing the related literature.}
Recently, \citet{agrawal2018proportional} studied a multi-round \emph{offline} proportional matching algorithm for this problem setting.  Our algorithm leads to a simple \emph{online} counterpart to \citet{agrawal2018proportional} that yields similar regret/complexity bounds. \citet{dughmi2017bernoulli} introduced a dual-based online algorithm for proportional matching with stochastic input with a multiplicative weights update. Their algorithm, however, requires an estimate of the value of the benchmark. When the value of the benchmark is not known, an estimate can be obtained by solving a convex optimization problem. Our algorithm, in comparison, does not require knowing the value of the benchmark nor solving convex optimization problems. {We are not aware of any results for the other input models.}

\nnvspace
\subsection{Personalized Assortment Optimization with Limited Inventories}\label{sec:assortment}

Personalized product assortment/recommendation is nowadays a central problem faced by many online retailers (see, e.g., \citealt{bernstein2015dynamic, golrezaei2014real}).

We here consider a retailer with $m$ products and inventories $B=T\rho\in\RR^m_+$ that limit the amount of products to sell over $T$ time periods. In period $t$, one customer arrives searching for a product. The firm needs to decide, in real time, a subset of the product $S\subseteq \{1, \ldots, m\}$  to offer to the consumer, based on the inventory level, and the customer's personal preferences. The consumer then chooses a product (or not) according to a general choice model specifying the probability that a certain product is purchased from the assortment $S$ and the consumer personal information. The consumer's choice generates a revenue for the firm and consumes the firm's inventory. Here, we assume the choice model is known by the retailer; in practice, such a choice model can be learnt by a separate machine learning procedure. 

This assortment optimization problem is a special case of our online allocation problem. We utilize a $m$-dimensional binary variable $x_t\in \cX_t\subseteq \{0,1\}^m$ to represent the assortment $S$ for the $t$-th customer, where $\cX_t$ is the subset of products that satisfy the customer's search and $(x_t)_j = 1$ if product $j$ is included in the assortment. The set $\cX_t$ can encode constraints on the displayed assortment, e.g., assortments might be restricted to include a small number of products. Let $b_t(x)\in\RR^m$ capture the personalized information of the choice model for the $t$-consumer, namely, the $j$-th entry of $b_t(x)$ corresponds to the probability that the $t$-th customer purchases the $j$-th product given the assortment $x$. {A frequently-used probabilistic model on $b_t(x)$ is the multinomial logit model (MNL) with $(b_t(x))_j= e^{(\theta_t)_j}/ (1 + \sum_{i\in S} e^{(\theta_t)_i}) \mathbf 1\{ j \in S\}$ 
where $S=\{i\mid x_i=1\}$ be the support set of $x$, and $\theta_t \in \mathbb R^m$ are the parameters of the MNL model that can be learnt from a learning procedure~\citep{anderson1992discrete}. The expected revenue for this assortment is thus $f_t(x_t)=\sum_{j=1}^m (r_t)_j \cdot (b_t(x_t))_j = r_t^\top b_t(x_t)$, where $(r_t)_j$ is the revenue for the $j$-th product.} The offline problem is given by:
\begin{align*}
	\max_{y: y_t\in \Delta(\cX_t)} \  \sum_{t=1}^T \sum_{x \in \cX_t}  y_t(x) \cdot r_t^\top b_t(x) \ \ \     \text{s.t.}    \   \sum_{t=1}^T \sum_{x \in \cX_t} y_t(x) \cdot b_t(x) \le B\,,
\end{align*}
where the variable $y_t(x)$ quantifies the probability that assortment $x \in \cX_t$ is offered at time $t$. These variables satisfy $y_t(x) \ge 0$ and $\sum_{x \in \cX_t} y_t(x) = 1$. As in the proportional matching example, we write the constraint in terms of the expected resource consumption $b_t(x)$ instead of the realized choice of the consumer, which leads to a valid upper bound. In our algorithm, however, the resources are updated according to the true realized consumption (see Section \ref{sec:stochastic-input} for a discussion). 

Algorithm~\ref{al:sg} and its analysis can be directly applied to the personal assortment optimization  problem, and the primal update \eqref{eq:primal_decision} becomes
\begin{align}\label{eq:primal-update-assortment}
       x_t \in \arg\max_{x\in\cX_t} \left\{(r_t-\mu)^\top b_t(x) \right\} \ ,
   \end{align}
which results in an online gradient $g_t = - b_t(x_t) + \rho$ that can be used in the dual mirror descent update \ref{eq:dual_update}. We remark that the primal decision-making step~\eqref{eq:primal-update-assortment} might not be efficiently solved for arbitrary choice models when the number of products is large. Under the MNL choice model, however, the optimal assortment is revenue ordered and \eqref{eq:primal-update-assortment} can be solved  efficiently in polynomial time even under additional constraints on the displayed assortment~\citep{TalluriVanRyzin2004}.

We conclude by comparing our results on assortment optimization with the existing literature. Under stochastic i.i.d.~inputs, the algorithm proposed in~\cite{golrezaei2014real} yields a $\frac{4}{3}$-competitive ratio, while our algorithms attains a $(1+\varepsilon)$-competitive ratio (since the regret is vanishing). \cite{golrezaei2014real} also propose a different algorithm that can obtain $O(T^{1/2})$ regret following classic results on online allocation problems. These algorithms, however, require solving (at least one) large linear programmings with $2^m$ variables or constraints, which becomes impractical with state-of-the-art solvers when $m\ge 30$. In contrast, our proposed algorithms are more practical and efficient as they do not require solving auxiliary linear programs. When inputs are adversarial, the competitive ratio of \cite{golrezaei2014real} is sharper than ours. {This follows because our analysis does not assume any structure between reward and consumption, while they assume that the revenue vector $r$ is the same for all requests and their analysis takes advantage of the fact that the reward of an assortment is linear in the consumption (namely $f_t(x_t)= r^\top b_t(x_t)$).} As in the AdWords problem, we conjecture that no algorithm that attains vanishing regret under stochastic input can attain fixed competitive ratios under adversarial input. We are not aware of any results under the three non-stationary input models presented in Section~\ref{sec:non-stationary}.

\nvspace

{
\section{Extensions and Numerical Experiments}\label{sec:extension}
\nnvspace
\subsection{Beyond Mirror Descent}\label{sec:beyond-md}

For simplicity, we state our algorithm in terms of online mirror descent. While mirror descent is a general algorithmic framework that allows to recover many other popular algorithms used in practice, most of our results extend to other popular algorithms for online linear optimization such as regularized follow-the-leader~\citep{shalev2007primal}, Adagrad~\citep{duchi2011adaptive} or Adam~\citep{kingma2015adam}. As a result, Algorithm~\ref{al:sg} can be interpreted as a meta-algorithm that can use an online optimization algorithm as a black box to solve online allocation problems.

An algorithm for online linear optimization takes in each step an action $\mu_t \in \RR_+^m$ and incurs a linear cost $g_t^\top \mu_t$. We remark that the gradients $g_t \in \RR^m$ are observed \emph{after} taking the action and can be adversarially chosen. Then, for any $\mu \in \RR^+_m$, we denote the regret of the online algorithm by
\[
    E(G,T,\mu) = \sup_{g_t : \|g_t\|_\infty \le G} \left\{ \sum_{t=1}^T g_t^\top (\mu_t - \mu) \right\}\,.
\]
The regret $E(G,T,\mu)$ measures the worst-case performance over all possible gradients $g_t$ with norm bounded by $\|g_t\|_\infty \le G$ against a fixed static action $\mu$. We assume, without loss, that the regret it non-decreasing in $T$. In the case of stochastic input, we can prove the following result:

	\begin{cor} Consider a variant of Algorithm~\ref{al:sg} that uses an online linear optimization algorithm with regret guarantee $E(G,T,\mu)$ to update the dual variables $\mu_t$ with sub-gradients $g_t = -b_t(\tx_t) + \rho$. Suppose Assumptions~\ref{ass:p} holds and the requests come from an i.i.d.~model with unknown distribution. Then, it holds for any $T\ge 1$ that
		\begin{align*}
		\Regret{A}\le C_1 + \max \left\{ E(G,T,\mu) : \mu \in \{0, (\bar f/\rho_1) e_1,\ldots, (\bar f/\rho_m) e_m\}  \right\}\,,
		\end{align*}
		where $G = \ubb + \ubrho$ and $C_1$ is defined in Theorem~\ref{thm:master}.
	\end{cor}

When the regret of the underlying online linear optimization algorithm is $E(G,T,\mu) = O(T^{1/2})$ the previous result implies similar regret guarantees for Algorithm~\ref{al:sg}. Our results for the case of adversarial input and adversarial corruptions can be similarly extended. For the case of ergodic and periodic input, we require that the dual variables are stable in the sense that they do not change much from one step to the next. This is used to bound the extra term $\sum_{t=1}^{\tA-k} \| \mu_{t+k} - \mu_t \dualnorm$ that appears in the regret bound. Stability of the dual variables holds for online mirror descent as discussed Proposition~\ref{prop:diff-iterate} in the appendix and can be shown to hold for other algorithms too.

\nnvspace
\subsection{Unknown Length of Horizon}\label{sec:random-time}

Our model assumes that the number of requests $T$ in the horizon is known in advance. In general, when the number of requests is adversarially chosen, it is not possible to attain vanishing regret even when the input is stochastic. Our algorithm, however, can incorporate unknown stochastic horizons. In this case, we would run our algorithm by setting the target resource vector to be $\rho = B / \EE T$ in the computation of the gradients $g_t = -b_t(\tx_t) + \rho$, i.e., using the expected number of time periods in the target. When $T$ is a stopping time for the request sequence $\vgamma = (\gamma_t)_{t\ge 1}$, we can use a similar regret analysis as the one in Theorem~\ref{thm:master} and obtain the same regret bound with the exception of an extra term $\bar f \cdot \EE[ \max(0, \EE T - T)]$ in the right-hand side of \eqref{eq:master}. This term can be upper bounded by $\bar f \cdot \text{Var}(T)^{1/2}$, which yields sublinear regret in the expected number of time periods $\EE T$, for example, when requests arrive according to a Poisson process.

\nnvspace

\subsection{Stochastic Reward and Resource Consumption}\label{sec:stochastic-input}

Sometimes, the reward and consumption are random and realized after the decision maker chooses an $x_t$. This is the case in proportional matching (Section~\ref{sec:matching}) and online assortment (Section~\ref{sec:assortment}). Our algorithm and analysis extend to settings in which the reward and resource consumption are stochastic given an action by making decisions based on the \emph{expected} reward and consumption. More formally, let $\vec \zeta$ denote the random variable determining the realization of the above process. Then, our theorems still hold after taking the expectation over $\vec \zeta$. We state and prove the result for the case of stochastic i.i.d.~input; results for other input models follow mutatis mutandis. We remark that, in this case, $\OPT(\vgamma)$ needs to be redefined in terms of achievable expected rewards and consumption pairs to provide a tight upper bound (see Appendix~\ref{sec:stochastic-functions} for details).

\begin{prop}\label{prop:matching}
	Consider Algorithm \ref{al:sg} with stochastic resource consumption, step-size $\eta \ge 0$ and initial solution $\mu_1\in \RR_+^m$. Suppose Assumptions~\ref{ass:p}-\ref{ass:h} are satisfied and the requests come from an i.i.d.~model with unknown distribution. Then, it holds for any $T\ge 1$ that
		\begin{align*}
		\begin{split}
		\sup_{\cP \in \Delta(\cS)}  \left\{ \EE_{\vgamma \sim \cP^T, \vec \zeta} \left[ \OPT(\vgamma) - R(A|\vgamma, \vec \zeta) \right] \right\}
		\le C_1 + C_2 \eta T + \frac {C_3}{\eta} \,,
		\end{split}
		\end{align*}
		where the constants $C_1, C_2, C_3$ are defined in Theorem~\ref{thm:master}.
\end{prop}

\nnvspace

\subsection{Approximately Solving the Sub-Problem \eqref{eq:primal_decision}}\label{sec:approximation}

In practice, the observed reward and consumption functions often come from machine learning models, which can be noisy and inexact. When using such noisy inputs in  sub-problem~\eqref{eq:primal_decision}, the obtained solution can be viewed as an approximated solution to the underlying true statistical model. Furthermore, solving the sub-problem \eqref{eq:primal_decision} exactly can sometimes be expensive, in particular when the sub-problem is non-convex, and an approximation algorithm is used. Interestingly, our algorithms and analysis are robust to inexact solutions to the primal sub-problem. 

Suppose the sub-problem is solved with additive error $\epsilon_t$, i.e., the reward collected at time $t$ verifies $f_t(x_t) - \mu_{t}^{\top} b_t(x_t) \ge f_t^*(\mu_t) -\epsilon_t$. In the analysis under stochastic i.i.d.~input (other settings follow through a similar argument), the only place we use \eqref{eq:primal_decision} is to show that the reward collected by the algorithm at time $t$ satisfies $f_t(x_t) - \mu_t^\top b_t(x_t) = f_t^*(x_t)$ (used in \eqref{eq:primal-approximation}). All other steps in the analysis of Theorem~\ref{thm:master} follow, and the errors $\epsilon_t$ would appear additively in the right-hand side of the regret bound \eqref{eq:master}, i.e., $\Regret{A}\le O(T^{1/2})+\sum_{t=1}^T \epsilon_t$ with properly chosen step-size.  If the cumulative errors are small enough, i.e., $\sum_{t=1}^T \epsilon_t = O({T^{{1/2}}})$, then we can still obtain $\Regret{A} = O(T^{1/2})$ in Theorem \ref{thm:master}. Obtaining additive errors in regret bounds is the best one can hope for in online algorithms when inexact solutions are available.

{Similarly, suppose the sub-problem is solved with multiplicative error $\alpha$ (this often happens when an approximation algorithm is used to solve the sub-problem), i.e., $\alpha(f_t(x_t) - \mu_{t}^{\top} b_t(x_t)) \ge f_t^*(\mu_t)$. In the analysis under stochastic i.i.d.~input (again, other settings follow through a similar argument),
we can simply replace \eqref{eq:primal-approximation} with $\alpha f_t(x_t) \ge f_t^*(\mu_t)+\alpha \mu_{t}^{\top} b_t(x_t)$, and the other steps in the analysis of Theorem~\ref{thm:master} follow with the value of $\alpha$ carried over. This shows that Algorithm \ref{al:sg} with proper step-size is asymptotic $\alpha$-competitive, i.e.,  $\EE_{\vgamma\sim \cP^T}[\OPT(\vgamma)-\alpha R(A|\vgamma)]\le O(\sqrt{T})$. 


}

\nnvspace
\subsection{Numerical Experiments}

In Appendix~\ref{sec:numerical}, we present numerical experiments for our algorithms on online linear programming with stochastic i.i.d.~inputs  (Appendix \ref{sec:num-olp}), and on proportional matching with ergodic inputs (Appendix \ref{sec:num-pm}). The experiments on online linear programming (Appendix \ref{sec:num-olp}) verify the theoretical dependence of regret for Algorithm~\ref{al:sg} over the time horizon $T$, resource dimension $m$ and primal decision dimension $d$. 
They show that online gradient descent (with dual update \eqref{eq:ogd}) and multiplicative weights update (with dual update \eqref{eq:mwu}) have $\tilde{O}(\sqrt{mT})$ regret, while multiplicative weights update with projection (Appendix \ref{sec:log-dependence}) has $\tilde{O}(\sqrt{T})$ regret. Furthermore the regret of all three algorithms are independent from the primal decision dimension $d$. These findings are consistent with Theorem~\ref{thm:master} and the discussions in Appendices~\ref{sec:reference-functions}~and~\ref{sec:log-dependence}. The experiments on proportional matching verify the $\tilde{O}(\sqrt{T})$ regret bound of Algorithm \ref{al:sg} under ergodic inputs, which is consistent with Theorem~\ref{thm:ergodic}. 


}

\nvspace

\section{Conclusion and Future Directions}

In this paper, we present a class of simple and robust algorithms for online allocation problems {with non-linear reward functions, non-linear consumption functions, and potentially integral decision variables}. We show that our algorithms attain vanishing regret under stochastic i.i.d.~and non-stationary inputs, and fixed competitive ratios under adversarial inputs. The performance of our algorithms, moreover, is shown to be optimal across various input models. Our algorithms are oblivious to the input model in the sense that they obtain good performance without knowing the type of input they are facing. {We discuss applications to online linear programming, bidding in repeated auctions with budgets, online matching with high entropy, and personalized assortment optimization with limited inventories. Our algorithms, in many cases, give new results or match/improve the performance of existing algorithms in the literature.}

An interesting future research direction is to explore whether better bounds can be obtained under more restrictive assumptions on the inputs. For example, when the input is adversarial, it is worth studying whether better competitive ratios can be obtained when the reward and resource consumption consumption are related to each other. Alternatively, when the input is stochastic, it is interesting to determine whether better regret bounds can be obtained when the expected dual function is better behaved.

\section*{Acknowledgement} The authors would like to thank Balasubramanian Sivan, Rad Niazadeh, and Shipra Agrawal for useful feedback provided. The authors thank the review team for thoughtful comments that strengthen the paper.

\setstretch{1}
\bibliographystyle{plainnat}
{\small{\bibliography{references,Lu-papers}}}

\newpage
\appendix

\section{Proofs in Section \ref{sec:formulation}}

\subsection{Proof of Proposition \ref{prop:upper_bound}}
\begin{proof}
It holds for any $\mu\in\cD$ that
\begin{align*}
\begin{split}
 \OPT(\vgamma)
= &  \mpran{
	\begin{array}{cl}
	\max_{x_t\in \cX} & \sum_{t=1}^T f_t(x_t)    \\
	\text{s.t.}     & \sum_{t=1}^T b_t(x_t) \le T \rho
	\end{array}}\\
\le &  \max_{x_t\in \cX} \left\{\sum_{t=1}^T f_t(x_t)+ T\mu^{\top} \rho - \mu^{\top} \sum_{t=1}^T  b_t(x_t) \right\} \\
= &  \sum_{t=1}^T f_t^*(\mu)  +T \rho^\top \mu \\
= & D(\mu | \vgamma ) \ ,
\end{split}
\end{align*}
where the first inequality is because we relax the constraint $\sum_{t=1}^T b_t(x_t) \le T \rho$ and $\mu \ge 0$, and the last equality utilizes the definition of $f^*$.
\end{proof}



\section{Regret Bounds for Example Reference Functions}\label{sec:reference-functions}

We now discuss the regret bounds for the sample reference functions presented in the paper and their dependence on the number of resources.

\subsection{Online Gradient Descent}

Recall that for online gradient descent algorithm the reference function is $h(\mu) = \frac 1 2 \|\mu\|_2^2$. First, note that the constant $C_1$ is in general independent of the number of resources. Second, for the constant $C_2$ we use that $h(\mu)$ is $(1/m)$-strongly convex over $\RR_+^m$ with respect to the $\|\cdot\dualnorm$ to obtain that $C_2 = m (\ubb + \ubrho)^2 / 2$. Third, for the constant $C_3$ we obtain, by choosing the initial point to be $\mu_1 = 0$, that
\[
    V_h(\mu,\mu_1) = h(\mu) - h(\mu_1) - \nabla h(\mu_1)^\top (\mu - \mu_1) = h(\mu) = \frac 1 2 \|\mu\|_2^2\,,
\]
because $h(0) = 0$ and $\nabla h(0) = 0$. Therefore, 
\[
    C_3 = \max \left\{  V_h(\mu, \mu_1) : \mu \in \{0, (\bar f/\rho_1) e_1,\ldots, (\bar f/\rho_m) e_m\}  \right\} = \ubf^2 / (2 \lbrho^2),
\]
and, as result, the constant is independent of the number of resources. Putting everything together, we obtain that the regret bound is
\[
\Regret{A} \le \frac{\ubf \ubb}{\lbrho} + \frac {\ubf (\ubb + \ubrho)} {\lbrho} \sqrt{m T}\,.
\]
We remark that in practice it is better to choose the reference function to be a squared weighted-$\ell_2$-norm. Weighting the norm yields better condition number as it allows us correct for different scales in the right-hand side of the resource constraints. In particular, one practically appealing choice is the reference function to
    \[
    h(\mu) = \frac 1 2  \sum_{j=1}^{m} (\rho_j \mu_j)^2\,.
    \]
This is equivalent to normalizing the resource consumption function according to $b_j(x):=b_j(x) / \rho_j$.
    
\subsection{Multiplicative Weights Algorithm}\label{sec:mwu-p}

For the multiplicative weights update algorithm the reference function is $h(\mu) = \sum_{j=1}^m \mu_j \log(\mu_j)$. First, note that the constant $C_1$ is in general independent of the number of resources. Second, for the constant $C_2$, we do not have that $h(\mu)$ is strongly convex over $\RR_+^m$. However, using Proposition~\ref{thm:stopping_time} we can show that the dual variables produced by the algorithm always remain in the box $\mathcal D = [0,\mumax_1] \times \ldots \times [0,\mumax_m]$ with $\mumax_j = \bar f/\rho_j + 1$ because the reference function is separable over resources. (For this result to hold, we need the step-size to be sufficiently small, which is always true for large enough lengths of the horizon $T$.) Note that the univariate function $\mu_j \log(\mu_j)$ is $(\mumax_j)^{-1}$-strongly convex over $[0, \mumax_j]$ because its second derivative $1/\mu_j$ is monotonically decreasing. Therefore, if we restrict the dual variables to the box $\mathcal D$, we obtain that $h(\mu)$ is $(\max_j \mumax_j)^{-1}$-strongly convex with respect to the $\|\cdot\|_2$ norm and, as a result, $(m\max_j \mumax_j)^{-1}$-strongly convex with respect to the $\|\cdot\dualnorm$ norm. This implies that the second constant is given by
\[
    C_2 = \frac 1 2 m \max_j \mumax_j (\ubb + \ubrho)^2\,.
\]
Third, for the constant $C_3$, we can write the Bregman divergence as
\[
    V_h(\mu,\mu_1) = h(\mu) - h(\mu_1) - \nabla h(\mu_1)^\top (\mu - \mu_1) = \sum_{j=1}^m \mu_j \log\left( \frac {\mu_j} {(\mu_1)_j} \right) - \sum_{j=1}^m \mu_j + \sum_{j=1}^m (\mu_1)_j \,.
\]
By choosing the initial point to be $\mu_1 = e/m$, which lies in the box $\mathcal D$, we obtain that $V_h(0,\mu_1) = 1$ and $V_h((\bar f/\rho_j) e_j,\mu_1) = (\ubf / \rho_j) \cdot ( \log(\ubf/\rho_j) + \log(m) - 1)  + 1$. Therefore, $C_3 \le (\ubf / \ubrho) \cdot ( \log(\ubf/\ubrho) + \log(m) - 1)  + 1$. As a result, we obtain that the regret bound is given by
\[
\Regret{A} \le \frac{\ubf \ubb}{\lbrho} + (\ubb + \ubrho) \sqrt{2 \left(\frac{\ubf}{\lbrho} + 1\right) \cdot \left( \frac{\ubf}{\lbrho} \cdot \left( \log\left(\frac{\ubf}{\lbrho}\right) + \log(m) - 1\right)  + 1\right) m T}\,.
\]
When $\ubf / \ubrho \ge 1$, the bound simplifies to
\[
\Regret{A} \le \frac{\ubf \ubb}{\lbrho} + 2 \frac{ \ubf (\ubb + \ubrho)}{\ubrho} \sqrt{( \log(\ubf/\ubrho) + \log(m) ) m T}\,.
\]

\section{Optimal Dependence on the Number of Resources}\label{sec:log-dependence}

In this section, we discuss how to obtain regret bounds with an optimal dependence on the number of resources when a uniform upper bound $\ubf$ on the reward function is known. The analysis leverages an idea pioneered by \cite{AgrawalDevanue2015fast}. 

Recall that in the analysis of our algorithm, we choose the pivot to be $\mu \in \{0, (\bar f/\rho_1) e_1,\ldots, (\bar f/\rho_m) e_m\}$. The convex hull of these points is a version scaled of the unit simplex, i.e., 
\[
    \mathcal D = \left\{ \mu \in \RR^m_+ : \sum_{j=1}^{m} \mu_j / \rho_j \le \bar f \right\}\,.
\]
Therefore, without loss, we can constraint the dual variables generated by our algorithm to lie in the set $\mathcal D$. By constraining the dual variables to lie in $\mathcal D$ we can take advantage of the fact that the multiplicative weights update algorithm can attain a $\log^{1/2}(m)$ dependence on the number of resources over the simplex. Using this stronger regret guarantees for the multiplicative weights update algorithm, we can obtain an optimal dependence of $\log^{1/2}(m)$ on the number of resources.

Without loss of generality we assume that $\rho_j = 1$ for every resource $j$ (this can be achieved, for example, by normalizing the resource consumption function as follows $b_j(x) := b_j(x) / \rho_j$). We consider the negative entropy reference function $h(\mu) = \sum_{j=1}^{m} \mu_j \log(\mu_j)$. The algorithm is exactly as the one stated in the main body of this paper with the difference that we restrict the dual variables to lie in $\mathcal D$ in the mirror descent step. That is, we have now have
\begin{equation*}
   \mu_{t+1} = \arg\min_{\mu \in \mathcal D} \left\{ \tg_t^\top \mu + \frac{1}{\eta} V_h(\mu,\mu_t) \right\} \ .
\end{equation*}
Let $\tilde \mu_{t+1}$ be the dual update without the projection, i.e., $\tilde \mu_{t+1} = {\mu_t} * \exp(-\eta \tg_t)$. Then, the dual update is obtained by projecting back to the scaled unit simplex:
\[
    \mu_{t+1} =
    \begin{cases}
        \tilde \mu_{t+1} & \text{if } \mathbf 1^\top \tilde \mu_{t+1} \le \bar f\\
        \bar f \frac{\tilde \mu_{t+1}} {\mathbf 1^\top \tilde \mu_{t+1}} & \text{otherwise}.
    \end{cases}
\]

The same regret bound for our algorithm applies because the pivot lies in $\mathcal D$ and the regret guarantee from online mirror descent given in Proposition~\ref{prop:omd} holds for any convex set $\mathcal D$. Using the optimal choice for the step size, which is $\eta = \sqrt{C_2/ (C_3 T)}$, yields the regret bound
\[
\Regret{A}\le C_1 + 2 \sqrt{C_2 C_3 T}\,.
\]
where $C_1 = \ubf \ubb  / \lbrho$, $C_2 = (\ubb + \ubrho)^2/2\sigma$, $C_3= \max \left\{  V_h(\mu, \mu_1) : \mu \in \{0, (\bar f/\rho_1) e_1,\ldots, (\bar f/\rho_m) e_m\}  \right\}$. We now discuss the dependence of the constants $C_1$, $C_2$, and $C_3$ on the number of resources.

First, note that the constant $C_1$ is in general independent of the number of resources. Second, for the constant $C_2$ we have that $h(\mu)$ is $\ubf^{-1}$-strongly convex with respect to the $\|\cdot\|_1$ norm over $\mathcal D$ by the following lemma. This implies that $C_2$ is also independent of the number of resources.

\begin{lem} The negative entropy $h(\mu) = \sum_{j=1}^{m} \mu_j \log(\mu_j)$ is $1/C$-strongly convex with respect to the $\|\cdot\|_1$ norm over $\mu > 0$ with $\sum_{j=1}^m \mu_j \le C$.
\end{lem}
\begin{proof}
Let $\mathcal D = \{ x \in \RR^m : x > 0 \text{ and } \sum_{j=1}^{m} x_j \le C\}$. By Proposition 3.1 of \cite{beck2003mirror}, it is enough to show that $\left( \nabla h(x) - \nabla h(y)\right)^\top (x - y) \ge C^{-1} \| x - y \|_1^2$ for every $x,y \in \mathcal D$. We have that
\begin{align*}
    \left( \nabla h(x) - \nabla h(y)\right)^\top (x - y) &= \sum_{j=1}^m (x_j - y_j) \log\left( \frac{x_j}{y_j}\right)\\
    &\ge \sum_{j=1}^{m} 2 \frac{(x_j - y_j)^2}{x_j + y_j}\\
    &= 2 \left(\sum_{j=1}^m (x_j + y_j)\right) \sum_{j=1}^{m} \frac{x_j + y_j} { \sum_{j=1}^m \left(x_j + y_j\right)} \left( \frac{x_j - y_j}{x_j + y_j}\right)^2\\
    &\ge2 \left(\sum_{j=1}^m (x_j + y_j)\right) \left( \sum_{j=1}^{m} \frac{x_j + y_j} { \sum_{j=1}^m \left(x_j + y_j\right)} \frac{|x_j - y_j|}{x_j + y_j} \right)^2\\
    &=2 \left(\sum_{j=1}^m (x_j + y_j)\right)^{-1} \left( \sum_{j=1}^{m} |x_j - y_j| \right)^2
    \ge C^{-1} \|x - y\|_1^2\,,
\end{align*}
where the first equality follows from the definition of the reference function $h(x) = \sum_{j=1}^m x_j \log(x_j)$, the first inequality follows from the first step of the proof of Proposition 5.1 in \cite{beck2003mirror},  the second inequality from Jensen's inequality with $p_j = (x_j + y_j)/ \sum_{j=1}^m (x_j + y_j)$ because the quadratic function is convex and $\sum_{j=1}^m p_j = 1$, and the last inequality follows from the definition of the $\ell_1$ norm and because $x,y \in \mathcal D$.
\end{proof}

Finally, for the constant $C_3$ recall that the Bregman divergence is given by
\[
    V_h(x,y) = h(x) - h(y) - \nabla h(y)^\top (x-y) = \sum_{j=1}^m x_j \log\left( \frac {x_j} {y_j} \right) - \sum_{j=1}^m x_j + \sum_{j=1}^m y_j\,.
\]
Choosing the initial point to be $\mu_1 = \bar f e/m \in \mathcal D$ we obtain that $V_h(\bar f e_i, \bar f e / m) = \bar f \log(m)$ and $V_h(0, \bar f e / m) = \bar f$. Therefore, we obtain that
\[
    C_3 = \bar f \max(1, \log(m))\,.
\]
Putting everything together, the final regret bound (assuming that $m>2$) is
\[
    \Regret{A}\le C_1 + 2 (\ubb + \ubrho) \bar f \sqrt{\log(m) T}\,.
\]

\section{Proofs in Section \ref{sec:adversarial}}

\subsection{Proof of Theorem \ref{thm:adversial-no-regularizer}}

We here discuss how the steps of the proof of Theorem~\ref{thm:master} need to be adapted to account for the adversarial requests. 

\textbf{Step 1 (Primal performance).} Fix a vector of requests $\vgamma \in \cS^T$ and let $x_t^* \in \cX_t$ be an optimal action of the $\OPT(\vgamma)$ at time $t$. Notice that $x_t\in \arg\max_{x\in \cX} \{f_t(x)-\mu_t^\top b_t(x)\}$, thus we have $f_t(x_t) \ge f_t(x_t^*) - \mu_{t}^{\top} \pran{ b_{t} (x_t^*) - b_{t} (x_t)}$ and $0 = f_t(0) \le f_t(x_t) - \mu_t^\top b_t(x_t)$, whereby
\begin{align}\label{eq:alpha-time}
\begin{split}
    \alpha f_t(x_t) 
     &= f_t(x_t) + (\alpha-1)(f_t(x_t)) \\
    & \ge  f_t(x_t^*) + \mu_t^{\top} b_t(x_t) -\mu_t^{\top} b_t( x_t^*) + (\alpha-1)(\mu_t^\top b_t(x_t)) \\
     &= f_t(x_t^*)  - \alpha \mu_t^{\top} ( \rho-b_t(x_t)) + \alpha \mu_t^{\top} \rho - \mu_t^{\top} b_t(x_t^*)\\
     &\ge f_t(x_t^*)  - \alpha \mu_t^{\top} ( \rho-b_t(x_t)) \ .
\end{split}
\end{align}
where the second inequality is because $\alpha \mu_t^{\top} \rho - \mu_t^{\top} b_t(x_t^*)\ge 0$ due to the definition of $\alpha$ and the fact that $\mu_t\ge 0$. Summing up \eqref{eq:alpha-time} over $t=1,\ldots,\tA$, we arrive at
\begin{align}\label{eq:target}
    \begin{split}
        \alpha \sum_{t=1}^{\tau_A} f_t(x_t)
        &\ge \sum_{t=1}^{\tau_A} f_t(x_t^*)  - \alpha \sum_{t=1}^{\tau_A}  \mu_t^{\top} ( \rho-b_t(x_t)) \ .
    \end{split}
\end{align}

\paragraph{Step 2 (Complementary slackness).} Denoting, as before, $w_t(\mu) = \mu^\top(\rho - b_t(x_t))$, this step applies directly because the analysis is deterministic in nature:
\begin{align}\label{eq:bound-cs-2}
\begin{split}
\sum_{t=1}^{\tA} w_t(\mu_t)
\le \sum_{t=1}^{\tA} w_t(\mu) + E(T,\mu)\,,
\end{split}
\end{align}
where $E(t,\mu)$ is the regret of the online algorithm as specified in~\eqref{eq:regret_omd}.


\paragraph{Step 3 (Putting it all together).}


Choosing $\mu \ge 0$ gives
\begin{align*}
    \OPT(\vec \gamma) - \alpha R(A | \vec \gamma) 
    & \le  \sum_{t=1}^{T} f_t(x_t^*)  - \alpha \sum_{t=1}^{\tA} f_t(x_t)\\
    & \le  \sum_{t=\tA + 1}^{T} f_t(x_t^*)  +  \alpha \sum_{t=1}^{\tA} w_t(\mu_t)  \\
    & \le \sum_{t=\tA + 1}^{T} f_t(x_t^*)   + \alpha \sum_{t=1}^{\tA} w_t(\mu) + \alpha E(T, \mu)
      \\
    & \le (T-\tA) \cdot  \ubf - \alpha \sum_{t=1}^{\tA} \mu^\top (b_t(x_t) - \rho) + \alpha E(T,\mu)\,,
\end{align*}
where the first inequality follows because $\tA \le T$ together with $f_t(\cdot) \ge 0$, the second inequality is from \eqref{eq:target}, the third inequality utilizes \eqref{eq:bound-cs-2}, and the last inequality utilizes $f_t(x^*_t) \le \ubf$.

If $\tA = T$, then set $\mu = 0$, and the result follows. If $\tA < T$, then there exists a resource $j\in[m]$ such that $\sum_{t=1}^{\tA} (b_t(x_t))_j + {\ubbinfty} \ge T \rho_j$. Set $\mu = (\ubf/(\alpha \rho_j)) e_j$ with $e_j$ the unit vector and repeat the steps of the stochastic i.i.d.~case to obtain:
\[
    \OPT(\vec \gamma) - \alpha R(A | \vec \gamma) \le \frac{\ubf \ubbinfty}{\rho_j} + \alpha E(T,\mu)\le \frac{\ubf\ubbinfty}{\lbrho} + \alpha E(T,\mu)\,,
\]
which finishes the proof by substituting the regret bound of online mirror descent (see Proposition \ref{prop:omd}) and using that {\color{black} $\mu \in \{0, (\bar f/(\alpha \rho_1)) e_1,\ldots, (\bar f/(\alpha/\rho_m)) e_m\}$ to bound the term $V_h(\mu,\mu_1)$ in $E(T,\mu)$}. \qed

\subsection{Proof of Proposition \ref{thm:stopping_time}}\label{sec:proof_st}
Define $\mumax_j:=\frac{{\ubf}}{\rho_j}+1$. The key step in the proof of Proposition \ref{thm:stopping_time} is the following lemma, which shows that the dual update \eqref{eq:dual_update} never exceeds the upper bound $\mumax$ when the step-size $\eta$ is small enough.
\begin{lem}\label{lem:dual_update} Fix $\mu \in \cD$.
Let $\tg= \rho - b(\tx)$ with $(b,f,\cX)\in \cS$, $\tx \in \arg\max_{x\in \cX}f(x) - \mu^\top b(x)$, and $\mu^+ = \arg\min_{\hmu\ge 0} \tg^\top \hmu + \frac{1}{\eta} V_h(\hmu, \mu)$. Suppose $\mu\le\mumax$ and $\eta\le \frac{\sigma_2}{\ubbinfty}$, then it holds that $\mu^+\le \mumax$.
\end{lem}
{\bf Proof.}
Denote by $J:=\{j \in [m] \mid \mu^+_j>0\}$ the set of indices with positive dual variables. Then, we just need to show $\mu^+_j\le \mumax_j$ for any $j\in J$. Following the update rule \eqref{eq:dual_update}, it holds for any $j\in J$ that 
\begin{equation}\label{eq:l2proof}
    \ddh_j (\mu_j^+) = \ddh_j (\mu_j) - \eta \tg_j =  \ddh_j (\mu_j) + \eta b_j(\tx) - \eta\rho_j.
\end{equation}



Define $h^*_j(c)=\max_{\mu_j} \{c \mu_j - h_j(\mu_j)\}$ as the conjugate function of $h_j(\mu_j)$, then by Assumption~\ref{ass:h-sep} it holds that $h^*_j(\cdot)$ is a $\frac{1}{\sigma_2}$-smooth univariate convex function~\citep{kakade2009duality}. Furthermore, $\ddh^*_j(\cdot)$ is increasing, and $\ddh^*_j(\ddh_j(\mu_j))=\mu_j$. 

By Assumption~\ref{ass:p} and using that $0\in\cX$ is feasible, it holds that $0=f(0)\le f(\tx)-\mu^\top b(\tx)\le \ubf-\mu^\top b(\tx)$, whereby $\mu^\top b(\tx)\le \ubf$. Since $\mu\ge 0, b(x)\ge 0, \tx\in\cX\subseteq \RR^d_{+}$, it holds for any $j\in J$ that $b_j(\tx) \le \frac{\ubf}{\mu_j}$. Meanwhile, it follows by the definition of $\ubbinfty$ that $b_j(\tx)\le \ubbinfty$. Together with \eqref{eq:l2proof}, it holds that
\begin{equation}\label{eq:ba_prox}
    \ddh_j (\mu_j^+) \le  \ddh_j (\mu_j) + \eta \min\pran{\frac{\ubf}{\mu_j}, \ubbinfty} - \eta\rho_j.
\end{equation}
If $\frac{\ubf}{\rho_j}\le \mu_j\le \mumax_j$, we have $\min\pran{\frac{\ubf}{\mu_j}, \ubbinfty} - \rho_j\le 0$, thus it holds that $\mu^+_j\le \mu_j\le\mumax_j$ by utilizing \eqref{eq:ba_prox} and convexity of $\ddh_j$. Otherwise, $\mu_j\le \frac{\ubf}{\rho_j}$, and furthermore,
\begin{align*}
    \mu^+_{j}&=\ddh^*_j(\ddh_j(\mu_j^+)) \le \ddh^*_j(\ddh_j(\mu_j)+ \eta\ubbinfty) \\
    &\le \ddh^*_j(\ddh_j(\mu_j)) + \frac{\eta\ubbinfty}{\sigma_2}\le \frac{\ubf}{\rho_j} + 1 = \mumax_j \ ,
\end{align*}
where the first inequality is from \eqref{eq:ba_prox} and the monotonicity of $\ddh^*_j(\cdot)$, the second inequality is from $\ddh^*_j(\ddh_j(\mu_j))=\mu_j$ and the $\frac{1}{\sigma_2}$-smoothness of $h^*_j(\cdot)$, the last inequality utilizes $\eta\le \frac{\sigma_2}{\ubbinfty}$, and the last equality follows from the definition of $\mumax$. This finishes the proof of Lemma \ref{lem:dual_update}. \qed

\vspace{0.2cm}

{\bf Proof of Proposition \ref{thm:stopping_time}:}
First, a direct application of Lemma~\ref{lem:dual_update} shows that for any $t$, $\mu_t\le\mumax$. Next, it follows by the definition of $\tA$ that there exist $j$ such that $\sum_{t=1}^{\tA} (b_t(x_t))_j + \ubb \ge \rho_j T $. By the definition of $g_t$, we have
\begin{equation*}
    \sum_{t=1}^{\tA} (g_t)_j = \rho_j \tA - \sum_{t=1}^{\tA} (b_t(x_t))_j \le \rho_j \tA - \rho_j T + \ubb \ ,
\end{equation*}
thus 
\begin{equation}\label{eq:stop_time_eq1}
    T-\tA \le \frac{\ubb - \sum_{t=1}^{\tA} (g_t)_j}{\rho_j}.
\end{equation}

On the other hand, it follows the update rule \eqref{eq:dual_update} that for any $t\le \tA$,
$$\ddh_j((\mu_{t+1})_j)\ge \ddh_j((\mu_{t})_j)-\eta (g_t)_j\ .$$
Thus,
\begin{align}\label{eq:stop_time_eq2}
    \begin{split}
        \sum_{t=1}^{\tA} -(g_t)_j & \le \frac{1}{\eta} \left( \ddh_j((\mu_{\tA+1})_j)-\ddh_j((\mu_{0})_j)\right) \\
        &\le \frac{1}{\eta}\left( \ddh_j(\mumax_j)-\ddh_j((\mu_{0})_j)\right) \ ,
    \end{split}
\end{align}
where the last inequality is due to the monotonicity of $\ddh_j(\cdot)$. Combining \eqref{eq:stop_time_eq1} and \eqref{eq:stop_time_eq2}, we reach
$$
T-\tA \le \max_j \left\{ \frac{\ddh_j(\mumax_j)-\ddh_j((\mu_{0})_j)}{\eta \rho_j} + \frac{{\ubbinfty}}{\rho_j} \right\} \ .
$$
We conclude the proof by noticing that $\rho_j\ge \lbrho$ and using that $\ddh_j(\mumax_j)-\ddh_j((\mu_{0})_j) \le \themax$. \qed

\section{Proofs in Section \ref{sec:non-stationary}}

\subsection{Proof of Theorem \ref{thm:independent}}
We discuss how the steps of the proof of Theorem~\ref{thm:master} need to be adapted to account for the time dependence of the distributions of requests. 

\paragraph{Step 1.} Repeating the argument in step 1, we obtain as in equation~\eqref{eq:f_and_r} that
\begin{align}\label{eq:f_and_r_robust}
\mathbb E\left[ \sum_{t=1}^{\tA} f_t(x_t)  \right] 
&= 
\mathbb E\left[ \sum_{t=1}^{\tA} D(\mu_t| \cP_t) \right ]  - \mathbb E\left[ \sum_{t=1}^{\tA} \mu_t^\top \left(  \rho - b_t(x_t) \right) \right]   
\end{align}
where $D(\mu_t| \cP_t) := \EE_{(f,b,\cX) \sim \cP_t} \left[ f^*(\mu_t) \right] +\mu_t^{\top}$ is the dual function when requests are distributed from $\cP_t$. Denoting by $\bar \cP = \frac 1 T \sum_{s=1}^T \cP_s$ the time-averaged distribution of requests, we have that
\begin{align}\label{eq:bound_dual_robust}
    \sum_{t=1}^{\tA} D(\mu_t| \cP_t) &= \sum_{t=1}^{\tA} \left( \EE_{(f,b) \sim \cP_t} \left[ f^*(\mu_t) \right] +\mu_t^{\top} \rho \right) \nonumber\\
    &= \sum_{t=1}^{\tA} D(\mu_t| \bar \cP) + \sum_{t=1}^{\tA} \left( \EE_{(f,b) \sim  \cP_t} \left[ f^*(\mu_t) \right] - \EE_{(f,b) \sim \bar \cP} \left[ f^*(\mu_t) \right] \right) \nonumber\\
    &\ge \tA D(\bar \mu_{\tA}| \bar \cP) - \ubf  \MD(\cP)\,,
\end{align}
where the inequality follows from denoting $\bmu_{\tA} = \frac 1 {\tA} \sum_{t=1}^{\tA} \mu_t$ to be the average dual variable and using that the dual function is convex and because $f^*(\mu)=\max_{x\in\cX}\{f(x)-\mu^\top b(x)\}$ satisfies $0 \le f^*(\mu) \le \ubf$ by Assumption~\ref{ass:p}. {The lower bound $0 \le f^*(\mu)$ follows because $0 \in \cX$ is a feasible choice, $f(0) = 0$ and $b(0)=0$, while the upper bound $f^*(\mu) \le \ubf$ follows because $f(x) \le \ubf$, $\mu \ge 0$, and $b(x) \ge 0$.}

\paragraph{Step 2.} This step applies directly because the analysis is deterministic in nature.

\paragraph{Step 3.} With some abuse of notation, define $\OPT(\vec \cP) := \EE_{\vec \gamma \sim \vec \cP} [ \OPT(\vec \gamma) ]$. Proposition~\ref{prop:upper_bound} implies that for every $\mu \in \cD$ we have 
\[
    \OPT(\vec \cP) = \EE_{\vec \gamma \sim \vec \cP} [ \OPT(\vec \gamma) ] \le 
    \EE_{\vec \gamma \sim \vec \cP} [ D(\mu | \vgamma) ] = 
    \sum_{t=1}^T \left( \EE_{(f,b,\cX) \sim \cP_t} \left[ f^*(\mu) \right] + \rho^\top \mu \right) = T  D(\mu| \bar \cP)\,,
\]
where the second equation follows by the linearity of expectation and the last because $\bar \cP$ is the time-averaged distribution of requests. Therefore, for any distributions $\cP_t \in \Delta(\cS)$ for $t \in [T]$ and $\tA \in [0,T]$ we have that
\begin{align}\label{eq:bound-opt-robust}
\OPT(\vec \cP) &=  \frac{\tA}{T} \OPT(\vec \cP) + \frac{T-\tA}{T}\OPT(\vec \cP) \le  \tA D(\bmu_{\tA}| \bar \cP) + \pran{T-\tA}{\ubf}\ ,
\end{align}
where the inequality uses that $\OPT(\cP)\le T \ubf$. Combining  \eqref{eq:f_and_r_robust}, \eqref{eq:bound_dual_robust} and \eqref{eq:bound-opt-robust}, it holds that
\begin{align*}
    \begin{split}
        \EE \mpran{R(A|\vgamma)}
        &\ge \mathbb E\left[ \sum_{t=1}^{\tA} D(\mu_t| \cP_t) \right ]  - \mathbb E\left[ \sum_{t=1}^{\tA} \mu_t^\top \left(  \rho - b_t(x_t) \right) \right]   \\
        &\ge  \mathbb E\left[ \tA D( \bar \mu | \bar \cP ) \right]  - \ubf   \MD(\vcP) - \mathbb  \EE\left[ \sum_{t=1}^{\tA} w_t(\mu_t)  \right] \\
        &\ge  \OPT(\cP) - \ubf   \MD(\cP)  - \EE\left[ \sum_{t=1}^{\tA} w_t(\mu) + E(T,\mu) \right] \ .
    \end{split}
\end{align*}
Repeating the arguments in the proof of Theorem~\ref{thm:master} and substituting $\mu$  as that in the proof of Theorem~\ref{thm:master}, we obtain
\begin{equation*}
    \OPT(\vcP) - \EE \mpran{R(A|\vgamma)} \le \ubf  \MD(\vcP) + \frac{\bar f \ubb }{\lbrho} + \frac 1 {2 \sigma} (\ubb + \ubrho)^2 \eta T + \frac{1}{\eta} \EE \left[ V_h(\mu, \mu_1) \right]\ ,
\end{equation*}
which finishes the proof by noticing $\MD(\vcP)\le \delta$ and bounding $\EE \left[ V_h(\mu, \mu_1) \right]$ as in the proof of Theorem~\ref{thm:master}.

\subsection{Proof of Theorem \ref{thm:lower-bound-inde}}

\begin{proof}
Let $\mathcal A$ be the class of all randomized algorithms. Notice the stochastic model is a special case of the corruption model with $\delta=0$, thus it directly follows from the lower bound of the stochastic model (see, for example, Lemma 1 from \citealt{ArlottoGurvich2019}) that for any $\delta\ge 0$
\[
    \inf_{A \in \mathcal A} \sup_{\vec \cP : \MD(\vec \cP) \le \delta} \Big[ \OPT(\vec \cP) - R(A | \vec \cP) \Big] \ge C_2'\sqrt{T} \ . 
\]
Here we just need to show that there exists $C_1'$ such that the 
\begin{equation}\label{eq:Delta}
    \inf_{A \in \mathcal A} \sup_{\vec \cP : \MD(\vec \cP) \le \delta} \Big[ \OPT(\vec \cP) - R(A | \vec \cP) \Big] \ge C_1'\delta  \ , 
\end{equation}
then \eqref{eq:lower_bound_corruption} holds with $C_1=\frac{1}{2}C_1'$ and $C_2=\frac{1}{2}C_2'$, and taking the average of the above bounds.


{
By Yao's Lemma \citep{yao1977probabilistic}, to lower bound the worst-case regret of all randomized algorithms it suffices to analyze the performance of deterministic algorithms over a known distribution of inputs. We provide more details for completeness. Let $\mathcal A^d$ be the class of deterministic mechanisms. Then, the space of randomized mechanisms $\mathcal A = \Delta(\mathcal A^d)$ can be thought of as distributions over deterministic mechanism. We write $R(A|\vcP) := \EE_{A^d \sim A} [R(A^d | \vec \cP)]$ for the expected performance of a randomized algorithm $A \in \mathcal A$. Additionally, denote by  $\mathbb Q(\delta) = \Delta\big( \mathcal C^{\rm ID}(\delta) \big) \subset \Delta(\Delta(\cS)^T)$ as the set of distributions over $\vec\cP$ that satisfy $\MD(\vec\cP)\le \delta$. Then, we have
\begin{align}\label{eq:Yao1}
    \begin{split}
        \inf_{A \in \mathcal A} \sup_{\vec \cP : \MD(\vec \cP) \le \delta} \Big[ \OPT(\vec \cP) - R(A|\vcP) \Big] &= \inf_{A \in \mathcal A} \sup_{\cQ \in \mathbb Q(\delta)}\EE_{\vcP \sim \cQ} \Big[ \OPT(\vec \cP) - R(A|\vcP) \Big]\\
        & \ge \sup_{\cQ \in \mathbb Q(\delta)} \pran{\EE_{\vcP \sim \cQ} \Big[ \OPT(\vec \cP) \Big] - \sup_{A \in \mathcal A} \EE_{\vcP \sim \cQ} \mpran{R(A | \vec \cP)}} \,, \\
    \end{split}
\end{align}
where the first equation follows because nature does not benefit from randomizing distributions over $\vcP$, and the inequality from the minimax inequality. Because rewards are bounded by $0 \le R(A | \vgamma) \le T \ubf$, we can use Fubini's theorem to obtain that
\begin{align}\label{eq:Yao2}
    \begin{split}
    \sup_{A \in \mathcal A} \EE_{\vcP \sim \cQ} \mpran{R(A | \vec \cP)}
    &= \sup_{A \in \mathcal A} \EE_{\vcP \sim \cQ} \EE_{A^d \sim A} \mpran{R(A^d | \vec \cP)}
    = \sup_{A \in \mathcal A} \EE_{A^d \sim A} \EE_{\vcP \sim \cQ} \mpran{R(A^d | \vec \cP)}\\
    &= \sup_{A^d \in \mathcal A^d} \EE_{\vcP \sim \cQ} \mpran{R(A^d | \vec \cP)}\,,    
    \end{split}
\end{align}
because the decision maker does not benefit from randomizing when the distribution over distributions of inputs is known. Yao's lemma is obtained from combining \eqref{eq:Yao1} and \eqref{eq:Yao2}.}

The rest of the proof involves constructing a distribution of distributions $\cQ \in \mathbb Q(\delta)$ so that  $\EE_{\vcP \sim \cQ} \Big[ \OPT(\vec \cP) \Big] - \sup_{A^d \in \mathcal A^d} \EE_{\vcP \sim \cQ}\mpran{ R(A^d | \vec \cP)}\ge C_1' \delta$. In other words, every deterministic mechanism incurs a regret of order $\delta$.

Consider a one-dimensional binary knapsack problem with linear reward, namely, $b_t(x) = b_t x$ with $b_t \in \RR_+$, $f_t(x_t)=f_t^\top x_t$ with $f_t\in\RR_+$, and $x_t \in \cX = \{0,1\}$. The average budget is set to $\rho=0.5$.
We define two vectors of distributions $\vec \cP^1$ and $\vec \cP^2$ of requests $(f, b)$ parameterized by $\hht\le \frac{T}{2}$ as follow: 

$$
\vec \cP^1 = \pran{\underbrace{\cP_a, \cP_a, \ldots, \cP_a}_{\hht \text{ periods}}, \underbrace{\cP_b, \cP_b, \ldots, \cP_b}_{\hht \text{ periods}}, \underbrace{\cP_d, \cP_d, \ldots, \cP_d}_{T-2\hht \text{ periods}}}
$$
$$
\vec \cP^2 = \pran{\underbrace{\cP_a, \cP_a, \ldots, \cP_a}_{\hht \text{ periods}}, \underbrace{\cP_c, \cP_c, \ldots, \cP_c}_{\hht \text{ periods}}, \underbrace{\cP_d, \cP_d, \ldots, \cP_d}_{T-2\hht \text{ periods}}} \ ,
$$
where $\cP_a$, $\cP_b$, $\cP_c$ and $\cP_d$ are deterministic distributions on $(f,b)$ defined by $\cP_a((f,b)=(0.5, 1))=1$, $\cP_b((f,b)=(0,1))=1$, $\cP_c((f,b)=(1,1))=1$ and $\cP_d((f,b)=(1,0.5))=1$.
Notice each vector of distributions $\vcP^1$ and $\vcP^2$ has three stages; the first two stages are $\hht$ periods long and the last one is $T-2\hht$ periods long, and the only difference between $\vcP^1$ and $\vcP^2$ are the second stage. 
Moreover, it is easy to see that the support of $\vec\cP^1$ is $\pran{(0.5, 1), (0, 1), (1, 0.5)}$ and the average distribution $\bar \cP_1 = \frac 1 T \sum_{s=1}^T \cP^1_s $ chooses the requests in the support with probability $(\hht/T, \hht/T, 1-2\hht/T)$ respectively, thus the mean deviation of $\vcP^1$ is bounded by
\begin{align}\label{eq:TV_bound}
\begin{split}
    \MD(\vec \cP^1) &= \hht \pran{\pran{1-\frac {\hht} T } + \frac {\hht} T  + \pran{1-\frac {2\hht} T }} + \hht \pran{ \frac {\hht} T  + \pran{1-\frac {\hht} T } + \pran{1-\frac {2\hht} T }}  + (T-2\hht) \frac {4 \hht} T  \\
    &\le 8 \hht \ .
\end{split}
\end{align}
A similar argument yields that
\begin{equation}\label{eq:TV_bound2}
    \MD(\vec \cP^2) \le 8 \hht \ .
\end{equation}

{
Consider a distribution of distributions $\cQ$ that takes value $\vec \cP^1$ and $\vec \cP^2$ with equal probability. Notice that any online, deterministic algorithm $A^d$ cannot distinguish between $\vec \cP^1$ and $\vec \cP^2$ until Stage 2 is reached, thus the number of taken items in Stage 1 is independent from the underlining distribution. Suppose the algorithm decides to take a deterministic amount $l \in [0, \hat t]$ of items from Stage 1. Clearly, any algorithm is better off taking all the items from Stage 3 (for both distributions) as these consume 0.5 units of resource and yield a revenue of 1. Because the total capacity is $T/2$, taking these times yields a remaining budget of $T/2 - l - 0.5 (T - 2 \hat t) = \hat t - l$. Therefore, the best action for Stage 2 is to take $\hht-l$ items. Thus, it holds for any online, deterministic algorithm $A^d$ that
\begin{align*}
\begin{split}
    \EE_{\vec \cP\sim Q} \mpran{R(A^d|\vec \cP)} & = \frac{1}{2} R(A^d|\vec \cP^1) + \frac{1}{2} R(A^d|\vec \cP^2) \\
    & \le \frac{1}{2} \pran{0.5l + (T-2\hht)} + \frac{1}{2} \pran{0.5l + (\hht-l) + (T-2\hht)} \\
    & = T-\frac{3}{2}\hht \ ,
\end{split}
\end{align*}
whereby it holds that 
\begin{equation}\label{eq:supA}
    \sup_{A^d \in \mathcal A^d}  \EE_{\vec \cP\sim \cQ} \mpran{R(A^d|\vec \cP)} \le T-\frac{3}{2}\hht \ .
\end{equation}}%
Meanwhile, it is easy to check that $\OPT(\vec \cP^1)= 0.5\hht+(T-2\hht)$ (with the benefit of hindsight, it is optimal to take all requests from Stage 1 and Stage 3) and $\OPT(\vec \cP^1)= \hht+(T-2\hht)$ (because it is optimal to take all requests from Stage 2 and Stage 3), whereby 
\begin{equation}\label{eq:EOPT}
    \EE_{\vec \cP\sim \cQ}\mpran{\OPT(\vec \cP)} = \frac{1}{2} \OPT(\vec \cP^1) + \frac{1}{2} \OPT(\vec \cP^2) = T-\frac{5}{4}\hht \ .
\end{equation}

Combining \eqref{eq:TV_bound}, \eqref{eq:TV_bound2}, \eqref{eq:supA}, \eqref{eq:EOPT}, it holds for any $\delta=8\hht$ with integer $\hht\le \frac{T}{2}$ that
\begin{align}\label{eq:bound_Delta}
\begin{split}
     \EE_{\cP \sim \cQ} \Big[ \OPT(\vec \cP) \Big] - \sup_{A^d \in \mathcal A^d} \EE_{\cP \sim \cQ}\mpran{R(A^d | \vec \cP)} \ge \pran{T-\frac{5}{4}\hht} - \pran{T-\frac{3}{2}\hht} = \frac{1}{4}\hht \ge \frac{1}{32} \delta \ .
\end{split}
\end{align}
Substituting \eqref{eq:bound_Delta} into the lower bound obtained from Yao's lemma, we arrive at 
$$
\inf_{A \in \mathcal A} \sup_{\vec \cP : \MD(\vec \cP) \le \delta} \Big[ \OPT(\vec \cP) - R(A | \vec \cP) \Big] \ge \frac{1}{32} \delta,
$$
for any $\delta=8\hht$ with integer $\hht\le \frac{T}{2}$. Therefore for any $8\le \delta \le 4 T$, we have 
\begin{align*}
\begin{split}
     \inf_{A \in \mathcal A} \sup_{\vec \cP : \MD(\vec \cP) \le \delta} \Big[ \OPT(\vec \cP) - R(A | \vec \cP) \Big] & \ge \inf_{A \in \mathcal A} \sup_{\vec \cP : \MD(\vec \cP) \le 8 \left\lfloor {\delta/8} \right\rfloor} \Big[ \OPT(\vec \cP) - R(A | \vec \cP) \Big] \\
     \ge & \frac{1}{32} 8 \left\lfloor {\delta/8} \right\rfloor \ge \frac{1}{64} \delta \ ,
\end{split}
\end{align*}
where the first inequality is because $\delta\le 8 \left\lfloor {\delta/8} \right\rfloor$.
Therefore, \eqref{eq:Delta} holds for $C_1'=\frac{1}{64}$, which finishes the proof.
\end{proof}

\subsection{Proof of Theorem \ref{thm:ergodic}}

{
We first show that the dual solutions produced by our algorithm are stable in the sense that the change in the multipliers from one period to the next is proportional to the step size.}

\begin{prop}\label{prop:diff-iterate}
It holds for any $t$ that
$$\|\mu_{t+1}-\mu_{t}\dualnorm\le \frac{\sqrt{2}}{\sigma} \eta\|g_t\primalnorm\le \frac{\sqrt{2}}{\sigma}\eta (\ubb+\ubrho) \ .$$
\end{prop}
\begin{proof}
The dual mirror descent update \eqref{eq:dual_update} can be written as 
\begin{equation}\label{eq:another_update}
    \nabla h(\tmu_t)=\nabla h(\mu_t)-\eta g_t\ , \ \ \ \  \mu_{t+1}=\arg\min_{\mu\in\cD} V_h(\mu,\tmu_t) \ ,
\end{equation}
which can be easily verified by noticing that the optimality condition of \eqref{eq:dual_update} and \eqref{eq:another_update} are both
$0\in \nabla h(\mu_{t+1})-\nabla h(\mu_t) +\eta g_t + \mathcal N_{\cD}(\mu_{t+1})$, where { $\mathcal N_{\cD}(\mu_{t+1})=\{g\in\RR^m | g^\top(\mu-\mu_{t+1})\le 0 \text{ for all } \mu\in\cD\}$} is the normal cone of $\cD$ at $\mu_{t+1}$. Let $h^*$ be the convex conjugate of $h$, then $\nabla h^*=(\nabla h)^{-1}$, and $\nabla h^*$ is $1/\sigma$-Lipschitz continuous in the primal norm by recalling that $h$ is $\sigma$-strongly-convex in the dual norm (see, e.g., \citealt{kakade2009duality}). Thus,
\begin{align}\label{eq:diff-mu}
\begin{split}
        \|\tmu_t-\mu_t\dualnorm=&\|\nabla h^*(\nabla h (\tmu_{t})) - \mu_t\dualnorm=\|\nabla h^*(\nabla h (\mu_{t})-\eta g_t) - \mu_t\dualnorm\\
    =&\|\nabla h^*(\nabla h (\mu_{t})-\eta g_t) - \nabla h^*(\nabla h (\mu_{t}))\dualnorm\le \frac{\eta}{\sigma} \|g_t\primalnorm \ .
\end{split}
\end{align}
Meanwhile, it follows by generalized Pythagorean Theorem of Bregman projection~\citep{nielsen2007bregman} that
\begin{equation}\label{eq:m-1}
    V_h(\mu_t, \tmu_t)\ge V_h(\mu_t, \mu_{t+1})+V_h(\mu_{t+1},\tmu_t) \ge V_h(\mu_t, \mu_{t+1}) \ge \frac{\sigma}{2}\|\mu_{t+1}-\mu_t\dualnorm^2 \ ,
\end{equation}
where the second inequality is from the non-negativity of Bregman divergence and the last inequality uses the strong-convexity of $h$ w.r.t.~dual norm. On the other hand, we have
\begin{align}\label{eq:m-2}
\begin{split}
    V_h(\mu_t, \tmu_t)&\le V_h(\mu_t, \tmu_t)+V_h(\tmu_t,\mu_t)= \left(\nabla h (\tmu_t)-\nabla h(\mu_t) \right)^\top \left(\tmu_t-\mu_t\rangle\right)\\
    &\le \|\nabla h (\tmu_t)-\nabla h(\mu_t)\| \|\tmu_t-\mu_t\dualnorm \le \frac{\eta^2}{\sigma}\|g_t\primalnorm^2\ ,
\end{split}
\end{align}
where the second inequality follows from Cauchy-Schwartz, and the third inequality utilizes \eqref{eq:another_update} and \eqref{eq:diff-mu}. 
We finish the proof by combining \eqref{eq:m-1} and \eqref{eq:m-2}, and noticing $\|g_t\primalnorm\le \|b_t(x_t)\primalnorm+\|\rho\primalnorm\le \ubb+\ubrho$.
\end{proof}

\textbf{Proof of Theorem \ref{thm:ergodic}}
\begin{proof} We discuss how the steps of the proof of Theorem~\ref{thm:master} need to be adapted to account for the time dependence of distributions of requests.

\paragraph{Step 1.} Fix an ergodic process $\cP \in \mathcal C^{\rm E}(\delta, k)$. By definition, there exists a one-period distribution $\bar \cP \in \Delta(\cS)$ such that $\TV_k(\cP, \bar \cP) \le \delta$. Consider a time $t \le \tA$ so that actions are not constrained by resources. By the definition of $f^*_t(\mu_t)$, we have that 
\[
f_t(x_t)  = f_t^*(\mu_t) + \mu_t^{\top} \rho - \mu_{t}^{\top} (\rho - b_t(x_t))\,.
\]
Summing up with $t=1,\ldots,\tA$ and taking expectations, we obtain that
\begin{align}\label{eq:f_and_r_ergodic}
\mathbb E\left[ \sum_{t=1}^{\tA} f_t(x_t)  \right] 
&= 
\mathbb E\left[ \sum_{t=1}^{\tA} f_t^*(\mu_t) + \mu_t^{\top} \rho \right ]-  \mathbb E\left[ \sum_{t=1}^{\tA} \mu_t^\top \left(  \rho - b_t(x_t) \right) \right]\,.
\end{align}
Let $\Psi_t(\mu) = f_t^*(\mu)$ and $\Psi(\mu) = \mathbb E_{(f,b) \sim \bar \cP}[f^*(\mu)]$. Then, we have that
\begin{align}
\begin{split}\label{eq:ergodic-four}
        \sum_{t=1}^{\tA} \Psi_t(\mu_t) =& 
    \sum_{t=1}^{\tA-k} \Big( \Psi_{t+k}(\mu_{t+k}) - \Psi_{t+k}(\mu_t) \Big) \\
    &+\sum_{t=1}^{\tA-k} \Big( \Psi_{t+k}(\mu_{t}) - \Psi(\mu_t) \Big)\\
    &+\sum_{t=1}^{k} \Psi_t(\mu_{t}) -  \sum_{t=\tA-k+1}^{\tA} \Psi(\mu_t)\\
    &+\sum_{t=1}^{\tA} \Psi(\mu_t)\,.
\end{split}
\end{align}
We bound each term in \eqref{eq:ergodic-four} at a time. For the first term, we use that for all $\mu \in \cD$ we have $\Psi_t(\mu_t) \ge \Psi_t(\mu) - \ubb \| \mu_t - \mu\dualnorm$. To see this, fix $\mu \in \cD$ and use that by definition of the convex conjugate that for all $x \in \cX$
\[
    f^*_t(\mu_t) \ge f_t(x) - \mu_{t}^{\top} b_t(x) = f_t(x) - \mu^{\top} b_t(x) + (\mu - \mu_t)^{\top} b_t(x) \ge f_t(x) - \mu^{\top} b_t(x) - \bar b \| \mu_t - \mu\dualnorm \,,
\]
where the last inequality follows from Cauchy-Schwartz and $\|b(x) \primalnorm \le \bar b$. Taking supremum over $x \in \mathcal X$, we obtain that $f^*_t(\mu_t) \ge f^*_t(\mu) - \bar b \| \mu_t - \mu\dualnorm $.  Let $C=\frac{\sqrt{2}}{\sigma} (\ubb+\ubrho)$, then it follows from Proposition \ref{prop:diff-iterate} that $\|\mu_{t}-\mu_{t-1}\dualnorm\le C \eta$. Therefore,
\[
    \sum_{t=1}^{\tA-k} \Psi_{t+k}(\mu_{t+k}) - \Psi_{t+k}(\mu_t) \ge - \ubb \sum_{t=1}^{\tA-k} \| \mu_{t+k} - \mu_t \dualnorm \ge - \ubb C T \eta k\,.
\]
For the second term, we first take expectations and use a martingale argument to write
\begin{align*}
    \mathbb E\left[\sum_{t=1}^{\tA-k} \Psi_{t+k}(\mu_{t}) - \Psi(\mu_t) \right]
    &= \mathbb E\left[\sum_{t=1}^{\tA-k} \mathbb E_{(f,b)\sim\cP_{t+k}(\gamma_{1:t - 1})}[ f^*(\mu_t) ] - \Psi(\mu_t ) \right]\\
    &\ge -\ubf \mathbb E\left[ \sum_{t=1}^{\tA-k} \| \cP_{t+k}(\gamma_{1:t - 1}) - \bar \cP \|_{\TV} \right]\\
    &\ge - \ubf \cdot T \cdot  \mathrm{TV}_k(\cP, \bar \cP)\,,
\end{align*}
because $0 \le f^*(\mu) \le \ubf$ as argued in the proof of Theorem~\ref{thm:independent}. For the third term, use our bounds on $f^*(\mu)$ again to obtain
\[
\sum_{t=1}^{k} \Psi_t(\mu_{t}) -  \sum_{t=\tA-k+1}^{\tA} \Psi(\mu_t) \ge - \ubf k\,.
\]
Therefore, denoting by $D(\mu | \bar \cP) = \Psi(\mu) + \mu^\top \rho$ the expected dual function at the stationary distribution,
\begin{align}\label{eq:bound_dual_ergodic}
    \mathbb E\left[ \sum_{t=1}^{\tA} f_t^*(\mu_t) + \mu_t^{\top} \rho \right ] 
    &\ge  {\mathbb E\left[ \sum_{t=1}^{\tA} \Psi(\mu_t)+\mu_t^{\top} \rho \right]  - \ubb C T \eta k - \ubf \cdot T \cdot  \mathrm{TV}_k(\cP, \bar \cP) - \ubf k} \nonumber \\
    & = \mathbb E\left[ \sum_{t=1}^{\tA} D( \mu_t | \bar \cP ) \right]  - \ubb C T \eta k - \ubf \cdot T \cdot  \mathrm{TV}_k(\cP, \bar \cP) - \ubf k \nonumber \\
    & \ge \mathbb E\left[ \tA D( \bar \mu | \bar \cP ) \right]  - \ubb C T \eta k - \ubf \cdot T \cdot  \mathrm{TV}_k(\cP, \bar \cP) - \ubf k
    \,,
\end{align}
where we denote by $\bar \mu = \frac 1 \tA \sum_{t=1}^{\tA} \mu_{t}$ the average multiplier and use that $D(\mu | \bar \cP)$ is convex in $\mu$.

\paragraph{Step 2.} This step applies directly because the analysis is deterministic in nature.

\paragraph{Step 3.} With some abuse of notation, define $\OPT(\cP) := \EE_{\vec \gamma \sim \cP} [ \OPT(\vec \gamma) ]$. Proposition~\ref{prop:upper_bound} implies that for every $\mu \in \cD$ we have 
\begin{align*}
    \OPT(\cP) &= \EE_{\vec \gamma \sim \cP} [ \OPT(\vec \gamma) ] \le \mathbb E_{\vec \gamma} \left[ \sum_{t=1}^T \left( f_t^*(\mu)+ \rho^\top \mu \right) \right]\,.
\end{align*}
Using the tower rule for conditional expectations, we obtain that
\begin{align*}
    \mathbb E_{\vec \gamma \sim \cP} \left[ \sum_{t=1}^T f_t^*(\mu) \right]
    &= \mathbb E_{\vec \gamma \sim \cP} \left[ \sum_{t=1}^k f_t^*(\mu) \right]
    + \mathbb E_{\vec \gamma \sim \cP} \left[ \sum_{t=k+1}^T \mathbb E_{(f,b) \sim \cP_{t+k}(\gamma_{1:t-1})}\left[ f^*(\mu) \right] \right]\\
    &\le \ubf k + (T-k) \Psi(\mu| \bar \cP) + \ubf \cdot (T-k) \cdot \mathrm{TV}_k(\cP, \bar \cP)\,,
\end{align*}
because $0 \le f^*(\mu) \le \ubf$. Thus, we obtain that
\[
    \OPT(\cP) \le T D(\mu | \bar \cP) + \ubf k  + \ubf \cdot T \cdot \mathrm{TV}_k(\cP, \bar \cP) \,,
\]
because $D(\mu | \bar \cP) = \Psi(\mu) + \rho^\top \mu$. Therefore, for any  $\tA \in [0,T]$ and $\mu \in \cD$ we have that
\begin{align}
\begin{split}\label{eq:bound-opt-ergodic}
\OPT(\cP) &=  \frac{\tA}{T} \OPT(\cP)  + \frac{T-\tA}{T}\OPT(\cP) \\
&\le  \tA D(\mu| \bar \cP) + \pran{T-\tA}{\ubf} + \ubf k  + \ubf \cdot T \cdot \mathrm{TV}_k(\cP, \bar \cP)\ ,
\end{split}
\end{align}
where the inequality uses that $\OPT(\cP) \le T\ubf$. Combining \eqref{eq:regret_omd}, \eqref{eq:f_and_r_ergodic}, \eqref{eq:bound_dual_ergodic} and \eqref{eq:bound-opt-ergodic}, and utilizing \eqref{eq:regret_omd}, it holds that
\begin{align*}
    \begin{split}
        \EE \mpran{R(A|\vgamma)}&
        \ge \mathbb E\left[ \sum_{t=1}^{\tA} f_t^*(\mu_t) + \mu_t^{\top} \rho \right ] - \mathbb E\left[ \sum_{t=1}^{\tA} \mu_t^\top \left(  \rho - b_t(x_t) \right) \right]\\
        &\ge  \mathbb E\left[ \tA D( \bar \mu | \bar \cP ) \right]  - \ubb C T \eta k - \ubf \cdot T \cdot  \mathrm{TV}_k(\cP, \bar \cP) - \ubf k - \EE\left[ \sum_{t=1}^{\tA}  w_t(\mu_t)  \right] \\
        &\ge  \OPT(\cP) - \ubb C T \eta k - 2 \ubf \cdot T \cdot  \mathrm{TV}_k(\cP, \bar \cP)  -\pran{T-\tA}{\ubf} - 2k \ubf  - \EE\left[ \sum_{t=1}^{\tA} w_t(\mu) + E(T,\mu) \right]\ .
    \end{split}
\end{align*}
Repeating the arguments in the proof of Theorem~\ref{thm:master} and substituting $\mu$ as that in the proof of Theorem~\ref{thm:master}, we obtain
\begin{align*}
\OPT(\cP) - \EE \mpran{R(A|\vgamma)} \le \ubb C T \eta k + 2 \ubf \cdot T \cdot  \mathrm{TV}_k(\cP, \bar \cP)  + 2k \ubf + \frac{\bar f \ubb }{\lbrho} + \frac 1 {2 \sigma} (\ubb + \ubrho)^2 \eta\cdot T + \frac{1}{\eta} \EE \left[ V_h(\mu, \mu_1) \right] \ ,
\end{align*}
which finishes the proof by noticing $\mathrm{TV}_k(\cP, \bar \cP)\le \delta$ {\color{black} and bounding $\EE \left[ V_h(\mu, \mu_1) \right]$ as in the proof of Theorem~\ref{thm:master}.}
\end{proof}

\subsection{Proof of Theorem \ref{thm:periodic}}

Here, we prove a more general statement. Theorem~ \ref{thm:periodic} the follows as a special case.

We assume $\vec \gamma \in \cS^T$ follows an arbitrary stochastic process $\cP \in \Delta(\cS^T)$. We need some definitions before stating our main result. We denote by $\gamma_{1:t} = (\gamma_s)_{s=1}^t$ and let $\mathcal P_t(\gamma_{1:s})$ be the conditional distribution of $\gamma_t$ given $\gamma_{1:s}$ for $s < t$. Let $\varPi = (\mathcal T_k)_{k = 1}^K$ be a partition of $[T]$ with $\mathcal T_k = \{t_{k}, \ldots, t_{k+1}-1\}$, $t_1 = 1$ and $t_{K+1} = T$. We denote by $\ell^2(\varPi) = \sum_{k=1}^K (t_{k+1} - t_{k})^2$ the sum of squared-lengths and 
\[
    \TV(\cP, \bar \cP, \varPi) = \sup_{\vec \gamma \in \mathcal S^{T}} \sum_{k=1}^K \left \| \sum_{t \in \mathcal T_k} \left( \cP_t \left(\gamma_{1:{t_k-1}}\right) -  \bar \cP \right) \right\|_\TV\,,
\]
the worst-case total variation distance between the distributions conditional on the data at the beginning of a partition and $\bar \cP \in \Delta(\cS)$.

\begin{thm}\label{thm:periodic-general}
Consider Algorithm \ref{al:sg} with step-size $\eta \ge 0$ and initial solution $\mu_1\in \cD$. Suppose Assumptions~\ref{ass:p}-\ref{ass:h} are satisfied, and the requests are drawn from an arbitrary stochastic process with a partition $\varPi$. Then, it holds for any $T\ge 1$ and mean deviation $\delta>0$ that
\[
 \EE_{\vec \gamma \sim \cP} \Big[ \OPT(\vec \gamma) - R(A | \vec \gamma) \Big] \le C_1 + C_2 \eta T + \frac {1} {\eta} C_3 + 2 \ubf  \TV(\cP, \bar \cP, \varPi) + C_4 \eta \ell^2(\varPi)\,,
\]
where $C_1=\frac{\bar f \ubb }{\lbrho}$, 
$C_2=\frac 1 {2 \sigma} (\ubb + \ubrho)^2$, 
$C_3= \sup_{\mu \in \cD: \|\mu\dualnorm \le  C_1  \maxeiterm}  V_h(\mu, \mu_1)$ and 
$C_4= \frac{\sqrt{2}}{2\sigma} (\ubb+\ubrho)^2$.
\end{thm}


\begin{proof} Fix a partition $\varPi$ of $[T]$. We discuss how the steps of the proof of Theorem~\ref{thm:master} need to be adapted to account for the time dependence of distributions of requests. 

\paragraph{Step 1.} Consider a time $t \le \tA$ so that actions are not constrained by resources. Again, we have that 
\[
f_t(x_t)  = f_t^*(\mu_t) + \mu_t^{\top} \rho - \mu_{t}^{\top} (\rho - b_t(x_t))\,.
\]
Summing from $t=1,\ldots,\tA$, decomposing the sum over the partitions and taking expectations we obtain that
\begin{align}\label{eq:f_and_r_stationary}
\begin{split}
    \mathbb E\left[ \sum_{t=1}^{\tA} f_t(x_t)  \right] 
&= 
\mathbb E\left[ \sum_{t=1}^{\tA} f_t^*(\mu_t) + \mu_t^{\top} \rho \right ]-  \mathbb E\left[ \sum_{t=1}^{\tA} \mu_t^\top \left(  \rho - b_t(x_t) \right) \right] \\
&= \sum_{k=1}^K \underbrace{ \mathbb E\left[  \sum_{t=t_k}^{(t_{k+1} - 1) \wedge \tA} \left( f_t^*(\mu_t) + \mu_t^{\top} \rho \right) \right]}_{A_k} -  \mathbb E\left[ \sum_{t=1}^{\tA} \mu_t^\top \left(  \rho - b_t(x_t) \right) \right]\,,
\end{split}
\end{align}
where by convention we take the sum to be zero if $\tA < t_k$. 

Fix a subset $k \in [K]$ of the partition. We next lower bound $A_k$. Fix $\mu \in \cD$. In the proof of Theorem~\ref{thm:ergodic} we showed that $f^*_t(\mu_t) \ge f^*_t(\mu) - \bar b \| \mu_t - \mu\dualnorm $. Therefore, taking $\mu = \mu_{t_k}$ we obtain:
\begin{align*}
   A_k &\ge \mathbb E\left[ \sum_{t=t_k}^{(t_{k+1} - 1) \wedge \tA} \left( f_t^*( \mu_{t_k}) + \rho^\top \mu_{t_k} \right) \right] - (\bar a + \bar b)  \sum_{t \in \mathcal T_k} \mathbb E\left[ \| \mu_t - \mu_{t_k}\dualnorm\right]\,.
\end{align*}
{ Let $D(\mu | \tilde{\cP}) = \EE_{(f,b,\cX) \sim \tilde{\cP}} \left[ f^*(\mu) \right] + \mu^\top \rho $ be the expected dual function when requests are distributed according to $\tilde{\cP} \in \Delta(\cS)$.} For the first term, use a martingale argument to write
\begin{align*}
    &\mathbb E\left[  \sum_{t=t_k}^{(t_{k+1} - 1) \wedge \tA} \left( f_t^*( \mu_{t_k}) + \rho^\top \mu_{t_k} \right) \right]
    = \mathbb E\left[  \sum_{t=t_k}^{(t_{k+1} - 1) \wedge \tA} D(\mu_{t_k} | \cP_t (\gamma_{1:t_k - 1}) ) \right]\\
    \quad&\ge \mathbb E\left[ \underbrace{\max( \min(t_{k+1} - t_k, \tA +1 -t_k ), 0 )}_{\tau_k} D(\mu_{t_k} | \bar \cP ) \right]
    - \ubf \mathbb E\left[ \left \| \sum_{t \in \mathcal T_k} \cP_t (\gamma_{1:t_k-1}) -  \bar \cP \right\|_\TV \right] \ ,
\end{align*}
because $0 \le f^*(\mu) \le \ubf$ as argued in the proof of Theorem~\ref{thm:independent}. Let $C=\frac{\sqrt{2}}{\sigma} (\ubb+\ubrho)$, then it follows from Proposition \ref{prop:diff-iterate} that $\|\mu_{t}-\mu_{t-1}\dualnorm\le C \eta$. Meanwhile, it holds from triangle inequality that:
\[
    \sum_{t=t_k}^{t_{k+1}-1} \| \mu_t - \mu_{t_k}\dualnorm \le C \eta \sum_{t=t_k}^{t_{k+1}-1} (t - t_{k}) \le \frac C 2 \eta  (t_{k+1} - t_{k})^2\,.
\]
Therefore,
\begin{align}\label{eq:bound_dual_stationary}
    \sum_{k=1}^K A_k &\ge \mathbb E\left[ \sum_{k=1}^K \tau_k D(\mu_{t_k} | \bar \cP )  - \ubf \sum_{k=1}^K \left \| \sum_{t \in \mathcal T_k} \cP_t (\gamma_{1:t_k-1}) -  \bar \cP \right\|_\TV \right] - \frac C 2 (\bar a + \bar b) \eta \sum_{k=1}^K (t_{k+1} - t_{k})^2 \nonumber\\
    &= \mathbb E\left[ \sum_{k=1}^K \tau_k D(\mu_{t_k} | \bar \cP ) \right]  - \ubf \cdot \TV(\cP, \bar \cP, \varPi) - \frac C 2 (\bar a + \bar b) \eta \cdot \ell^2(\varPi) \nonumber\\
    &\ge \mathbb E\left[ \tA D(\bar \mu | \bar \cP) \right]  - \ubf \cdot \TV(\cP, \bar \cP, \varPi) - \frac C 2 (\bar a + \bar b) \eta \cdot \ell^2(\varPi)\,,
\end{align}
where we used that $\sum_{k=1}^K \tau_k = \tA$ together with $\bar \mu = \frac 1 \tA \sum_{k=1}^k \tau_k \mu_{t_k}$ and that $D(\mu | \cP)$ is convex in $\mu$.

\paragraph{Step 2.} This step applies directly because the analysis is deterministic in nature.

\paragraph{Step 3.} Define $\OPT(\cP) := \EE_{\vec \gamma \sim \cP} [ \OPT(\vec \gamma) ]$. Proposition~\ref{prop:upper_bound} implies that for every $\mu \in \cD$ we have 
\begin{align*}
    \OPT(\cP) &= \EE_{\vec \gamma \sim \cP} [ \OPT(\vec \gamma) ] \le \mathbb E_{\vec \gamma \sim \cP} \left[ \sum_{t=1}^T \left( f_t^*(\mu)+ \rho^\top \mu \right) \right]\\
    &= \mathbb E_{\vec \gamma \sim \cP} \left[ \sum_{k=1}^K \sum_{t \in \mathcal T_k} D(\mu | \cP_t (\gamma_{1:t_k-1}) ) \right] \le T D(\mu | \bar \cP) + \ubf \cdot \TV(\cP, \bar \cP, \varPi) \,,
\end{align*}
where the second equality follows form the tower rule for conditional expectations, and the last inequality is because $0 \le f_t^*(\mu) \le \ubf$. Therefore, for any  $\tA \in [0,T]$ and $\mu \in \cD$ we have that
\begin{align}\label{eq:bound-opt-stationary}
\OPT(\cP) &=  \frac{\tA}{T} \OPT(\cP)  + \frac{T-\tA}{T}\OPT(\cP) \le  \tA D(\mu| \bar \cP) + \pran{T-\tA}{\ubf} + \ubf \cdot \TV(\cP, \bar \cP, \varPi)\ ,
\end{align}
where the inequality uses that $\OPT\le T \ubf$. Combining \eqref{eq:f_and_r_stationary}, \eqref{eq:bound_dual_stationary} and \eqref{eq:bound-opt-stationary}, it holds that
\begin{align*}
    \begin{split}
        \EE \mpran{R(A|\vgamma)}
        &\ge  \sum_{k=1}^K A_k - \mathbb E\left[ \sum_{t=1}^{\tA} \mu_t^\top \left(  \rho - b_t(x_t) \right) \right]\\
        &\ge  \mathbb E\left[ \tA D(\bar \mu | \bar \cP) \right]  - \ubf \cdot \TV(\cP, \bar \cP, \varPi) - \frac C 2 (\bar a + \bar b) \eta \cdot \ell^2(\varPi)- \mathbb E\left[ \sum_{t=1}^{\tA} \mu_t^\top \left(  \rho - b_t(x_t) \right) \right] \\
        &\ge  \OPT   - 2 \ubf \cdot \TV(\cP, \bar \cP, \varPi) -   \frac C 2 (\bar a + \bar b) \eta \cdot \ell^2(\varPi) - \pran{T-\tA}{\ubf}- \mathbb E\left[ \sum_{t=1}^{\tA} w_t(\mu) + E(T,\mu) \right]\ ,
    \end{split}
\end{align*}
where the last inequality utilize \eqref{eq:regret_omd}.
Repeating the arguments in the proof of Theorem~\ref{thm:master} and substituting $\mu$ as that in the proof of Theorem~\ref{thm:master}, we obtain
\begin{align*}
\OPT(\cP) - \EE \mpran{R(A|\vgamma)} \le 2 \ubf \cdot \TV(\cP, \bar \cP, \varPi) +   \frac C 2 (\bar a + \bar b) \eta \cdot \ell^2(\varPi) + \frac{\bar f \ubb }{\lbrho} + \frac 1 {2 \sigma} (\ubb + \ubrho)^2 \eta\cdot T + \frac{1}{\eta} \mathbb E\left[ V_h(\mu, \mu_1) \right] \ ,
\end{align*}
which finishes the proof by noticing $\mathrm{TV}_k(\cP, \bar \cP,\varPi)\le \delta$ {\color{black} and bounding $\EE \left[ V_h(\mu, \mu_1) \right]$ as in the proof of Theorem~\ref{thm:master}.}
\end{proof}

Theorem \ref{thm:periodic} is a special case of Theorem \ref{thm:periodic-general}. Suppose the requests are periodic with cycles of length $q$, that is, $\gamma_{1:q} \stackrel{(d)}{=} \gamma_{q+1:2q} \stackrel{(d)}{=} \gamma_{2q+1:3q}$ and so on. In this case, we let the partition $\varPi$ coincide with the cycles, i.e., $\mathcal T_k = \{(k-1)q + 1,\ldots, kq\}$ and $\bar \cP = \frac 1 q \sum_{t=1}^q \cP_t$ where $\cP_t$ denotes the unconditional distribution of the $t$-th entry of a cycle, thus $\ell^2(\varPi) = q^2 \lceil T/q \rceil\le 2qT$ and $\TV(\cP, \varPi, \bar \cP) = 0$. Applying Theorem \ref{thm:periodic-general} directly results in Theorem \ref{thm:periodic}.


\section{Proofs in Section~\ref{sec:applications}}

\subsection{Proof of Proposition~\ref{prop:matching} }\label{sec:stochastic-functions}

Suppose that both rewards and resource consumption are random. Let $\vec \zeta = (\zeta_1, \ldots, \zeta_T)$ denote a vector of i.i.d.~random variables that determine the realization of the above processes. We assume without loss that $\zeta_t$ are uniformly distributed in $[0,1]$ as any other random variable can be obtained from inverse transform sampling. At time $t$, the request $\gamma_t = (f_t, b_t, \cX_t)$ gives a random reward function $f_t : \cX_t \times [0,1] \rightarrow \RR_+$ and random resource consumption function $b_t : \cX_t \times [0,1] \rightarrow \RR_+^m$. As in Assumption~\ref{ass:p}, we assume these random quantities are bounded almost surely as follows $0\le f_t(x,\zeta) \le \ubf$ and $\|b_t(x,\zeta) \primalnorm \le \ubb$.

\paragraph{Online Algorithm.} The timing of events is as follows. At the beginning of time period $t$, the decision maker receives a request $\gamma_t$. The decision maker chooses an action $x_t \in \cX_t$. Afterwards, the random variable $\zeta_t$ is revealed, which leads to a reward $f_t(x_t, \zeta_t)$ and consumes $b_t(x_t, \zeta_t)$ resources. We emphasize that the realization of $\zeta_t$ is revealed \emph{after} taking action.  We denote by $R(A | \vgamma, \vec \zeta) = \sum_{t=1}^T f_t(x_t, \zeta_t)$ the performance of algorithm $A$.

Algorithm~\ref{al:sg} can be adapted to this setting by computing, for each request, the expected reward function $f_t(x) = \mathbb E_{\zeta_t}[ f_t(x, \zeta_t) ]$ and expected resource consumption function $b_t(x) = \mathbb E_{\zeta_t}[ b_t(x, \zeta_t)]$, and then choosing the primal decision and the sub-gradient in expectation. That is, we set 
	\begin{align*}
		\tilde{x}_{t} &= \arg\max_{x\in\cX_t}\left\{f_{t}(x)-\mu_{t}^{\top} b_{t}(x)\right\} \\ 
		\tg_t &= -b_t(\tilde x_t) + \rho\ .
	\end{align*}
Resources are updated according to the realized resource consumption.

\paragraph{Benchmark.} The benchmark is now defined in terms of the expected reward and resource consumption functions.  Let $\Omega_t = (f_t(x), b_t(x))_{x\in\cX_t} \subseteq \mathbb \RR_+^{m+1}$ be the set of achievable expected reward and resource consumption pairs, where we denote the expected reward function by $f_t(x) = \mathbb E_{\zeta_t}[ f_t(x, \zeta_t) ]$ and the expected resource consumption function by $b_t(x) = \mathbb E_{\zeta_t}[ b_t(x, \zeta_t) ]$. Let $\hat \Omega_t$ be its convex hull, i.e., the smallest convex set containing $\Omega_t$. Then, we define the offline optimum as follows
\begin{align}\label{eq:OPT-random}
    \OPT(\vec \gamma)=
	\max_{(\hat f_t, \hat b_t) \in \hat \Omega_t} \left\{  \sum_{t=1}^T \hat f_t    \,\,\text{s.t.}\,
	\sum_{t=1}^T \hat b_t \le T\rho \right\}\,.
\end{align}
The offline problem provides an upper bound on the performance of every algorithm, i.e., $$\EE_{\vec \zeta} [R(A | \vgamma, \vec \zeta)] \le \OPT(\vec \gamma)\,.$$ To see this, set $\hat f_t = \mathbb E_{\vec \zeta}[ f_t(x_t, \zeta_t) ]$ and $\hat b_t = \mathbb E_{\vec \zeta}[ b_t(x_t, \zeta_t) ]$ and note that $(\hat f_t, \hat b_t) \in \hat \Omega_t$. Here we used that $\mathbb E_{\vec \zeta}[ f_t(x_t, \zeta_t) ] = \mathbb E_{\vec \zeta_{1:t-1}}[ \mathbb E_{\zeta_t}[ f_t(x_t, \zeta_t) ]] = \mathbb E_{\zeta_{1:t-1}}[ f_t(x_t) ]$ by the tower rule and because $x_t$ is independent of $\zeta_s$ for $s \ge t$. We remark that the action $x_t$ could be potentially dependent on the past realizations $\zeta_{1:t-1}$. Every feasible algorithm should satisfy $\sum_{t=1}^T b_t(x_t, \zeta_t) \le T \rho$. Taking expectations, we obtain that $(\hat f_t, \hat b_t)_{t=1}^T$ is feasible for \eqref{eq:OPT-random} and attains an objective value of $\sum_{t=1}^T \hat f_t$. The claim follows.

The same duality bound in terms of the expected reward and resource consumption functions applies. To see this, dualize the resource constraint with a Lagrange multiplier $\mu \ge 0$ to obtain
\begin{align*}
 \OPT(\vgamma)
&\le   \sum_{t=1}^T \sup_{(\hat f_t, \hat b_t) \in \hat \Omega_t} \left\{ \hat f_t - \mu^\top \hat b_t \right\} + T \rho^\top \mu\\
&=  \sum_{t=1}^T \sup_{x_t \in \cX_t} \left\{ f_t(x_t) - \mu^\top b_t(x_t) \right\} + T \rho^\top \mu\\
&=   \sum_{t=1}^T f_t^*(\mu)  +T \rho^\top \mu 
=  D(\mu | \vgamma ) \ ,
\end{align*}
where the first equality uses that the maximum of a linear function over a convex set is always verified at an extreme point together with the fact that the extreme points of $\hat \Omega_t$ lie in $\Omega_t$, and the last equality utilizes the definition of $f^*_t$ in terms of the expected functions.

\paragraph{Performance Bounds.} Our theorems still hold after taking the expectation over $\zeta$. In the sequel, we prove the result for the case of stochastic i.i.d.~input.

The proof essentially follows exactly from the proof of Theorem~\ref{thm:master} after taking the expectation on $\zeta$. The only difference is that the stopping time $\tA$ is now defined in terms of the realized consumption $b_t(x_t, \zeta_t)$ instead of the expected consumption $b_t(x_t)$. In particular, we define the stopping time $\tA$ as the first time less than $T$ that there exists resource $j$ such that 
\begin{equation*}
\sum_{t=1}^{\tA} (b_t(x_t, \zeta_t))_j + \ubbinfty \ge \rho_j T \ . 
\end{equation*}
As before, we are guaranteed not to violate the resource constraints before the stopping time $\tA$.

In the last step, where we used the complementary slackness term to bound the stopping time, we now have that the functions
\[
    w_t(\mu) = \mu^\top(\rho - b_t(x_t))
\]
are given in terms of the expected resource depletion. Using a martingale argument, we can show that
\begin{equation}\label{eq:martingal}
    \mathbb E\left[ \sum_{t=1}^{\tA} w_t(\mu) \right] 
    = 
    \mathbb E\left[ \sum_{t=1}^{\tA} \mu^\top(\rho - b_t(x_t))  \right]
    = 
    \mathbb E\left[ \sum_{t=1}^{\tA} \mu^\top(\rho - b_t(x_t, \zeta_t))  \right]\,,
\end{equation}
and the rest of the proof follows. To prove \eqref{eq:martingal}, consider the sigma algebra $\mathcal F_t = \sigma(\gamma_{1:t+1}, \zeta_{1:t})$. Note that $b_{t+1}(x_{t+1}) \in \mathcal F_t$ but $b_{t+1}( x_{t+1}, \zeta_{t+1}) \in \mathcal F_{t+1}$. At first, it holds that $M_t:=\sum_{s=1}^t b_s(x_s, \zeta_s) - b_s(x_s)$ is a martingale with respect to $\mathcal F_t$, because $M_t \in \mathcal F_t$, $\EE_\zeta(\|M_t\primalnorm)\le 2\ubb t$ for any $t$, and
\begin{equation*}
    \EE \mpran{M_{t+1}-M_{t}|\mathcal F_t}= \EE_{\zeta_{t+1}}[b_{t+1}(x_{t+1}, \zeta_{t+1})] - b_{t+1}(x_{t+1}) = 0 \ .
\end{equation*}
Moreover, $\tA$ is measurable with respect to $\mathcal F_t$ and, thus, $\tA$ is a stopping time with respect to $\mathcal F_t$. It then follows by Martingale Optional Stopping Theorem that $\EE_{\zeta} [M_{\tA}] = \EE_{\zeta}[M_1]=0$. The claim follows.

\section{Online Mirror Descent}

We reproduce a standard result on online mirror descent for completeness.

\begin{prop}\label{prop:omd}
Consider the sequence of convex functions $w_t:\mathcal D \rightarrow \mathbb R$ with $\mathcal{D} \subseteq \RR^m$ a convex set. Let $g_t \in \partial_\mu w_t(\mu_t)$ be a sub-gradient and
   \begin{equation*}
       \mu_{t+1} = \arg\min_{\mu\in\mathcal{D}} g_t^\top \mu + \frac{1}{\eta} V_h(\mu, \mu_t)\,.
   \end{equation*}
Suppose sub-gradients are bounded by $\| g_t \primalnorm \le G$ and the reference function $h$ is $\sigma$-strongly convex with respect to $\|\cdot\dualnorm$-norm. Then, for every $\mu \in \mathcal D$ we have
\begin{align}\label{eq:mirror-descent-bound}
    \sum_{t=1}^{T} w_t(\mu_t) - w_t(\mu) \le \frac{G^2 \eta}{2\sigma}  T + \frac 1 \eta V_h(\mu,\mu_1)\,.
\end{align}
\end{prop}

\begin{proof}
Because the Bregman divergence $V_h$ is differentiable and convex in its first argument, and $\mathcal D$ is convex, the first order conditions for the Bregman projection (see, e.g., Proposition 2.1.2 in \citealt{Bertsekas}) are given by
\begin{align}\label{eq:foc-bregman}
    \left( g_t + \frac 1 \eta \left( \nabla h(\mu_{t+1}) - \nabla h(\mu_t)\right)\right)^\top \left( \mu - \mu_{t+1} \right) \ge 0\,, \quad \forall \mu \in \mathcal D\,.
\end{align}
Therefore, it holds for any $\mu \in \mathcal D$ that
\begin{equation}  \label{eq:ineq_chain1}
\begin{split}
     \langle g_t, \mu_t - \mu \rangle\
    &=\langle g_t, \mu_t - \mu_{t+1} \rangle\ + \langle g_t, \mu_{t+1} - \mu \rangle\ \\
    &\le\langle g_t, \mu_t - \mu_{t+1} \rangle\ + \frac 1 \eta \left( \nabla h(\mu_{t+1}) - \nabla h(\mu_t)\right)^\top \left( \mu - \mu_{t+1} \right) \\
    &=\langle g_t, \mu_t - \mu_{t+1}\rangle + \frac{1}{\eta} V_h(\mu,\mu_t) - \frac{1}{\eta} V_h(\mu,\mu_{t+1}) - \frac{1}{\eta} V_h(\mu_{t+1},\mu_t) \\
 &\le  \langle g_t, \mu_t - \mu_{t+1}\rangle + \frac{1}{\eta} V_h(\mu,\mu_t) - \frac{1}{\eta} V_h(\mu,\mu_{t+1}) - \frac{\sigma}{2\eta} \|\mu_{t+1}-\mu_t\dualnorm^2 \\
  &\le   \frac{\eta}{2\sigma}\|g_t\primalnorm^2  + \frac{1}{\eta} V_h(\mu,\mu_t) - \frac{1}{\eta} V_h(\mu,\mu_{t+1}) \\
  &\le \frac{\eta}{2\sigma} G^2  + \frac{1}{\eta} V_h(\mu,\mu_t) -\frac{1}{\eta} V_h(\mu,\mu_{t+1}) \ , \\
\end{split}
\end{equation}
where the first inequality follows from \eqref{eq:foc-bregman}; the second equality follows from Three-Point Property stated in Lemma 3.1 of \citet{chen1993convergence}; the second inequality is by strong convexity of $h$; the third inequality uses that $a^2 + b^2 \ge 2 a b$ for $a,b \in \RR$ and Cauchy-Schwarz to obtain
\begin{align*}
    \frac{\sigma}{2\eta} \|\mu_{t+1}-\mu_t\dualnorm^2+\frac{\eta}{2\sigma}\|g_t\primalnorm^2 &\ge \|\mu_{t+1}-\mu_t\dualnorm\|g_t\primalnorm\\ 
    &\ge   |\langle g_t, \mu_t - \mu_{t+1}\rangle| \ ,
\end{align*}
and the last inequality follows from the bound on gradients. Therefore, by convexity of $w_t(\cdot)$, we obtain that
\begin{align}\label{eq:part1}
     \sum_{t=1}^{T} w_t(\mu_t) - w_t(\mu) \le \sum_{t=1}^{T} \langle g_t, \mu_t - \mu \rangle
     \le  \frac{G^2 \eta}{2\sigma}  T + \frac 1 \eta V_h(\mu,\mu_1) \ .
\end{align}
where the inequality follows from summing up  \eqref{eq:ineq_chain1} from $t=1$ to $t=T$, telescoping, and using that the Bregman divergence is non-negative.
\end{proof}

\vspace{0.2cm}

\section{Numerical Experiments}\label{sec:numerical}

{\color{black}
\subsection{Online Linear Programming}\label{sec:num-olp}
In this section, we present numerical experiments on online linear programming (Section \ref{sec:nrm}) under synthetic stochastic i.i.d.~inputs. The goal of this section is to verify the dependence of regret for Algorithm \ref{al:sg} on time horizon $T$, resource dimension $m$ and primal decision dimension $d$ as stated in Theorem \ref{thm:master} with different reference functions.

\textbf{Data generation:} We consider an online linear program, as specified in Section \ref{sec:nrm}. At time $t=1,\ldots,T$, the firm receives a consumption matrix $c_t\in\RR^{m \times d}$ and a reward vector $r_t\in\RR^d$, and needs to make a real-time decision $x_t\in \cX_t=\Delta_d$, where $\Delta_d$ is the standard simplex in $\RR^d$, i.e., $\Delta_d=\{x\in \RR^d: \|x\|_1=1, x\ge 0\}$. In our experiment, the entries in the $i$-th row of the matrix $c_t$ are generated from a Bernoulli distribution with probability parameter $p_i$, for $i=1,\ldots,m$. Then the reward vector is generated as $r_t=\text{Proj}_{[0,\bar r]}\left\{\theta^\top c_t + \delta_t \mathbf 1\right\}$, where $\theta \in \RR^m$ is an unknown parameter that connects consumption $c_t$ with reward $r_t$ and $\delta_t \in \RR$ is an i.i.d.~Gaussian noise. The projection step makes sure that each entry of the reward vector is always larger than $0$ and smaller than $\bar{r}$. The resource vector $\rho$ is set to be $\beta * p$, where $\beta\in\RR^m$ and each entry is generated from a uniform distribution $\text{U}(0.25, 0.75)$. The values of $\beta$ control the consumption budget-ratio for each resource.

In our experiments, we set $p_i=(1+\alpha)/2$ where $\alpha\sim \text{Beta}(1,3)$ comes from a beta distribution. This guarantees $\rho=\beta * p \in [0.125, 0.75]^m$ are generated from a rescaled beta distribution with $\ubrho=0.75$ and $\lbrho=0.125$. Furthermore, setting $\bar{r} = 10$ guarantees that $\ubf=10$ is a valid upper bound on the reward function, and $\ubb=\max_{x\in\Delta_d}\|c_t x\|_{\infty}=1$ provides a valid upper bound on the consumption. We generate $\theta$ from a standard multi-variate Gaussian distribution $N(0,\text{diag}(1))$ and then it is normalized to satisfy $\|\theta\|_2=1$. The noise $\delta_t$ is generated i.i.d.~from Gaussian distribution $N(0,1)$ for each time $t$.

\textbf{Random trials:} There are two layers of randomness in the experiments: randomness coming from generating the parameters of the model (namely $\alpha, \beta, p, \theta$), and randomness coming from generating the reward vector $r_t$ and consumption matrix $c_t$ from the model with given parameters. In the numerical experiments, we first obtain $10$ random sets of the parameters (for the first layer of randomness), and for each set of parameters, we run our algorithm (Algorithm \ref{al:sg}) $10$ times (for the second layer of randomness). In total, we have $100$ random trials and report the average regret as well as its confidence interval.

\textbf{Regret computation:} For each random trial, we compute the cumulative reward obtained by Algorithm \ref{al:sg}. We then compute the average cumulative reward over the $100$ trials as our expected reward of Algorithm \ref{al:sg}, i.e., $\EE_{\vgamma\in\vcP} \left[ R(A|\vgamma) \right]$. Computing $\OPT(\vgamma)$ exactly is computationally prohibiting for the scale of our experiments as the offline problem can be too large to be stored in memory. We instead utilize $\EE_{\vgamma\in\vcP} \left[ D(\bmu_T|\vgamma)\right]$ as an upper bound of $\OPT$, where $\bmu_T = \frac{1}{T} \sum_{t=1}^T \mu_t$ is the average of the dual variables produced by our algorithm. We compute (an upper bound of) the regret as $ \EE_{\vgamma \sim \vcP} \left[D(\bmu_T|\vgamma) - R(A|\vgamma) \right]$.

\textbf{Results:} We here examine the dependence of the regret for Algorithm \ref{al:sg} over $T, m, d$  with three dual updates: online gradient descent (OGD) as defined in \eqref{eq:ogd}, multiplicative weights update (MWU) as defined \eqref{eq:mwu}, and multiplicative weights update with projection (MWU-P) as explained in Appendix \ref{sec:mwu-p}. The step-sizes $\eta$ in the experiments are set as follows
\begin{itemize}
    \item Online gradient descent ({OGD}): $\eta=s/\sqrt{Tm}$,
    \item Multiplicative weights update ({MWU}): $\eta=s/\sqrt{T}$, 
    \item Multiplicative weights update with projection ({MWU-P}): $\eta=s/\sqrt{T}$ ,
\end{itemize}
where $s$ is tuned from the set $\{0.1, 1, 10, 100\}$ to achieve the smallest regret.

\begin{figure}
\centering
  \begin{subfigure}[b]{0.49\textwidth}
    \includegraphics[width=\textwidth]{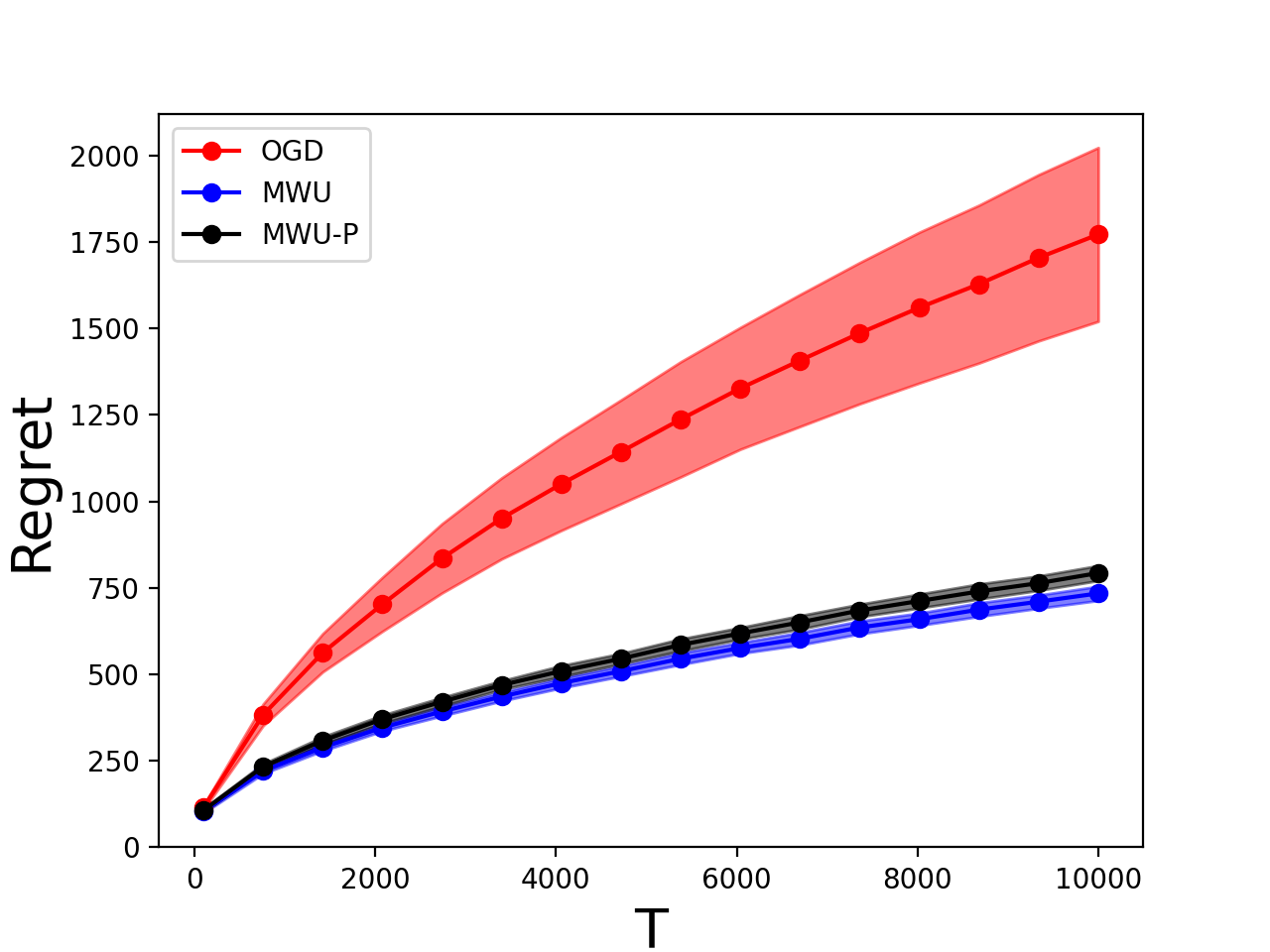}
    \caption{regret versus $T$, with $m=100$ and $d=10$}
  \end{subfigure}
  \hfill
  \begin{subfigure}[b]{0.49\textwidth}
    \includegraphics[width=\textwidth]{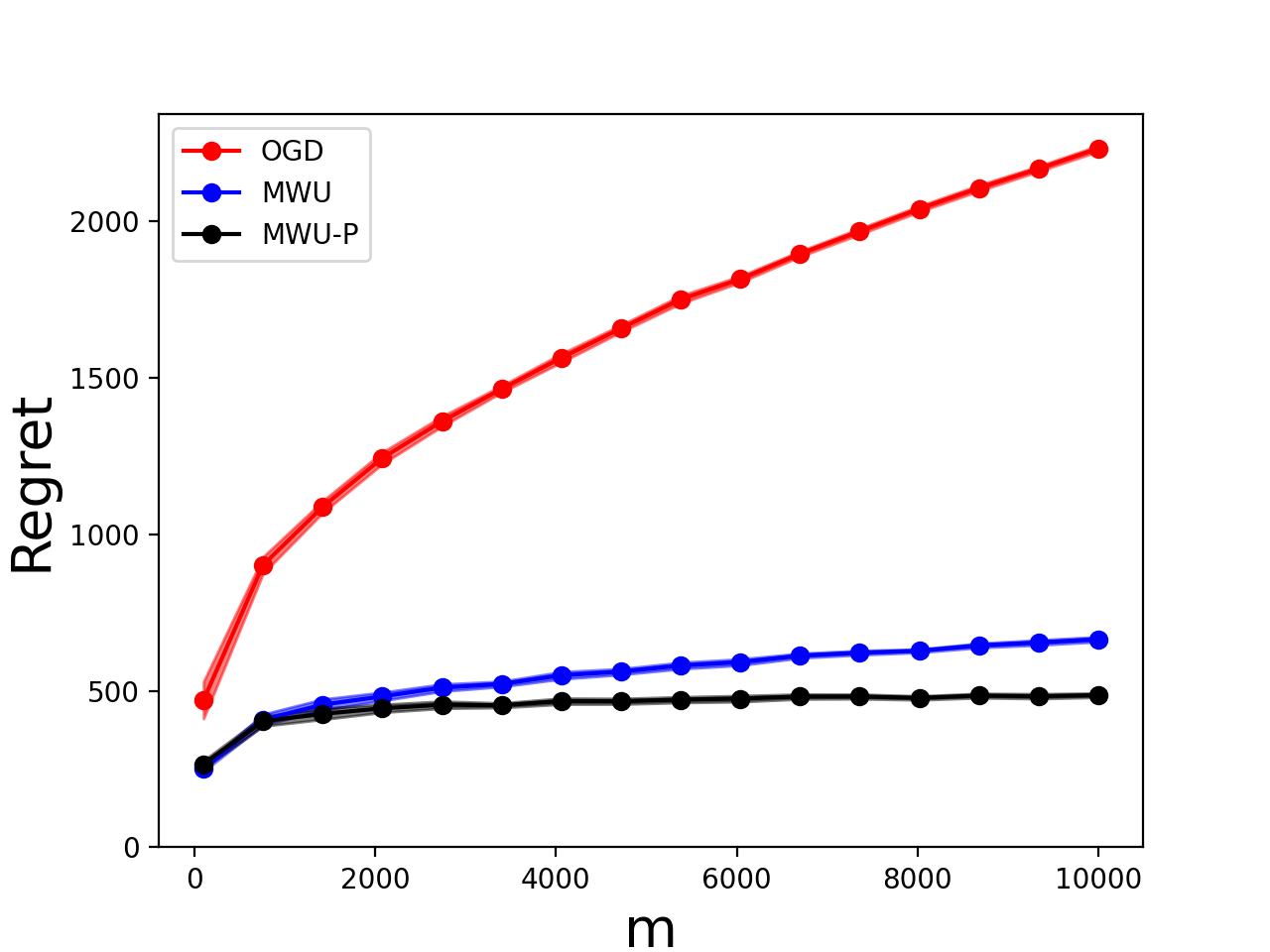}
    \caption{regret versus $m$, with $T=1000$ and $d=10$}
  \end{subfigure}
  \hfill
  \begin{subfigure}[b]{0.49\textwidth}
    \includegraphics[width=\textwidth]{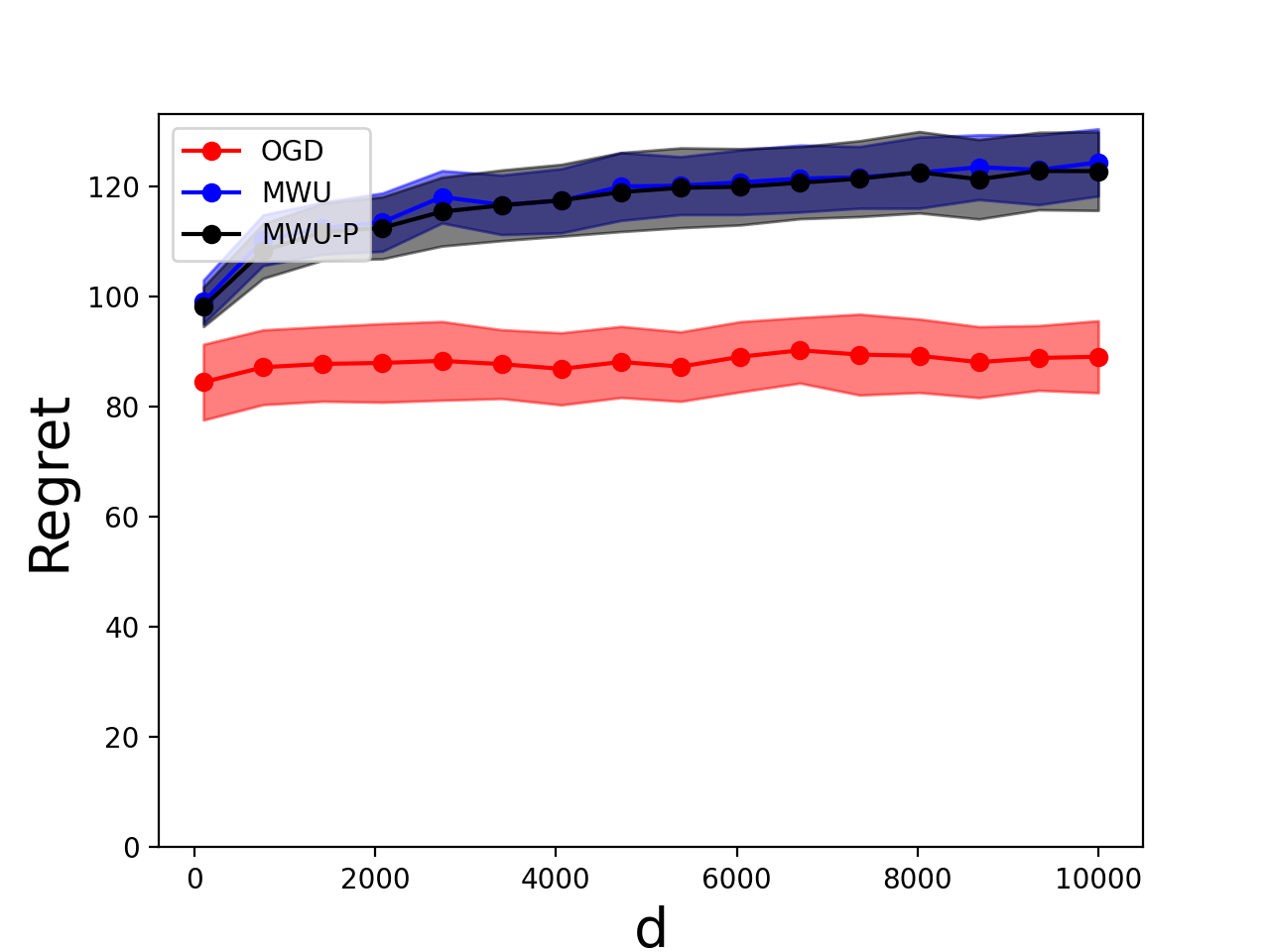}
    \caption{regret versus $d$, with $T=1000$ and $m=100$}
  \end{subfigure}
  \caption{Plots showing the regret and its $95\%$ confidence interval (in shadow) obtained by OGD, MWU and MWU-P versus time horizon $T$, resource dimension $m$, and decision dimension $d$, respectively. \label{fig:online-lp}}
\end{figure}

Figure \ref{fig:online-lp} plots the regret as well as its $95\%$ confidence interval (in shadow) obtained by OGD, MWU and MWU-P versus time horizon $T$, resource dimension $m$, and decision dimension $d$. Each dot is the average regret for $100$ random trials of the corresponding algorithm. The follow observations are in order.
\begin{enumerate}[(a)]
    \item \emph{Length of horizon.} For all three algorithms, the regret has $O(T^{1/2})$ growth, which is consistent with our theory.
    \item \emph{Number of resources.} The regret of MWU-P has slower growth compared with MWU and OGD in $m$, which is consistent with our theory that the regret of MWU-P has $O(\log(m))$ growth, while the regrets of OGD and MWU have $O(m^{1/2})$ growth.
    \item \emph{Primal decision dimension.} For all three algorithms, the regret does not grow with $d$, which is consistent with our theory.
\end{enumerate}

}


\subsection{Proportional Matching}\label{sec:num-pm}
In this section, we present numerical experiments on proportional matching with high entropy (Section \ref{sec:matching}) under i.i.d.~input and ergodic input.

\textbf{Data generation:} For stochastic i.i.d.~input, we use the dataset introduced by  \citet{balseiro2014yield}. They consider the problem faced by a publisher who has to deliver impressions to advertisers so as to maximize click-through rates. (They consider the secondary objective of maximizing revenue from a spot market, which we do not take into account in this experiments). We incorporate the entropy regularizer $H(x)$ to the objective with parameter $\lambda = 0.0002$, which was tuned to balance the diversity and efficiency of the allocation. In each problem instance, there are $m$ advertisers; advertiser $j$ can be assigned at most $\rho_j T$ impressions. The reward vector $r_t$ gives the expected click-through rate of assigning the impression to each advertiser. In their paper, they parametrically estimate click-through rates using mixtures of log-normal distributions. Because they do not report the actual data used to estimate their model, we instead take their estimate model as a generative model and sample impressions from the distributions provided in their paper. We generated 500,000 samples for each publisher, and we present the results for publisher 2 and publisher 5 from their dataset.

{
To test the performance of our algorithm under ergodic input, we perturb the dataset of \citet{balseiro2014yield} to introduce autocorrelation in the click-through rates while maintaining the same marginal distributions for each time period. In our ergodic dataset, the click-through rates follow an AR(1) process, i.e., an autoregressive process with a one-period lag. Namely, we set $\log(r_t) = c \log(r_{t-1}) + \epsilon_t$ where $c \in [0,1]$ captures the amount of autocorrelation and $\epsilon_t$ are multivariate i.i.d.~random variables with mean and covariance matrix chosen so that $r_t$ is distributed as in the original dataset. In particular, setting $c=0$ recovers the case of i.i.d.~input, while setting $c=1$ leads to rewards being constant over time. In our experiments, we set $c \in \{0, 0.5, 0.9\}$.}

\textbf{Random trials:} There are two layers of randomness in Algorithm~\ref{al:sg} when resource consumption is stochastic: randomness coming from the data (i.e., $\cP$), and randomness coming from the proportional matching (i.e., $\vec\zeta$). In the numerical experiments, we first obtain $50$ random datasets with size $T$ (for the first layer of randomness), and for each dataset, we run our algorithm $50$ times (for the second layer of randomness).

\textbf{Regret and relative reward computation:} For each random trial with given round $T$, we compute the cumulative reward obtained by Algorithm~\ref{al:sg}. We then compute the average cumulative reward over the 2,500 trials as our expected reward of Algorithm~\ref{al:sg}, i.e., $\EE_{\vgamma\in\vcP} \left[ R(A|\vgamma) \right]$. {Computing $\OPT(\vgamma)$ exactly can be expensive, as this is a large convex optimization problem. We instead utilize $\EE_{\vgamma\in\vcP} \left[ D(\bmu_T|\vgamma)\right]$ as an upper bound of $\OPT$, where $\bmu_T = \frac{1}{T} \sum_{t=1}^T \mu_t$ is the average of the dual variables produced by our algorithm. We compute (an upper bound of) the regret as $ \EE_{\vgamma \sim \vcP} \left[D(\bmu_T|\vgamma) - R(A|\vgamma) \right]$, and (an lower bound of) the relative reward as $\EE_{\vgamma \sim \vcP} \left[ R(A|\vgamma) \right] / \EE_{\vgamma\in\vcP} \left[D(\bmu_T|\vgamma)\right]$. To showcase the robustness of our algorithm, we use the same step-size for both input models.}

\textbf{Results:}
We report in Figure \ref{fig:regret} the regret as well as its $95\%$ confidence intervals for publisher 2 and publisher 5 with both i.i.d.~input (i.e. $c=0$) and ergodic input (with different correlation level $c$). Each dot plots the average regret for 2,500 random trials. In these experiments, we utilize step-size $\eta=T^{-1/2}$ without much tuning. We can clearly see that the regret for both stochastic i.i.d.~input and ergodic input have $T^{1/2}$ growth, verifying our theory. As expected, the performance when requests have autocorrelation get worse.

Similarly, Figure \ref{fig:relative_reward} plots the relative rewards for both publishers with i.i.d.~input and ergodic input. As $T$ increases, the relative rewards increases, and eventually should converge to one for all correlation levels since the regret is $\tilde O(T^{1/2})$ and the rewards collected by our algorithm grow at a rate $\Omega(T)$. The relative reward goes above $80\%$ within 10,000 online samples for i.i.d.~input ($c=0$) and ergodic input with small correlation ($c=0.5$), which showcases the effectiveness of our proposed algorithm.

Figure \ref{fig:step_size} plots the regret and its $95\%$ confidence intervals as a function of $T$ with different step-size $\eta=s \cdot T^{-1/2}$ for the i.i.d.~input with $s \in \{0.1, 1, 10\}$. We see that the regret has $T^{1/2}$ grow for all step-size levels $s$. Furthermore, as we can see from Figure \ref{fig:step_size}, the regret is less sensitive to the step-size for publisher~2, while it is more sensitive for publisher~5. This suggests that further performance improvements can be obtained by properly tuning the step-size.

\begin{figure}
\centering
  \begin{subfigure}[b]{0.49\textwidth}
    \includegraphics[width=\textwidth]{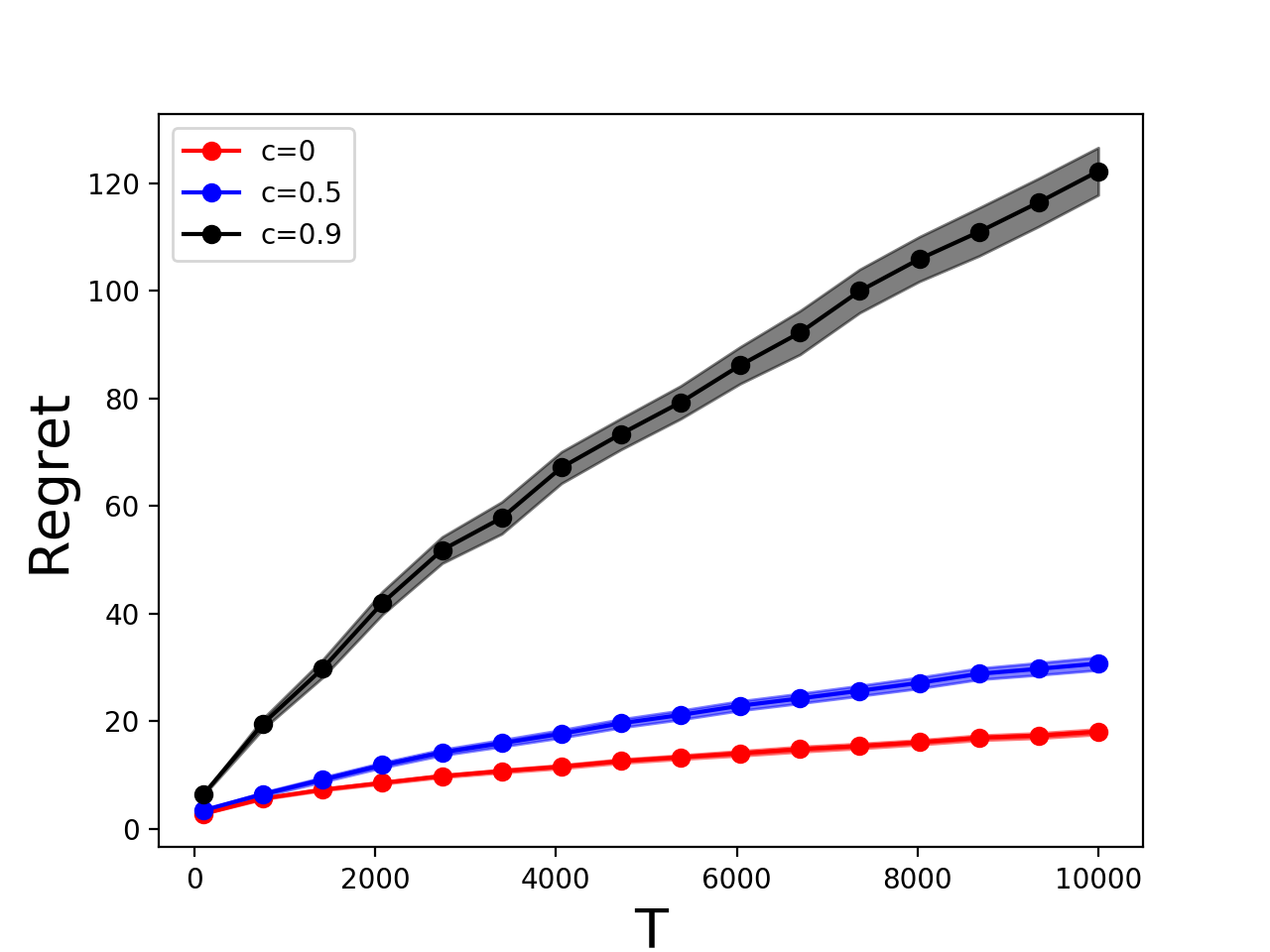}
  \end{subfigure}
  \begin{subfigure}[b]{0.49\textwidth}
    \includegraphics[width=\textwidth]{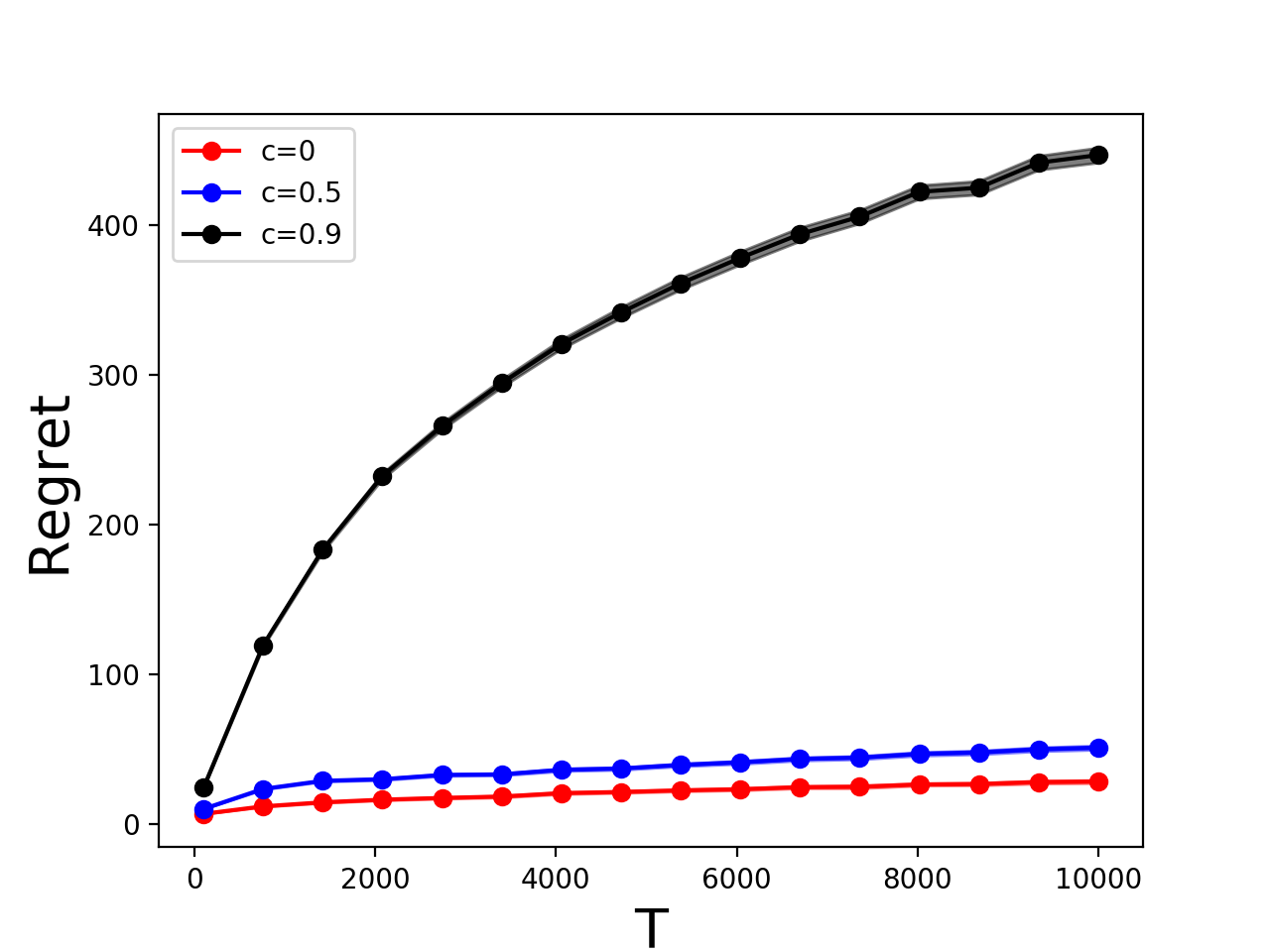}
  \end{subfigure}
  \caption{Regret versus horizon $T$ for i.i.d.~(c=0) and ergodic input for publisher 2 (left) and publisher 5 (right).\label{fig:regret}}

\centering
  \begin{subfigure}[b]{0.49\textwidth}
    \includegraphics[width=\textwidth]{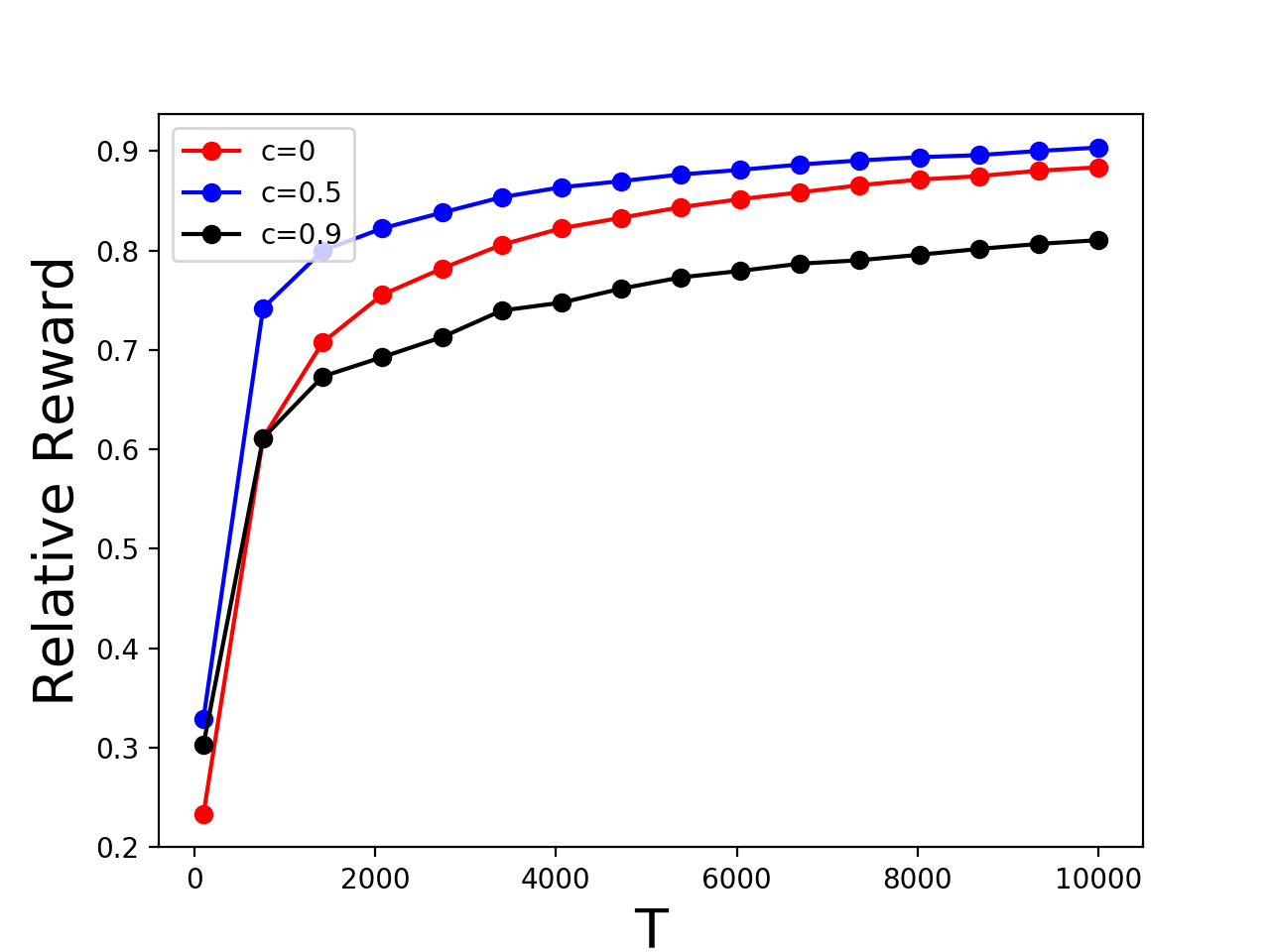}
  \end{subfigure}
  \begin{subfigure}[b]{0.49\textwidth}
    \includegraphics[width=\textwidth]{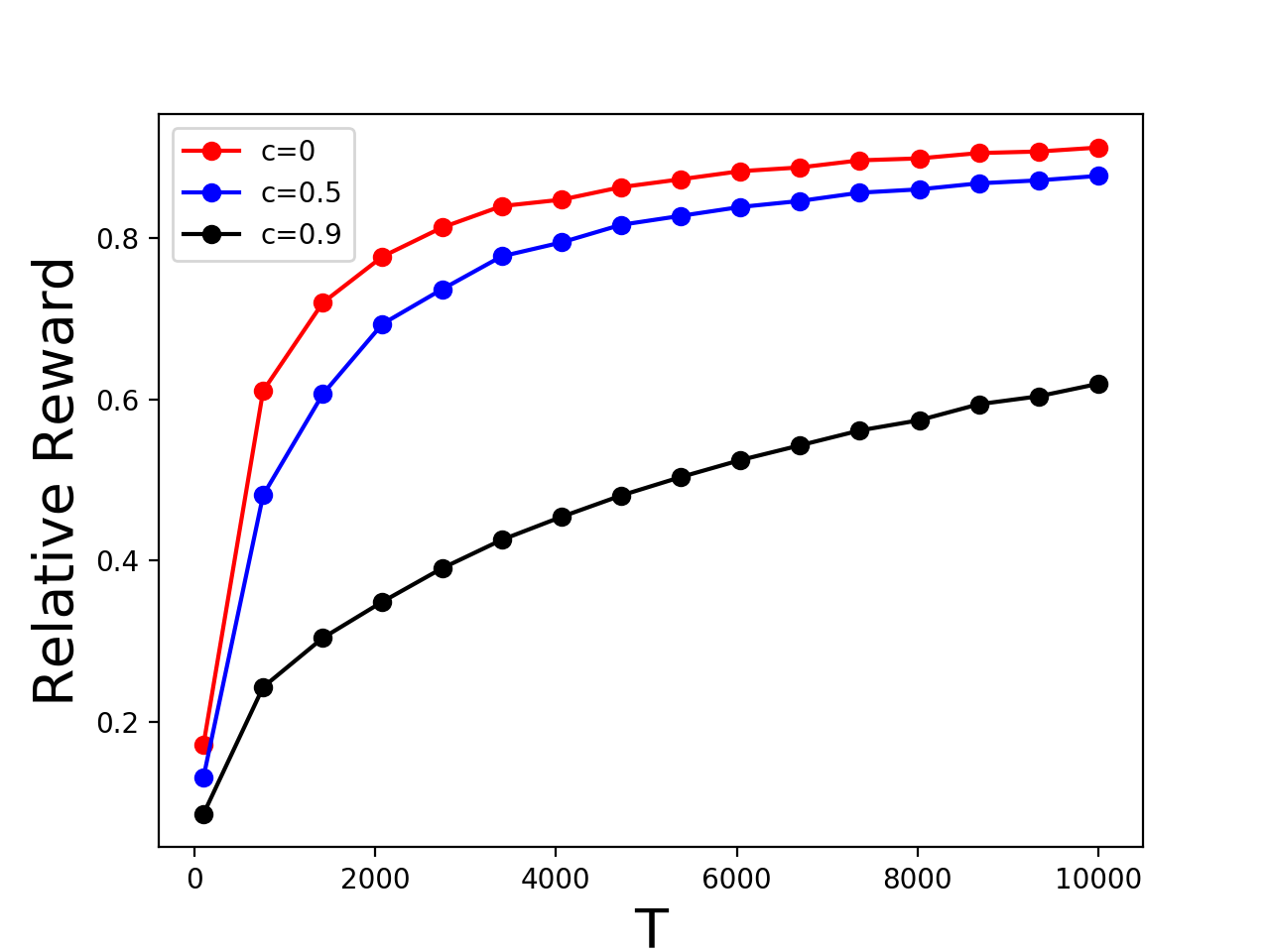}
  \end{subfigure}
  \caption{Relative reward versus horizon $T$ for i.i.d.~(c=0) and ergodic input for publisher 2 (left) and publisher 5 (right).\label{fig:relative_reward}}

\centering
  \begin{subfigure}[b]{0.49\textwidth}
    \includegraphics[width=\textwidth]{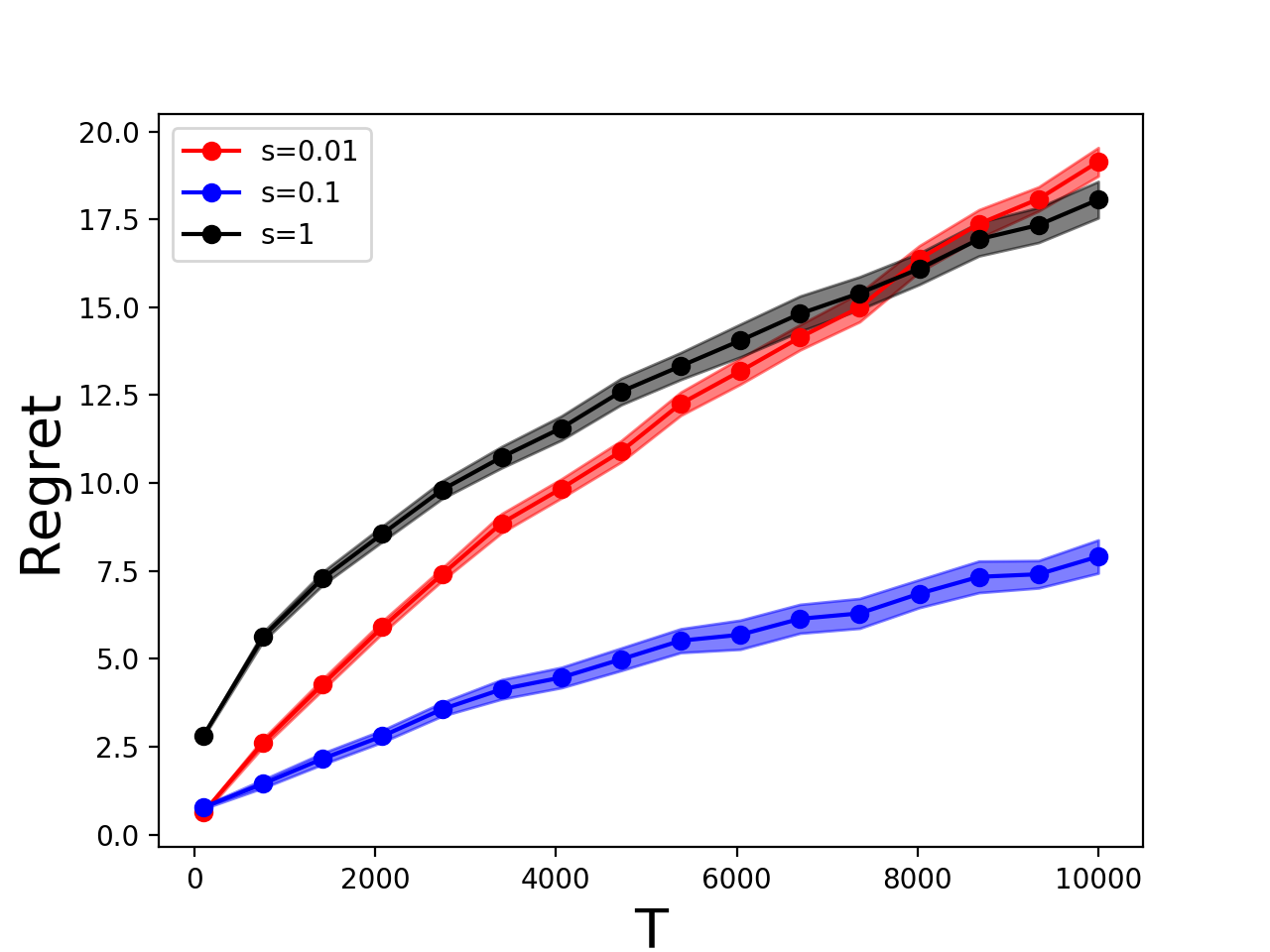}
  \end{subfigure}
  \begin{subfigure}[b]{0.49\textwidth}
    \includegraphics[width=\textwidth]{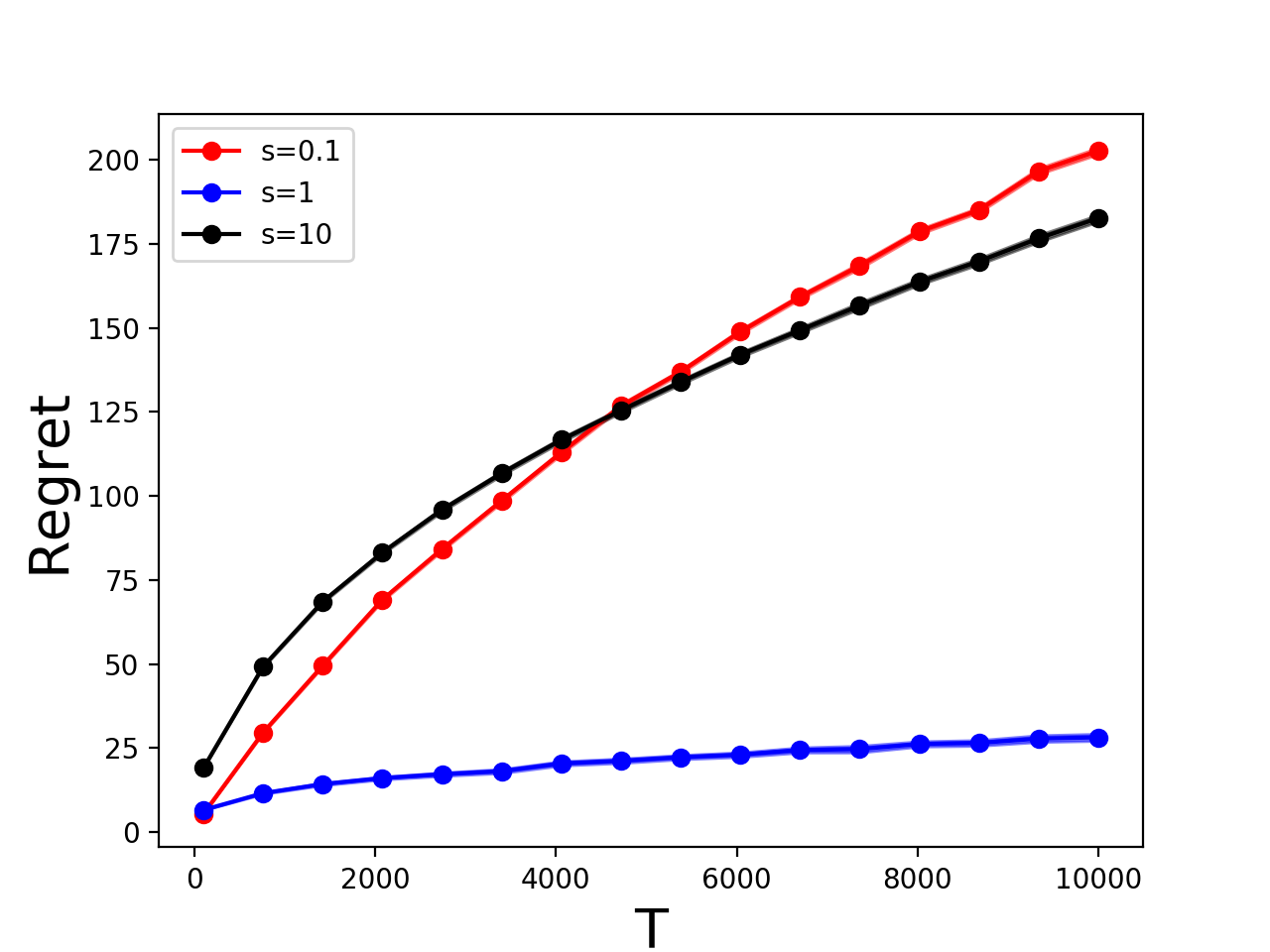}
  \end{subfigure}
  \caption{Regret versus horizon $T$ with different step-size value $\eta=s \cdot T^{-1/2}$ and i.i.d.~input for publisher 2 (left) and publisher 5 (right).\label{fig:step_size}}
\end{figure}




\vspace{0.2cm}

\end{document}